%% file: main.tex
\documentclass[orivec]{lmcs}
\pdfoutput=1

\usepackage{lastpage}
\renewcommand{\dOi}{14(4:2)2018}
\lmcsheading{}{1--\pageref{LastPage}}{}{}%
{Jul.~07,~2017}{Oct.~23,~2018}{}

\usepackage{ifpdf}

\usepackage{cite}
\usepackage{graphicx}

\usepackage{MnSymbol}
\usepackage{amsmath}
\usepackage{amsfonts}
\usepackage{mathrsfs}
\usepackage{dsfont}
\usepackage{amsbsy}
\usepackage{stmaryrd}

\usepackage{type1cm}
\usepackage{algorithm}
\usepackage{algorithmic}
\usepackage{xpatch}
\usepackage[tight,footnotesize]{subfigure}
\usepackage{url}
\usepackage{tikz}
\usetikzlibrary{shapes,arrows,shadows,positioning,automata,circuits.logic.US,circuits.logic.IEC}

\xpatchcmd{\algorithmic}{\setcounter}{\algorithmicfont\setcounter}{}{}
\providecommand{\algorithmicfont}{}
\providecommand{\setalgorithmicfont}[1]{\renewcommand{\algorithmicfont}{#1}}

\input{defs}
\renewcommand{\vec}[1]{\mathbf{#1}}
\keywords{Spatio-Temporal Logic, Signal Temporal Logic, Quantitative Semantics, Monitoring Algorithms, Real-Time Verification, Turing Pattern, Bike Sharing, Statistical Model Checking}

\usepackage{thmtools}
\usepackage{thm-restate}

\begin{document}

\setalgorithmicfont{\footnotesize}

\title[Monitoring Spatio-Temporal Properties with SSTL]{Qualitative and Quantitative Monitoring of Spatio-Temporal Properties with SSTL\rsuper*}
\titlecomment{{\lsuper*}Work partially funded by the EU-FET project QUANTICOL (nr. 600708) and the  FRA-UNITS.
 We thank Diego Latella and Ezio Bartocci for the discussions, EB for sharing the code to generate traces of the Turing Pattern model and Cheng Feng for the data of the Bike Sharing system.}

\author[L.~Nenzi]{Laura Nenzi\rsuper{a}}
\author[L.~Bortolussi]{Luca Bortolussi\rsuper{b}}
\author[V.~Ciancia]{Vincenzo Ciancia\rsuper{c}}	
\author[M.~Loreti]{Michele Loreti\rsuper{d}}	
\author[M.~Massink]{Mieke Massink\rsuper{e}}

\address{{\lsuper{a}}Institute of Computer Engineering, University of Technology, Vienna, Austria}
\address{{\lsuper{b}}DMG, University of Trieste, Trieste, Italy}
\address{\lsuper{c,e}Istituto di Scienza e Tecnologie dell'Informazione ``A. Faedo'' - CNR, Pisa, Italy}	
\address{\lsuper{d}Scuola di Scienze e Tecnologie, University of Camerino, Camerino, Italy}	


%

\begin{abstract}


In spatially located, large scale systems, time and space dynamics interact and drives the behaviour. Examples of such systems can be found in many smart city applications and Cyber-Physical Systems.
In this paper we present the {\it Signal Spatio-Temporal Logic} (SSTL), a modal logic that can be used to specify spatio-temporal properties of {\it linear time} and {\it discrete space} models. The logic is equipped with a {\it Boolean} and a {\it quantitative semantics} for which efficient monitoring algorithms have been developed. As such, it is suitable for real-time verification of both white box and black box complex systems. These algorithms can also be combined with stochastic model checking routines. SSTL combines the until temporal modality with two spatial modalities, one expressing that something is true {\it somewhere} nearby and  the other capturing the notion of being {\it surrounded} by a region that satisfies a given spatio-temporal property. The monitoring algorithms are implemented in an open source Java tool. We illustrate the use of SSTL analysing  the formation of patterns in a Turing Reaction-Diffusion system and spatio-temporal aspects of a large bike-sharing system. 
\end{abstract}
\maketitle



\input{intro}
\input{background}
\input{SSTL}
\input{algo}
\input{implementation}
\input{stoch}
\input{caseStudy2}
\input{conc}



\bibliographystyle{splncs03}
\bibliography{biblio}

\clearpage

\appendix
\input{appendix}

\end{document}

%% file: defs.tex
\newcommand{\bb}[1]{\mathbb{#1}}
\addtocounter{page}{0}

\renewcommand{\phi}{\varphi}

\newcommand{\ev}[1]{\mathcal{F}_{#1}}
\newcommand{\glob}[1]{\mathcal{G}_{#1}}
\newcommand{\until}[1]{\mathcal{U}_{#1}}
\newcommand{\somewhere}[1]{\diamonddiamond_{#1}}
\newcommand{\everywhere}[1]{\boxbox_{#1}}
\newcommand{\surround}[1]{\mathcal{S}_{#1}}



%% file: intro.tex

\section{Introduction}
\label{sec:intro}

The ongoing digital revolution is concretising itself in scenarios where a large number of computational devices, located in space, interact among each other and with humans in an open and mutating environment. Designing and controlling such systems, and guaranteeing some sort of correctness at run-time, are incredibly challenging problems.

In this paper, we consider spatially located systems, where time and space  dynamics interact and drive the system behaviour. Examples are  Cyber-Physical Systems, like pacemaker devices controlling the rhythm of heart beat, and  Collective Adaptive Systems, such as  bike sharing systems in smart cities or the guidance of crowd movement in emergency situations. 
%
%

Controlling and designing  spatio-temporal behaviours requires appropriate formal tools to describe such properties, and to monitor and verify whether, and how robustly, they are satisfied by a system.  Formal methods play a central role in these tasks, providing both formal languages to specify spatio-temporal models and properties, and algorithms to verify such properties on models and on traces of monitored systems. In this paper we consider large and complex systems for which standard model checking procedures, i.e. those that check in an exhaustive way which states of a system model satisfy a given property, are not feasible.  For these kinds of systems simulation and testing are currently the preferred verification methods. This is the area of monitoring (run-time verification)~\cite{Maler2004, Donze2013}, where an individual simulation trace $\vec x$ of a system is checked against a modal logic formula, using an automatic verification procedure.



\input{relatedworks}

\subsection*{Contributions.} 
In this work we present a formal definition of the {\it Signal Spatio-Temporal Logic} (SSTL). 
SSTL integrates the temporal modalities of STL with two spatial modalities; the {\it somewhere} operator and a novel {\em bounded} version of the topological {\em surrounded} operator. The surrounded operator is inspired by the topological spatial until operator of SLCS~\cite{ciancia2014}, adding useful metric bounds on spatial distances. A key point of our logic is that the spatial properties are not only considered as atomic propositions of a temporal logic as in SpaTeL. This means that the temporal and spatial operators can be arbitrarly nested, permitting the specification of complex spatio-temporal behaviours.
 The logic is provided with a Boolean and a quantitative semantics, both equipped with efficient monitoring algorithms.  The major challenge is to monitor the bounded surrounded operator for the quantitative semantics, for which we propose a novel fixed point algorithm, discussing in detail its correctness and computational cost. 
A prototype tool has been developed in \textsf{Java} and can be downloaded from \url{https://github.com/Quanticol/jsstl}. 

This article is an extension of the conference papers \cite{VALUETOOLS14, RV15}. Specifically, the description of the monitoring algorithms has been enhanced, describing the procedure in detail, including proofs of all statements (in particular, correctness of the monitoring algorithm for the \emph{surrounded} operator).  The Turing Reaction-Diffusion case study of \cite{RV15} has been used as a running example throughout the paper, to illustrate the intuition and formal semantics of the new spatial operators. 
We show how SSTL can deal with stochastic scenarios, considering the effects of external perturbations on the Turing system, adding a white Gaussian noise to the set of equations. To this purpose SSTL monitoring algorithms are combined with statistical model checking techniques, following up on earlier works on such combinations with STL~\cite{TCS2015}.
Furthermore, we illustrate the application of SSTL, and the verification methods we have developed, on a larger case study concerning issues of the spatio-temporal distribution of bikes in a bike-sharing system, modelled as a Continuous Time Markov Chain (CTMC).

\subsection*{Paper structure.} 
The paper is organised as follows: Section~\ref{sec:basics} introduces some background concepts on STL and on discrete topologies. 
Section~\ref{sec:sstl} presents the syntax and the semantics of SSTL and the Turing system running example.
Section~\ref{sec:alg} introduces the monitoring algorithms 
and the implementation details of the model checker. 
Section~\ref{sec:stochAnalysis} describes the extension of the monitoring procedure to stochastic systems and its applications on the running example.
In Section~\ref{sec:bikeSharing}, the logic SSTL is applied to a Bike-Sharing system, while conclusions are drawn in Section~\ref{sec:disc}. Proofs of the main results are provided in Appendix~\ref{sec:appendix}.

%
%
%
%
%


%% file: relatedworks.tex
\subsection*{Related work.} 
Logical specification and monitoring of {\em temporal} properties of systems is a well-developed area. Here we mention Signal Temporal Logic (STL)~\cite{Donze2013, Maler2004}, an extension of Metric Interval Temporal Logic (MITL)~\cite{Alur1996}, describing linear-time properties of real-valued signals. STL has monitoring routines both for its  Boolean and quantitative semantics. The latter provides a measure of the satisfaction degree of a formula~\cite{Maler2004,Donze2010,Donze2013}. In monitoring algorithms verification is performed over single trajectories.

Much work can be found also in the area of spatial logic, mostly to reason about continuous space (so-called \emph{topological} spatial logics, see~\cite{handbookSP}) and focussing on theoretical issues such as expressivity and decidability. Instead, model checking and monitoring procedures have a more recent history. Some logics have been proposed to specify properties of networks of processes   \cite{Reif1985},  for (graph) rewriting theories \cite{rew1}, \cite{rew2}, bigraphs \cite{bigraph} and data structures such as graphs \cite{que} and heaps \cite{heaps1}.  However these logics have often high computational cost, and sometimes they are even undecidable. 

In the present work, we will focus on a notion of {\em discrete} space. The reason is that  many applications, such as bike sharing systems or  meta-population epidemic models~\cite{mari_modelling_2012}, are naturally framed in a discrete spatial structure. Moreover, in many circumstances continuous space is abstracted as a grid or as a mesh. This is the case, for instance, in many numerical methods to simulate the spatio-temporal dynamics of Partial Differential Equations (PDE). Hence, this class of models is naturally dealt with by checking properties on such a discretisation. 

One of the recent spatial logics that has inspired the work in the current paper is the {\it Spatial Logic for Closure Spaces} (SLCS)~\cite{ciancia2014,Ci+16}.  SLCS is a logic proposed for both a discrete and topological notion of space and is based on the theory of (quasi discrete) closure spaces~\cite{Gal99,Gal03}. In SLCS some interesting spatial operators have been introduced such as the (unbounded) surround operator and the near-operator inspired by the closure operator of closure spaces and many derived operators among which well-known operators such as somewhere and everywhere. 

Even more challenging is the combination of spatial and temporal operators~\cite{handbookSP} and few works exist with a practical perspective. Among those, SLCS has been extended with a branching time temporal logic in STLCS~\cite{CGLLM15,BertinoroTutorial16} leading to an implementation of a spatio-temporal model checker. STLCS has been applied in the context of smart public transportation, such as public buses~\cite{CianciaGGLLM15} and bike sharing systems~\cite{CLMP15,CLMPV16}. 
The temporal aspects are introduced as ``snapshot'' models, where each snapshot represents the situation of the discrete space at a particular point in time. 
Differently from the logic considered in this paper, STLCS is only equipped with a Boolean semantics. 
%


SpaTeL~\cite{bartocci2015} is a {\em linear-time} spatio-temporal logic that combines the temporal logic STL with the {\em Tree Spatial Superposition Logic} (TSSL)~\cite{bartocci2014}.  It has been provided with both a Boolean semantics as well as with a quantitative semantics. 
The spatial logic TSSL is suited to specify properties of quad trees. These are spatial data structures that are constructed by recursively partitioning a finite space into uniform quadrants.  
TSSL provides a statistical way to describe the distribution of discrete states in a particular partition of the space. Its spatial logic operators can be used to zoom in and zoom out of particular areas of interest. In this way, very complex spatial structures may be captured, but at the price of a complex formulation of spatial properties.
 The properties can in practice only be learned from some template image using supervised learning algorithms and are not human-readable, even too long to be displayed.
In contrast, the sub-language for spatial properties that we adopt in this paper
is descriptive and topological in nature.

%% file: background.tex

\section{Background material}
\label{sec:basics}

\newcommand{\clop}{\mathcal{C}}

In this section we provide some general background material and notation directly relevant to the results developed in subsequent sections. 
\subsection{Weighted undirected graphs.}
We consider discrete models of space that can be represented as finite undirected graphs. Edges of such graphs are equipped with a positive weight, providing a metric structure to the space in terms of shortest path distances. The weight will often represent the distance between two nodes. This is the case, for instance, when the graph is a discretisation of continuous space. However, the notion of weight is more general, and may be used to encode different kinds of  information. As an example, in a model where nodes represent locations in a city and edges represent streets, the weight could represent the average travelling time, which can be different between two paths with the same physical length but different levels of congestion or different number of traffic lights.

\begin{defi}
A  (positive) {\bf weighted undirected graph} is   a tuple  $G=(L,E,w)$, where:
\begin{itemize}
\item  $L$ is a finite set of locations (nodes), $L \not= \emptyset$
\item $E \subseteq L \times L$ is a symmetric relation, namely the set of connections between nodes (edges),
\item $w :E \rightarrow \mathbb{R}_{> 0}$ is a function that returns the positive cost/weight of each edge.
\end{itemize}
\end{defi}


We will use both $(\ell,\ell')\in E$ and $\ell\;E\;\ell'$ as equivalent notations for an edge in the relation $E$. The space is equipped with a distance metric.
\begin{defi}\label{wdist}
For $\ell, \ell' \in L$, the {\bf weighted distance} is defined as \[d(\ell,
  \ell'):= \min \{  \sum_{e \in \sigma} w(e)\; |\; \sigma \mbox{ is a path
    between } \ell \mbox{ and }  \ell' \}. \]
\end{defi}
This means that the weighted distance is a metric  that returns  the cost of the shortest path, for each pair of nodes of the graph; where the shortest path is the path that minimises the sum of the weights of the edges of the path. 

\begin{rem}
Let $E^{*}$ be the set containing all the pairs of connected locations, i.e. the transitive closure of $E$.
If $L$ is finite, and if we define an order on the locations, $L= \{ \ell_{1}, ..., \ell_{i},... \}$, then the weighted distance can be seen as a matrix $(d)_{(\ell_{i},\ell_{j}) \in E^{*}}$, where $d[i,j]$ is the distance between $\ell_{i}$ and $\ell_{j}$. 
\end{rem}

Furthermore,  we denote by $L^{\ell}_{[ d_{1}, d_{2} ]}$  the set of locations $\ell'$ at a distance between $d_1$ and  $d_2$ from $\ell$, formally
$$L^{\ell}_{[ d_{1}, d_{2} ]}
:= \{\ell' \in L  \:|\:  d_1 \leq d(\ell,\ell') \leq d_2, \mbox{ with } d_1, d_2 \geq 0 \}.$$

\subsection{Closure spaces.}
In this work we focus on finite graphs as an algorithmically tractable representation of space. However, \emph{spatial} logics traditionally use more abstract structures, very often of a topological nature (see \cite{handbookSP} for an exhaustive reference). Indeed, the logic we propose can also be defined in a more abstract setting. We can frame a generalised notion of topology on graphs within the so-called \emph{Cech closure spaces}, a superclass of topological spaces allowing a clear formalisation of the semantics of the spatial surrounded operator on both topological and graph-like structures (see \cite{ciancia2014,BertinoroTutorial16,Ci+16} and the references therein). As an example, we mention the topological notion of  \emph{external boundary} of a set of nodes $A$, instantiated on weighted graphs as the set of nodes directly connected to an element of $A$ but not part of it.

\begin{defi}
Given a subset of locations $A \subseteq L$,  we define the {\it external boundary of A} as:
\[B^{+}(A) := \{ \ell \in L \mid \ell \notin A  \wedge\exists \ell' \in A \mbox{ s.t. }  (\ell',\ell) \in E \}. \]
\end{defi}

\subsection{Signal Temporal Logic. }
{\it Signal Temporal Logic} (STL)~\cite{Donze2013, Maler2004} is a linear dense time-bounded temporal logic that extends {\it Metric Interval Temporal Logic} (MITL)~\cite{Alur1996} with a set of atomic propositions $\{\mu_{1}, ..., \mu_{m}\}$ that specify properties of real valued traces, therefore mapping real valued traces into Boolean values.

Let $\vec x : \mathbb{T} \rightarrow \mathbb{D} $ be a trace that describes an evolution of our system, where  $\mathbb{T}= \mathbb{R}_{\geq 0}$ represents continuous time and $\mathbb{D}= \mathbb{D}_{1} \times \cdots \times \mathbb{D}_{n} \subseteq \mathbb{R}^{n}$ is the domain of evaluation; then
each $\mu_{j}:  \mathbb{D}  \rightarrow \mathbb{B}$ is of the form $\mu_{j}(x_1,\ldots,x_n) \equiv (f_{j}(x_1,\ldots,x_n) \geqslant 0)$, where $f_{j}:  \mathbb{D}  \rightarrow \mathbb{R}$ is a (possibly non-linear) real-valued function and $\mathbb{B}=\{\mathit{true} ,\mathit{false}\}$ is the set of Boolean values. 
The projections $x_{i}: \mathbb{T} \rightarrow \mathbb{D}_{i}$ on the $i^{th}$ coordinate/variable are called the {\it primary signals}
and, for all $j$, the function $s_{j}: \mathbb{T} \rightarrow \mathbb{R}$, defined by point-wise application of $f_{j}$ to the image of $\vec x$, namely  $s_{j}(t):= f_{j}(x_{1}(t), ..., x_{n}(t))$, is called the {\it  secondary signal} \cite{Donze2010}.

The syntax of STL is given by

{\hfill {
$
\varphi :=  \mu  \mid   \neg \varphi  \mid  \varphi_{1} \wedge \varphi_{2} \mid  \varphi_{1} \: \mathcal{U}_{[t_{1},t_{2}]} \: \varphi_{2} 
$,
\hfill
}
}

\noindent
where $\mu$ is an atomic proposition, conjunction and negation are the standard Boolean
connectives, $[t_{1},t_{2}]$ is a real positive dense interval with $t_{1} < t_{2}$ and $ \mathcal{U}_{[t_{1},t_{2}]}$ is the {\it bounded until} operator. 
The latter specifies that formula $\varphi_1$ holds until, at time $t\in [t_{1},t_{2}]$, formula $\varphi_2$ is satisfied.
This operator can be used to define the operators \emph{eventually} and \emph{always}.
The {\it eventually} operator  $\mathcal{F}_{[t_{1},t_{2}]}$ and  the {\it always } operator $\mathcal{G}_{[t_{1},t_{2}]}$ can be defined as usual with $\top$ denoting {\em true}:
 $\mathcal{F}_{[t_{1},t_{2}]}  \varphi := \top \mathcal{U}_{[t_{1},t_{2}]} \varphi$, $\mathcal{G}_{[t_{1},t_{2}]} \varphi := \neg \mathcal{F}_{[t_{1},t_{2}]} \neg \varphi.$
 Formula $\mathcal{F}_{[t_{1},t_{2}]}\varphi$ states that $\varphi$ is eventually satisfied at a time $t\in [t_1,t_2]$, while $\mathcal{G}_{[t_{1},t_{2}]} \varphi$ indicates that at each time $t\in [t_1,t_2]$, $\varphi$ is satisfied. 
 
%


%% file: SSTL.tex

\section{SSTL: Signal Spatio-Temporal Logic}
\label{sec:sstl}

{\it Signal Spatio-Temporal Logic}  (SSTL) is a spatial extension of  STL~\cite{Donze2013, Maler2004} with two spatial modalities: the \emph{bounded somewhere} operator $\somewhere{[d_{1},d_{2}]}$ and the {\it bounded surrounded } operator $\surround{[d_1, d_2]}$. SSTL is interpreted on spatio-temporal, real-valued signals.  In this section, we first introduce the signals and then present the syntax and the Boolean and quantitative semantics of SSTL, using a pattern formation model as running example.

\subsection{Spatio-Temporal Signals.} We define signals with continuous time and discrete space. In particular, the space is represented by a weighted undirected graph $G= (L,E,w)$, defined in Section \ref{sec:basics}, while the time domain $\mathbb{T}$  will usually be the real-valued interval $[0,T]$, for some $T>0$. Formally,
a spatio-temporal signal, is a function
$\vec s: \mathbb{T} \times L \rightarrow \mathbb{D}$, where $L$ is the set of locations and $\mathbb{D}$ is the domain of evaluation. $\mathbb{D}$  is a subset of $\mathbb{R}^{*}= \mathbb{R} \bigcup \{ +\infty, -\infty \}$. Signals with  $\mathbb{D}= \mathbb{B} = \{ 0, 1\}$  are called Boolean signals, whereas those where $\mathbb{D} = \mathbb{R}^{*}$ are called real-valued or quantitative signals.

A spatio-temporal trace is a function
$\vec x: \mathbb{T} \times L \rightarrow  \mathbb{R}^n$ s.t. ${ \vec x (t, \ell)}=(x_{1}(t, \ell), \cdots, x_{n}(t, \ell)) \in \mathbb{D}$, where each $x_{i}:  \mathbb{T} \times L \rightarrow \mathbb{R}$, for $i= 1,...,n$, is the projection on the $i^{th}$ coordinate/variable.
 Note that these projections have the form of quantitative signals. They are called the {\it primary signals} of the trace. We can thus see the trace as a set of primary signals. 
 This means that SSTL can be used to specify spatio-temporal properties of such traces.
Spatio-temporal traces can be obtained  by simulating a stochastic model or by computing the solution of a deterministic system. In previous work~\cite{VALUETOOLS14} the framework of patch-based population models is discussed, which generalise population models and are a natural setting from which both stochastic and deterministic spatio-temporal traces of the considered type emerge. An alternative source of traces are measurements of real systems. For the purpose of this work, it is irrelevant which is the source of traces, as we are interested in their off-line monitoring. 

Spatio-temporal traces are then converted into spatio-temporal Boolean or quantitative signals in the following way: 
similarly to the case of STL, each \emph{atomic predicate} $\mu_j$ is of the form $\mu_j(x_1,\ldots,x_n) \equiv (f_j(x_1,\ldots,x_n) \geq 0)$, for some function $f_j : \mathbb D \to \mathbb R$, where $(x_1,\ldots,x_n)$ are the primary signals. Each atomic proposition gives rise to a spatio-temporal signal. In the Boolean case, one may define function $s_j: \mathbb{T} \times L \rightarrow \mathbb{B}$; given a trace $\vec{x}$, this gives rise to the Boolean signal $s_j(t,\ell) = \mu_j(\vec{x}(t,\ell))$ by point-wise lifting. Similarly, a quantitative signal is obtained as the real-valued function $s_j: \mathbb{T} \times L \rightarrow \mathbb{R}$, with $s_j(t,\ell) = f_j(\vec{x}(t,\ell))$. These signals, derived from the atomic proposition, are called the {\it secondary signals}.

When the space $L$ is finite, as in our case, we can represent a spatio-temporal signal as a finite collection of temporal signals. More specifically, the signal $s(t,\ell)$ can be equivalently represented by the collection $\{s_\ell(t)~|~\ell\in L\}$. We will stick mostly to this second notation in the following, as it simplifies the presentation. 

%
%
%
%
%
\subsection{SSTL Syntax. }
 The syntax of SSTL is given by
\[
\varphi :=  \mu \mid  \neg \varphi \mid  \varphi_{1} \wedge \varphi_{2} \mid  \varphi_{1} \: \until{[t_{1},t_{2}]} \: \varphi_{2}\mid \ \somewhere{[d_{1},d_{2}]} \varphi \mid  \varphi_{1} \: \surround{[ d_{1}, d_{2}]}  \varphi_{2}.
\]
Atomic predicates, Boolean operators, and the time bounded until operator $\until{[t_{1},t_{2}]} $ are those of STL. 
The new spatial operators are the {\it somewhere} operator,  $\somewhere{[d_{1},d_{2}]} $,  and  the {\it bounded surrounded} operator $\surround{[ d_{1}, d_{2}]}$, where $[d_{1},d_{2}]$ is a closed real interval with $d_{1} < d_{2}$. We can derive also the {\it everywhere} operator  as $\everywhere{[d_{1},d_{2}]} \varphi := \neg  \somewhere{[d_{1},d_{2}]} \neg \varphi$.
%

The somewhere and the everywhere operators were inspired by the modal operators of the {\it Multiprocess Network Logic} \cite{Reif1985}; the idea originated from the necessity to describe behaviours at a certain distance from a specific point, e.g., ``from a bike sharing station, in a radius of 100 meters, there are more than 30 bikes'' or ``in all the positions around my location, at a distance less than 1 km, there are no infected individuals''.  
Formally, the spatial somewhere operator  $\somewhere{[d_{1}, d_{2}]} \varphi$ requires $\varphi$ to hold in a location reachable from the current one with a cost 
greater than or equal to $d_{1}$ and less than or equal to $d_{2}$. 
The cost is given by the shortest weighted distance between the locations, i.e. the sum of the weights of the edges of the shortest path (Def.~\ref{wdist}).
We use the word ``cost''  to distinguish it from the classical spatial notion of distance. Indeed, we can have two streets with the same distance but different travel time due to the presence of traffic lights or congestion. 
In Figure~\ref{fig:graph}, representing a regular grid where the distance between adjacent nodes is 1, we provide some examples of spatial properties. In the graph of the figure, the orange point satisfies the property $\boldsymbol{ \: \diamonddiamond_{[3,5]} \: {\color{red!50} pink}}$. Indeed, there exists a point at a distance 4 from the orange point that satisfies the pink property.
The {\it everywhere} operator  $\everywhere{[d_{1},d_{2}]} \varphi := \neg  \somewhere{[d_{1},d_{2}]}  \neg \varphi$ requires $\varphi$ to hold in {\em all} the locations reachable from the current one with a total cost between $d_{1}$ and $d_{2}$. 
In Figure~\ref{fig:graph},  the orange point of the graph satisfies the property $\boldsymbol{ \: \boxbox_{[2,3]} \:
{\color{yellow!90!black}yellow}}$. Indeed, all the points at a distance between 2 and 3 from the orange point satisfy the yellow property.  

%
%
%

The (bounded) surrounded operator $\varphi_{1} \: \surround{[ d_{1}, d_{2}]}  \varphi_{2}$ is inspired by the \emph{surrounded} operator of the {\it Spatial Logic for Closure Spaces} SLCS (see ~\cite{ciancia2014,Ci+16}), and it is a spatial interpretation of the temporal \emph{until} connective.
It expresses the topological notion of being surrounded by a $\varphi_2$-region, while being in a $\varphi_{1}$-region,  with additional metric constraints.
Informally speaking, the intended meaning of $\varphi_{1} \: \surround{[ d_{1}, d_{2}]}  \varphi_{2}$ is that one cannot escape from a $\varphi_{1}$-region without passing from a node that satisfies $\varphi_2$ and, in any case, one has to reach a $\varphi_2$-node at a distance between $d_{1}$ and $d_{2}$. An example drawn from analysis of bike sharing facilities in \emph{smart cities} is the property ``there are no bikes in the station where I am, but all the bike stations directly connected with this one have at least one bike available, and are located at a distance less than 100 meters from here''. 

The surrounded operator makes the logic strictly more expressive. Whenever its distance bounds are trivial (i.e., equal to $[0,\infty)$), it allows us to express \emph{global} properties, depending upon the state of the whole system and not only on the  local neighbourhood of each point.
This version of the surrounded operator with non-trivial metric bounds adds an additional level of expressivity, allowing us to restrict the attention to subregions of the space defined by the distance constraints. In any case, the topological properties captured by the surrounded operator cannot be captured by the simpler somewhere and everywhere operators, which in turn cannot be defined in terms of the surrounded operator, due to the effect of metric bounds.

Formally, the bounded variant of the surrounded formula, i.e. $\varphi_{1} \: \surround{[ d_{1}, d_{2}]}  \varphi_{2}$, is true in a location $\ell$, when $\ell$ belongs to a set of locations $A$ satisfying $\varphi_1$ and at a distance less than $d_{2}$ from $\ell$, the external boundary $B^{+}(A)$ of $A$ must contain only locations satisfying $\varphi_2$. Furthermore, locations in $B^{+}(A)$ must be reached from $\ell$ with a cost between $d_{1}$ and $d_2$. 
 $B^{+}(A)$ is the set of all the locations that do not belong to A but that are directly connected with a location in $A$.
In Figure~\ref{fig:graph},   the green points satisfy  $\boldsymbol{ {\color{green!80!black}green} \: \mathcal{S}_{[0,100]} \:{\color{blue!70!}blue} }$. Indeed, for each green point we can find a region that contains the point, such that all its points are green and all the points connected with an element that belongs to the region are blue and satisfy the metric constraint. Instead, the property $\boldsymbol{ {\color{green!80!black}green} \: \mathcal{S}_{[2,3]} \:{\color{blue!70!}blue} }$  is satisfied only by the dark green point. The reason is that such a dark green point is the only point for which there exists a region (the green region) such that all its elements are at a distance less than 3 from it  and are green; and all the elements of  the external boundary (the blue region) are  at a distance between 2 and 3 from it.

\begin{figure}[H]
\center
\begin{tikzpicture}{}

\foreach \i /\j / \fc / \bc in { 
	1/1/orange!80/orange!50, 
	2/1/black/white,
	3/1/yellow!90!black/yellow!50, 
	4/1/yellow!90!black/yellow!50, 
	5/1/black/white,
	6/1/blue!50/blue!20,
	7/1/blue!50/blue!20,
	8/1/blue!50/blue!20,
	9/1/black/white,
	1/2/black/white,
	2/2/yellow!90!black/yellow!50, 
	3/2/yellow!90!black/yellow!50, 
	4/2/black/white,	
	5/2/blue!50/blue!20,
	6/2/green!60/green!30,
	7/2/green!60/green!30,
	8/2/green!60/green!30,
	9/2/blue!50/blue!20,
	1/3/yellow!90!black/yellow!50, 
	2/3/yellow!90!black/yellow!50, 
	3/3/black/white,
	4/3/black/white,
	5/3/blue!50/blue!20,
	6/3/green!60/green!30,
	7/3/green!60!black/green!80!black,
	8/3/green!60/green!30,
	9/3/blue!50/blue!20,
	1/4/yellow!90!black/yellow!50, 
	2/4/red!50/red!20,
	3/4/black/white,
	4/4/black/white,
	5/4/blue!50/blue!20,
	6/4/green!60/green!30,
	7/4/green!60/green!30,
	8/4/green!60/green!30,
	9/4/blue!50/blue!20,
	1/5/black/white, 
	2/5/black/white, 
	3/5/black/white,
	4/5/black/white,
	5/5/black/white,
	6/5/blue!50/blue!20,
	7/5/blue!50/blue!20,
	8/5/blue!50/blue!20,
	9/5/black/white
	} {
		\node [circle,draw=\fc,fill=\bc,thick,inner sep=0pt,minimum size=4mm] at (\i*0.8,\j*0.8) (node\i\j) {};
}

\foreach \i [evaluate = \j as \ipp using int(\i+1)] in {1,2,3,4,5,6,7,8,9} {
	\foreach \j [evaluate = \j as \jpp using int(\j+1)] in {1,2,3,4,5} {
		\ifnum\j<5
		\draw[<->] (node\i\j) -- (node\i\jpp);
		\fi
		\ifnum\i<9
		\draw[<->] (node\i\j) -- (node\ipp\j);
		\fi
	}
}
\end{tikzpicture}
\caption[Example of spatial properties]{Example of spatial properties. The orange point satisfies $\boldsymbol{ \: \diamonddiamond_{[3,5]} \: {\color{red!50} pink}}$. The orange point satisfies $\boldsymbol{ \: \boxbox_{[2,3]} \:
{\color{yellow!90!black}yellow}}$. All green points satisfy $\boldsymbol{ {\color{green!80!black}green} \: \mathcal{S}_{[0,100]} \:{\color{blue!70!}blue} }$. The dark green point satisfies also $\boldsymbol{ {\color{green!80!black}green}  \: \mathcal{S}_{[2,3]} \:{\color{blue!70!}blue} }$. }
\label{fig:graph}
\end{figure}

\subsection{Running Example: Turing Patterns}
\label{rex:turing}
To illustrate more complex examples of spatial and spatio-temporal properties in SSTL we introduce a model of pattern formation in a reaction-diffusion system which will be used as a running example in the following sections.

Alan Turing conjectured in~\cite{turing_chemical_1952} that pattern formation is a consequence of the coupling of reaction and diffusion phenomena involving different chemical species, and that these can be described by a set of PDE reaction-diffusion equations, one for each species.The natural analogue, systems of agents interacting and moving in continuous space, is however prohibitively expensive to analyse computationally; an approach that is more amenable to analysis is to discretise space into a number of cells which are assumed to be spatially homogeneous, and to replace spatial diffusion with transitions between different cells.

From the point of view of formal verification, the formation of patterns is an inherently spatio-temporal phenomenon, in that the relevant issue is how the spatial organisation of the system changes over time.
 
\subsection*{Pattern formation model}
Our model, similar to that in~\cite{bartocci2014, bartocci2015},  describes the production of skin pigments that generate spots in animal furs.
The reaction-diffusion system is discretised, according to a Finite Difference scheme, as a system of ODEs whose variables are organised in a $K \times K$ rectangular grid.
More precisely, we treat the grid as a weighted undirected graph, where each cell $(i,j) \in L = \{1,\ldots,K\} \times \{1,\ldots,K\}$ 
is a location (node), edges connect each pairs of neighbouring nodes along four directions (so that each node has at most $4$ adjacent nodes), and the weight of each edge is always equal to the spatial length-scale $\delta$ of the system\footnote{For simplicity, here we fix $\delta=1$. Note that using a non-uniform mesh
the weights of the edges of the resulting graph will not be uniform.}.
We consider two species (chemical substances) $A$ and $B$ in a  $K \times K$ regular grid, obtaining the system:
 \begin{equation}
\begin{cases} 
\frac {d x^{A}_{i,j}}{d t} =  R_{1} x^{A}_{i,j}  x^{B}_{i,j} - x^{A}_{i,j}  + R_{2} + D_{1} (\mu^{A}_{i,j} - x^{A}_{i,j}) &  i=1..,K, \mbox{ } j=1,..,K,\\  
\frac {d x^{B}_{i,j}}{d t} =  R_{3} x^{A}_{i,j}  x^{B}_{i,j}  + R_{4}  + D_{2} (\mu^{B}_{i,j} -  x^{B}_{i,j}) & i=1..,K, \mbox{ }  j=1,..,K,\\
\end{cases}
\label{model}
 \end{equation}
where: $x^{A}_{i,j}$ and $ x^{B}_{i,j}$ are the concentrations of the two species in the cell $(i,j)$;
$R_{i}$, $i = 1,...,4$ are the parameters that define the reaction between the two species;
 $D_{1}$ and $D_{2}$ are the diffusion constants; $\mu^{A}_{i,j}$ and $\mu^{B}_{i,j}$ are the inputs for the $(i,j)$ cell, that is
 \begin{equation}
 \mu^{n}_{i,j}= \frac{1}{|\nu_{i,j} |} \sum_{\nu \in \nu_{i,j}} x^{n}_{\nu}   \qquad n \in \{ A, B\},
 \end{equation}
 where $\nu_{i,j}$ is the set of indices of cells adjacent to  $(i,j)$.
The spatio-temporal trace of the system is the function $\vec x = (x^{A}, x^{B}) : [0, T] \times L\rightarrow \mathbb{R}^{K \times K} \times \mathbb{R}^{K \times K}$ where each $x^{A}$ and $x^{B}$ are the projection on the first and second variable, respectively.
 In Figure \ref{sim} the concentration of species A is shown for a number of time points, generated by the numerical integration of System \ref{model}. The initial conditions have been set randomly to concentrations of both species in a range between values 0 and 16. 
 It can be observed that at times $t=20$ and $t=50$ the shape of the pattern appears and then remains stable. Clearly, some regions (in blue) have a low concentration of A surrounded by regions with a high concentration of A. We consider the regions with low concentration of species A as the `spots' in the pattern. The opposite happens for the value of B (high density regions surrounded by low density regions, not shown).

\begin{figure}[tbp]
\centering
\includegraphics[width=1. \textwidth]{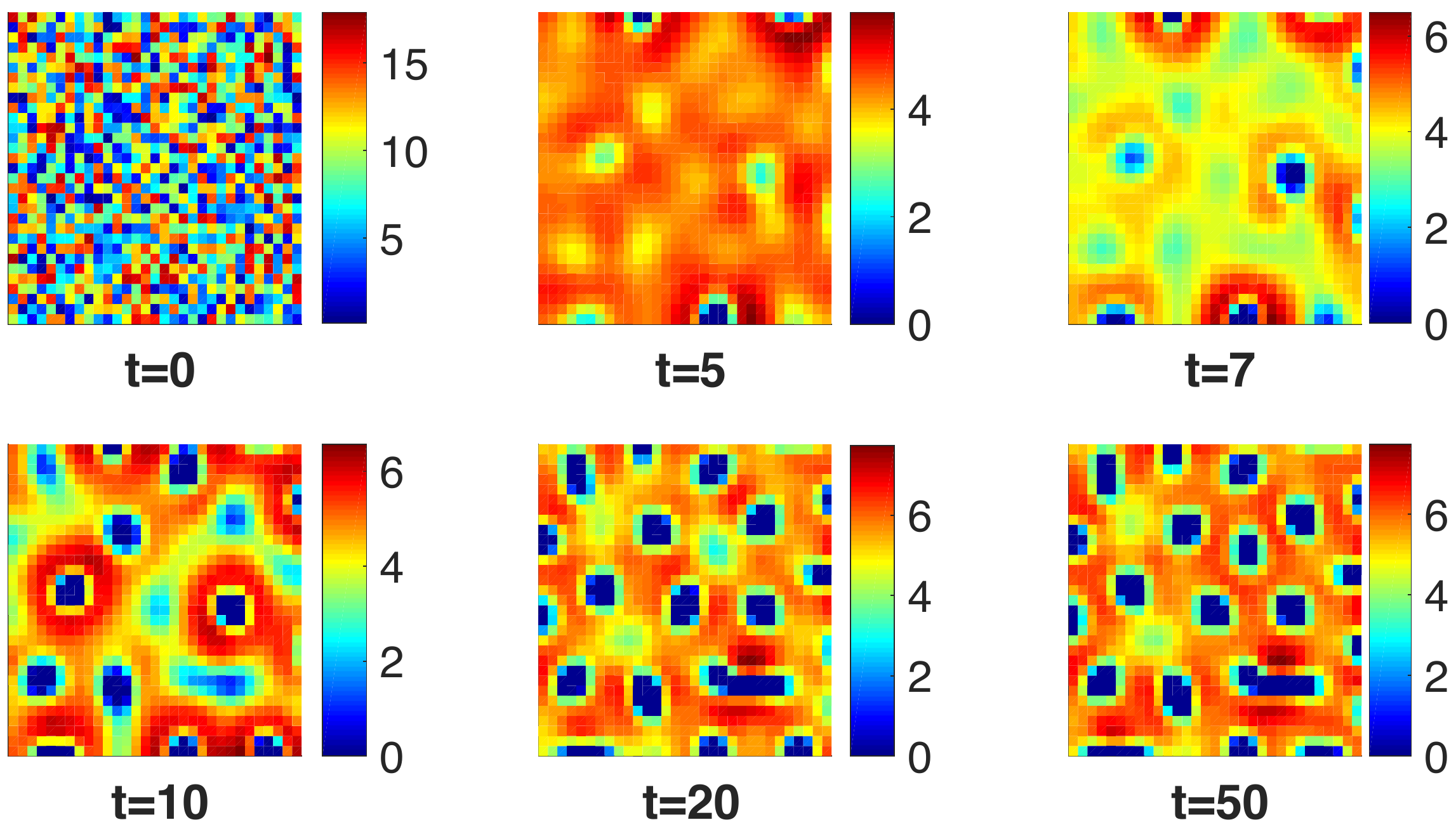}
\caption{Value of $x^{A}$ for the system (\ref{model}) for $t=0, 5, 7, 12, 20,50$ time units with parameters $K=32, R_{1}=1, R_{2}=-12, R_{3}=-1, R_{4}= 16, D_{1} = 5.6$ and $D_{2}= 25.5.$  The initial condition has been set randomly in a range between values 0 and 16. The colour map for the concentration is specified in the legend on the right of each subplot.}
\label{sim}
\end{figure}

%
%

%

\subsection{SSTL Boolean Semantics.}
We first define the Boolean semantics for SSTL.
This semantics, as customary, returns true/false 
depending on whether the observed trace satisfies the SSTL specification.

\begin{defi}[{\bf SSTL Boolean semantics}]
The Boolean satisfaction relation for an SSTL formula $\varphi$ over a  spatio-temporal  trace $\vec x$ is given by:
{\small
\begin{align*}  
 &(\vec x,t, \ell)  \models \mu &\Leftrightarrow&\mbox{ } \phantom{a}   \mu(\vec x(t, \ell))=1
\\ &(\vec x,t, \ell) \models \neg \varphi  &\Leftrightarrow &\mbox{ } \phantom{a}  (\vec x,t, \ell) \not\models \varphi 
\\& (\vec x,t, \ell) \models \varphi_{1} \wedge  \varphi_{2}  &\Leftrightarrow&\mbox{ } \phantom{a}  (\vec x,t, \ell) \models  \varphi_{1}  \wedge  (\vec x,t, \ell) \models  \varphi_{2} 
\\ &(\vec x,t, \ell) \models  \varphi_{1} \: \mathcal{U}_{[t_{1},t_{2}]} \varphi_{2} &\Leftrightarrow &  \mbox{ } \phantom{a}   \exists t' \in t + [t_{1}, t_{2}]  
 :(\vec x, t', \ell) \models \ \varphi_{2} \wedge
\forall t'' \in [t, t'], (\vec x,t'', \ell) \models \varphi_{1} 
\\ &(\vec x,t, \ell) \models  \somewhere{[d_{1},d_{2}]} \varphi &\Leftrightarrow &  \mbox{ } \phantom{a}    \exists \ell' \in{L} : (\ell',\ell)\in E^{*} \wedge d_{1} \leqslant d(\ell',\ell) \leqslant d_{2} \wedge (\vec x, t, \ell') \models  \varphi 
 \\&  (\vec x,t, \ell) \models \varphi_{1} \: \surround{[ d_{1}, d_{2}]}  \varphi_{2}&  \Leftrightarrow    &  \mbox{ } \phantom{a}   \exists A \subseteq L^{\ell}_{[ 0,d_{2} ]} : \ell \in A    \wedge \forall \ell' \in A, (\vec x, t, \ell') \models  \varphi_{1}   
  \\ & & &    \quad \wedge  B^{+}(A)  \subseteq  L^{\ell}_{[ d_{1},d_{2} ]}   \wedge \forall \ell''  \in B^{+}(A),  (x,t, \ell'') \models \varphi_{2}.
 \end{align*}
 }
A trace $\vec x$ satisfies $\varphi$ in location $\ell$, denoted by  $(\vec x, \ell ) \models \varphi $, if and only if $ (\vec x,0, \ell) \models \varphi $. 
\label{def:boolean_sem}
\end{defi}

\begin{exa}[{\bf Spot Formation properties}]
We illustrate the Boolean semantics of SSTL and, in particular, the use of the surrounded operator to characterise spot and pattern formation of the running example introduced in Section~\ref{rex:turing}.
 In order to classify spots, sub-regions of the grid have to be identified that present a high (or low) concentration of one of the species, surrounded by a low (high, respectively) concentration of the same species.  
 For example, one can capture the spots of the A species using the spatial formula
 \begin{equation}
 \phi_{ \mathrm{spot}} := ( x^{A} \leq h   ) \surround{[d_{1},d_{2}]} ( x^{A} > h  ). 
\end{equation}
A trace $\vec{x}$ satisfies $\phi_{ \mathrm{spot}}$ at time $t$, in the location $(i,j)$, $(\vec{x},t, (i,j)) \models \phi_{ \mathrm{spot}}$, if and only if there is 
a subset $L'\subset L$, that contains $(i,j)$, such that all elements in $L'$ have a distance less than $d_{2}$ from $(i,j)$, and $x^{A}$, at time $t$, is less or equal to a constant value $h$. Furthermore, in each element in the boundary of $L'$ we have that $x^A>h$ at time $t$, and its distance from $(i,j)$ is in the interval $[d_{1}, d_{2}]$. Note that the use of distance bounds in the surrounded operator allows one to constrain the size (diameter) of the spot. 
%
%
Recall that we are considering a spatio-temporal system, so this spatial property alone is not enough to describe the formation of a pattern {\em over time}; to identify the insurgence time of the pattern and whether it remains stable over time we combine the spatial property with  temporal operators in the following way:
\begin{equation}
 \phi_{ \mathrm{spotFormation}} := \mathcal{F}_{[T_{ \mathrm{pattern}}, T_{ \mathrm{pattern}} + \delta]} \mathcal{G}_{[0,T_{ \mathrm{end}}]} ( \phi_{ \mathrm{spot}});
 \label{phi}
\end{equation}
$\phi_{ \mathrm{spotFormation}}$ states that eventually, at a time between $T_{ \mathrm{pattern}}$ and  $T_{ \mathrm{pattern}} + \delta $, the property $\phi_{ \mathrm{spot}}$ becomes true and remains true for at least $T_{ \mathrm{end}}$ time units.
%
Figure \ref{result}(b) shows the satisfaction of the property $\phi_{\mathrm{spotFormation}}$ in each cell $(i,j) \in L$ for a simulated trace, 
for the Boolean semantics with formula parameters set to $d_{1}=1$ and $d_{2}=6$. In Figure~\ref{result}(d) we similarly show the satisfaction of the same property for a smaller value of the maximal dimension of a spot $d_{2}=4$ (the other parameters are as indicated in the caption of Figure~\ref{result}). The plot shows the satisfaction at time $t = 0$ because by default we have that $(\vec x, \ell ) \models \varphi $, if and only if $ (\vec x,0, \ell) \models \varphi $.
%
%
%
It is clear that the satisfaction of the property is capable of  identifying which locations belong to the spots and which not. 
If we make the distance constraint stricter, by reducing the width of the interval $[d_{1}, d_{2}]$, we are able to identify only the ``centre'' of the spot, as shown in Figure \ref{result} (d). However, in this case we may fail to identify spots that have an irregular shape (i.e., that deviate too much from a circular shape).
\end{exa} 
%

\begin{exa}[{\bf Pattern Formation properties}]
The next formula, $\phi_{\mathrm{pattern}}$, describes the global spatial pattern in which every location is part of a spot or has a nearby spot. This is expressed by the following SSTL formula:
\begin{equation}
\phi_{\mathrm{pattern}} := \everywhere{[0,d_{max}]} \somewhere{[0,d_{spot}]}  \phi_{\mathrm{spot}}, 
\label{phi:pattern}
\end{equation}
 where $\diamonddiamond$ and $\boxbox$ are the everywhere and somewhere operators, $d_{max}$ is chosen to cover the whole space, and $d_{spot}$ specifies the maximal distance between spots. Checking this formula in a random location of our space is enough to verify the presence of the pattern; this is so because the first part of the formula, $\everywhere{[0,d_{max}]}$, covers all the locations of the grid. This is an example of how one can describe a global property also with a semantics that verifies properties in single locations. For $d_{max}=45$ and $d_{spot}=15$ (and the other parameters are as indicated in Figure~\ref{result}) for a solution of the system (\ref{model}) we obtain the result {\em true} for the Boolean semantics.
 %
%
\begin{figure}[tbp]
\begin{center}
\subfigure[]{
\label{fig:check_a}
\includegraphics[height=4.1cm]{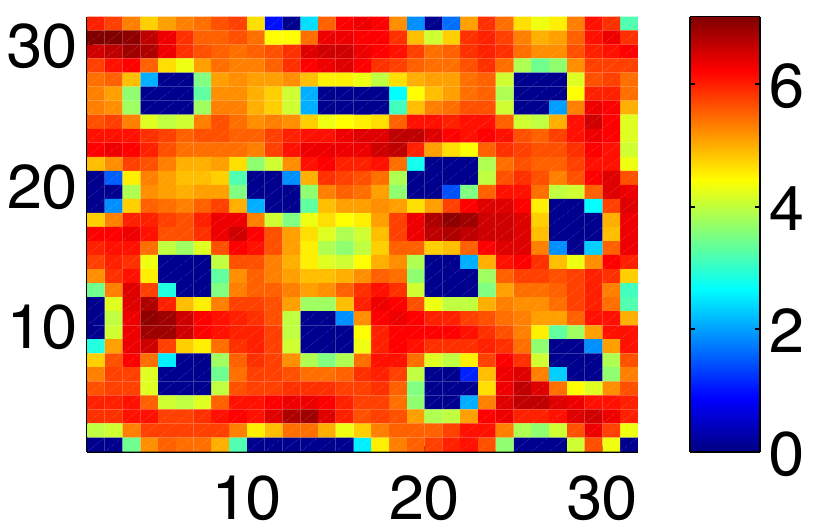} 
}
\subfigure[]{
\label{fig:check_b}
\includegraphics[height=4.1cm]{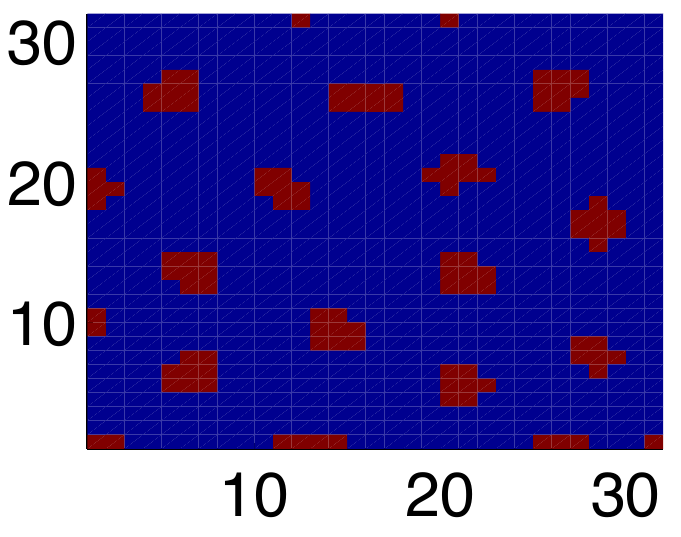} 
}

\hspace{5mm}

\subfigure[]{
\label{fig:check_c}
\includegraphics[height=4.1cm]{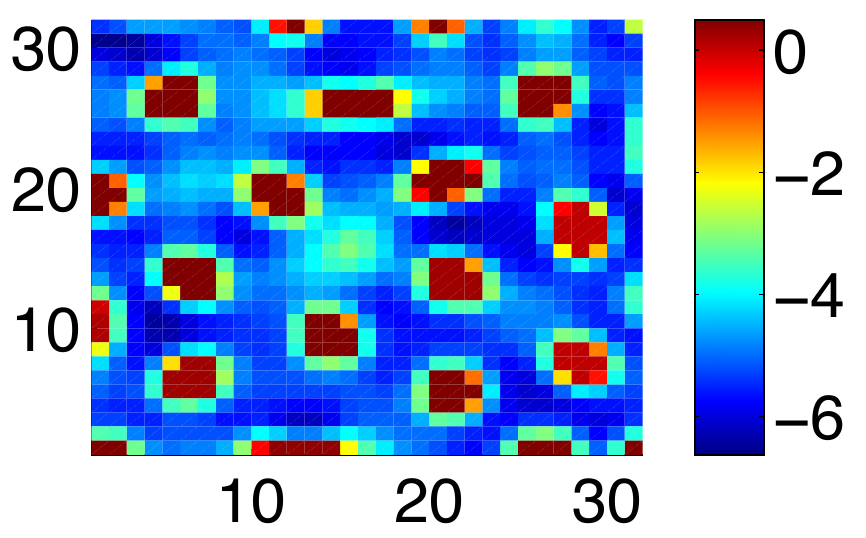} 
}
\subfigure[]{
\label{fig:check_d}
\includegraphics[height=4.1cm]{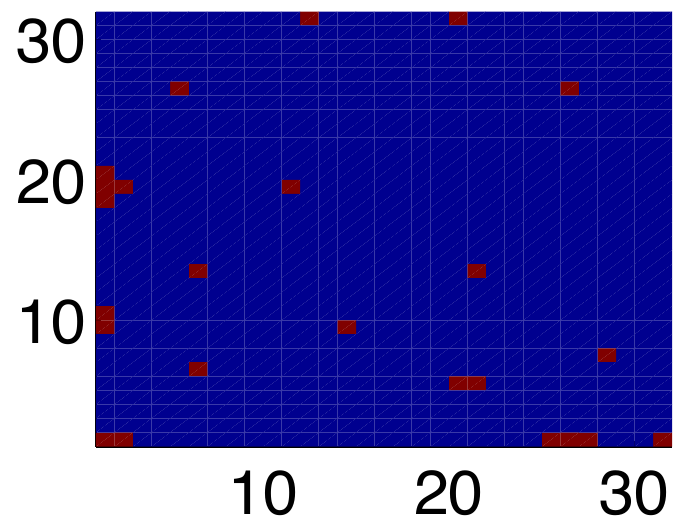} 
}
\end{center}
%
%
%
\caption{Satisfaction of formula (\ref{phi}) with $h=0.5, T_{\mathrm{pattern}}= 19, \delta = 1, T_{\mathrm{end}} =30,  d_{1}=1, d_{2}=6$ for (b), (c) and $d_{2}=4$ for (d). (a) Concentration of $A$ at time t = 50; (b) (d) Boolean semantics of the property $\phi_{ \mathrm{spotFormation}}$; the cells (locations) that satisfy the formula are in red, the others are in blue;
(c) Quantitative semantics of the property $\phi_{ \mathrm{spotFormation}}$; The value of the robustness is given by a colour map as specified in the legend on the right of the figure.}
\label{result}
\end{figure}
Changing the diffusion constants $D_{1}$ and $D_{2}$ affects the shape and size of the spots or disrupts them, as we can see in Figure~\ref{fig:NoPattern}. We evaluate the pattern formula (\ref{phi:pattern}) with parameters as in Figure~\ref{result}, for the patterns in Figure~\ref{fig:NoPattern}(a) and (b),
where $D=[1.5, 23.6]$ and $D=[8.5, 40.7]$, respectively, and the other parameters equal to the previous model. This gives the result {\em false} in both cases. 
Formula
 (\ref{phi}),  though,  is still true in some locations.  This is due to the irregularity of the spots (where, as in Figure~\ref{fig:nopattern1}, some spots can have a shape similar to the model in Figure~\ref{result} (a)), or due to particular boundary effects on the border of the grid (where fractions of spots still remain, as in Figure \ref{fig:nopattern1}). 
\begin{figure}[!t]
\begin{center}
\subfigure[]{
\label{fig:nopattern1}
\includegraphics[width=.39\textwidth]{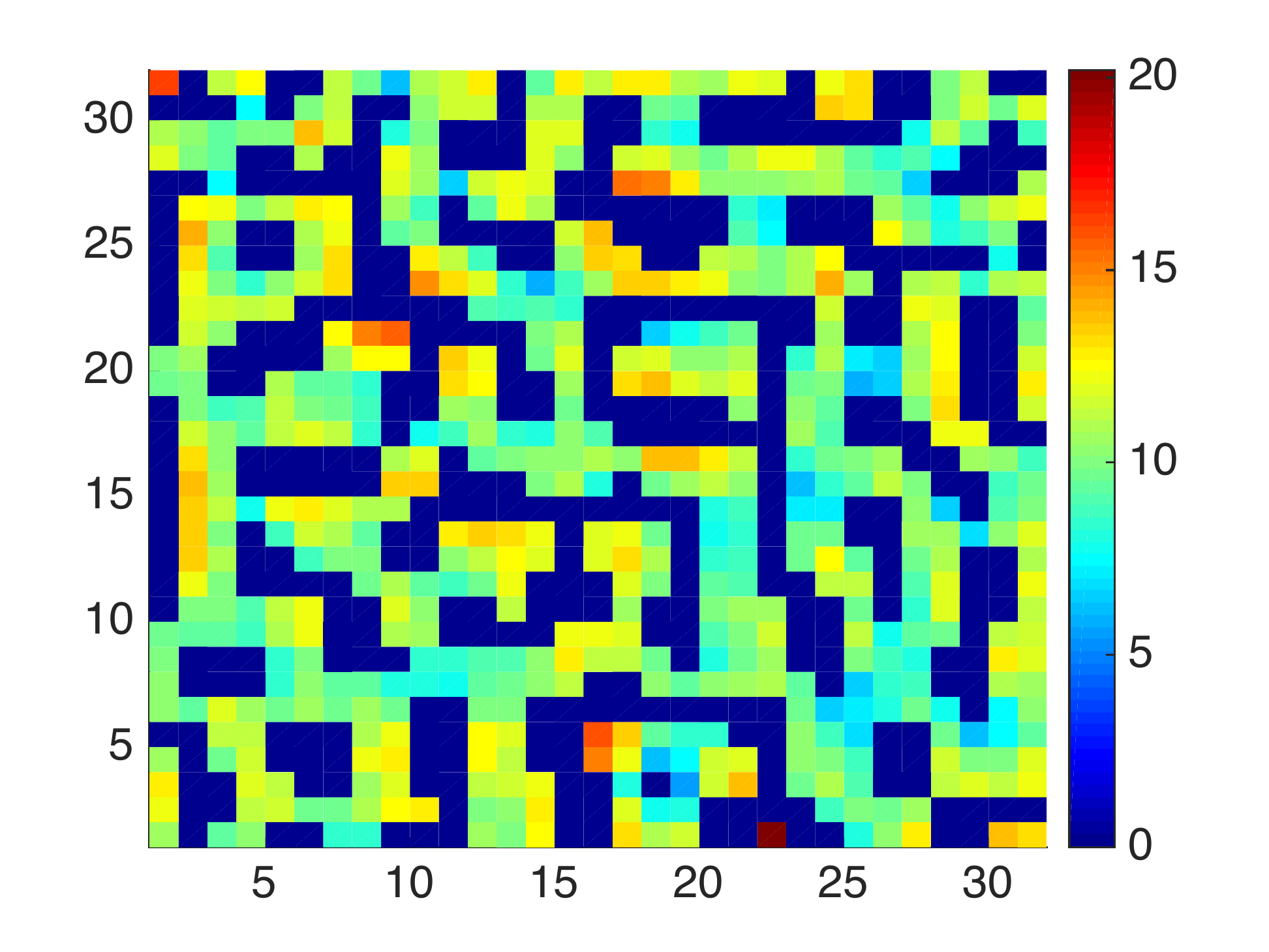} 
}
\subfigure[]{
\label{fig:nopattern2}
\includegraphics[width=.39\textwidth]{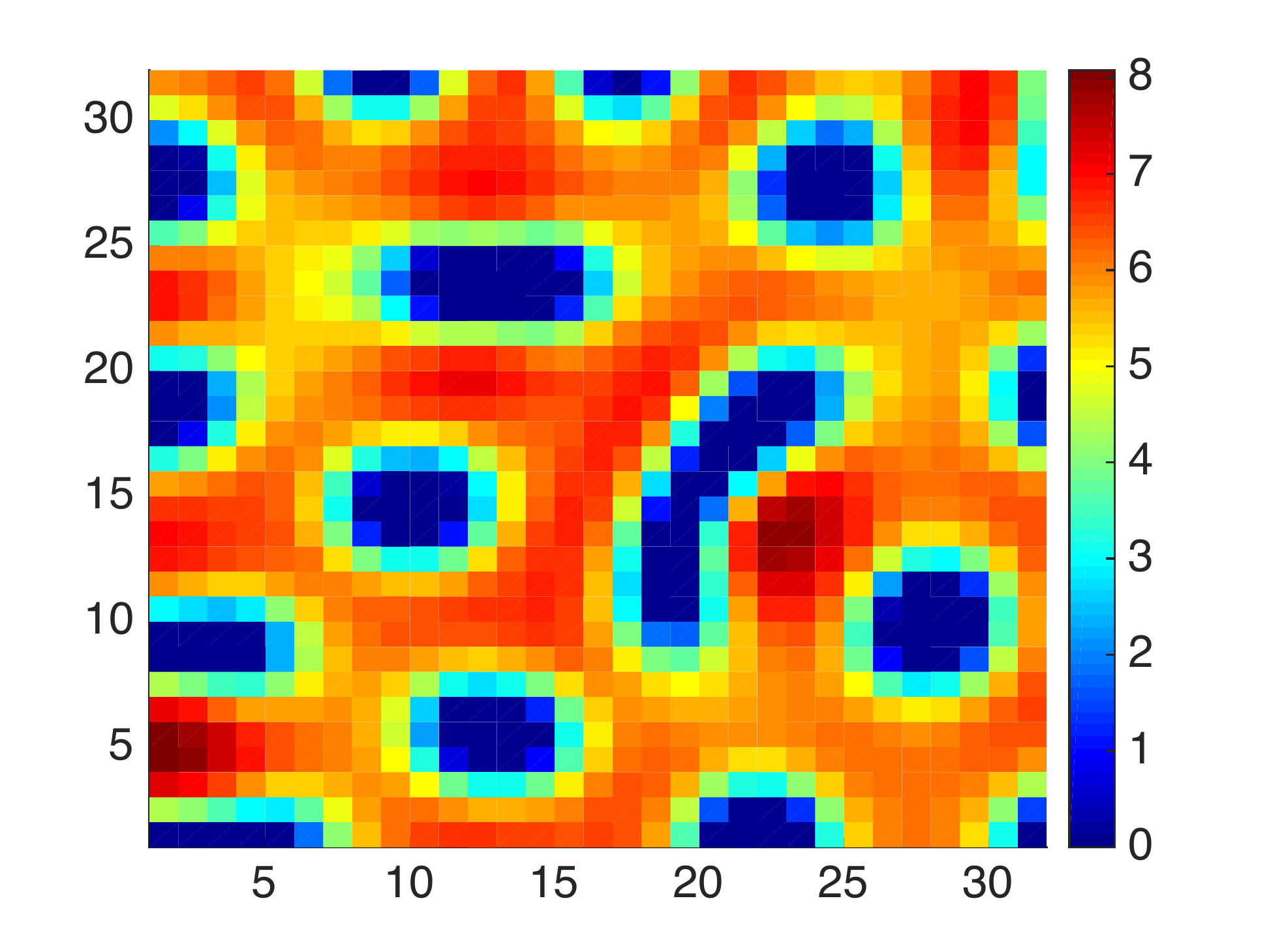} 
}
\end{center}
\caption[Snapshots of the model with 
``bad'' diffusion rate D]{Snapshots at time $t=50$ of $x^A$ for the model (\ref{model}) with parameters with parameters $K=32, R_{1}=1, R_{2}=-12, R_{3}=-1, R_{4}= 16$, and $D=[1.5, 23.6]$ in (a) and $D=[8.5, 40.7]$ in (b).}
\label{fig:NoPattern}
\end{figure}
\end{exa}

\subsection{Quantitative Semantics.} 
The quantitative semantics returns  a {\em real value} that can be interpreted as a measure of the strength with which the specification is satisfied or falsified by an observed trajectory. More specifically, the sign of such a satisfaction score is related to the truth of the formula (positive stands for true), while the absolute value of the score is a measure of the robustness of the satisfaction or dissatisfaction. This definition of quantitative measure is based on \cite{Donze2010,Donze2013}, and it is a reformulation of the  robustness degree of \cite{fainekos_robustness_2009}.

\begin{defi} [{\bf SSTL Quantitative Semantics} ] 
\label{quantitative}
The quantitative satisfaction function $\rho(\varphi, \vec{x},t, \ell)$ for an SSTL formula $\varphi$ over a spatio-temporal trace $\vec x$ is given by:
{\small
\begin{align*}
&\rho(\mu, \vec{x},t, \ell)  & = &\mbox{ } \phantom{a}  f( \vec x(t, \ell)) \quad \mbox{ where } \mu \equiv (f \geq 0)\\ 
& \rho (\neg \varphi,\vec x,t, \ell)  & =  & \mbox{ } \phantom{a}  - \rho (\varphi,\vec x,t, \ell)\\
&\rho( \varphi_{1} \wedge  \varphi_{2}, \vec x,t, \ell)  & = &\mbox{ } \phantom{a} \min ( \rho( \varphi_{1},\vec x,t, \ell),\rho( \varphi_{2},\vec x,t, \ell) )\\
& \rho(  \varphi_{1} \: \mathcal{U}_{[t_{1},t_{2}]}  \varphi_{2}, \vec x,t, \ell) & = &\mbox{ } \phantom{a}  \sup_{t'\in t+[t_{1},t_{2}]}(\min ( \rho( \varphi_{2},\vec x,t', \ell), \inf_{t'' \in [t,t']}(\rho( \varphi_{1},\vec x,t'', \ell))))\\
& \rho(\somewhere{[d_{1},d_{2}]} \varphi , \vec x,t, \ell) & = &\mbox{ } \phantom{a}  \max\{\rho( \varphi ,\vec x,t, \ell')\mid\ell' \in L, (\ell',\ell)\in E^{*},
d_{1}\leqslant d(\ell',\ell) \leqslant d_{2}\}\\
& \rho(  \varphi_{1} \: \mathcal{S}_{[  d_{1},  d_{2}]}  \varphi_{2} , \vec x,t, \ell) & = &\mbox{ } \phantom{a}  \max_{A \subseteq L^{\ell}_{[0, d_{2}]}, \ell \in A, B^{+}(A)\subseteq L^{\ell}_{[d_{1}, d_{2}]}}( \min (\min_{\ell' \in A}\rho(\phi_{1}, \vec x, t, \ell'),
\min_{ \ell'' \in B^{+}(A)}\rho(\phi_{2}, \vec x, t, \ell'')))
\end{align*}
}
where $\rho$ is the quantitative satisfaction function, returning a real number $\rho(\varphi, \vec{x},t)$ quantifying the degree of satisfaction of the property $\varphi$ by the trace $\vec{x}$ at time $t$. Moreover, $\rho(\varphi, \vec{x}, \ell):=\rho(\varphi, \vec{x},0, \ell)$.
\end{defi}
%
The definition for the surrounded operator is essentially obtained from the Boolean semantics by replacing conjunctions and universal quantifications with the minimum and disjunctions and existential quantifications with the maximum, as done in  \cite{Donze2010,Donze2013} for STL.  


\begin{rem}[Soundness Property] As for STL, the quantitative semantics of SSTL is sound with respect to the Boolean semantics. It means the  $\rho$ is positive whenever the property holds in the Boolean semantics, and that it has a negative value if the property does not hold. This can be proved by induction on the structure of the formula, whereas it is already shown for STL operators (atomic propositions, Boolean and temporal operators, for fixed locations). The new spatial operators are defined as compositions of maximum and minimum functions (on finite sets and for a fixed time). Therefore, it is easy to prove that the property holds also for them.
Indeed, considering for example $\min\{\rho( \varphi ,\vec x,t, \ell_1), \rho (\varphi ,\vec x,t, \ell_2)\}$, by induction $\rho( \varphi ,\vec x,t, \ell_1)$ and $\rho( \varphi ,\vec x,t, \ell_2)$ satisfy the soundness property. If they are both positive, they both satisfy $\phi$ and also the minimum is positive. If one $\rho( \varphi ,\vec x,t, \ell_i)$ is negative, it does not satisfy $\phi$ and the minimum is negative. A similar consideration can be done for the maximum.
\end{rem}

\begin{rem}[Correctness Property] The quantitative semantics satisfies the correctness property with respect to the Boolean semantics. It means that for each formula $\phi$ and $\forall \vec x_1, \vec x_2$, it holds: $$ (\varphi ,\vec x_1,t, \ell) \models \phi \mbox{ and } || \vec x_1 - \vec x_2||_\infty < \rho(\varphi ,\vec x_1,t, \ell) \Rightarrow (\varphi ,\vec x_2,t, \ell) \models \phi$$
Recalling that the new spatial operators are defined as compositions of maximum and minimum functions (on finite set and for fixed time), also this property can be proved by induction on the structure of the formula.
 It holds for atomic predicates, Boolean and temporal operators by the correctness proof of the quantitative semantics for MITL~\cite{Fainekos2009}. We need just to prove that it holds for maximum and minimum functions of the robustness function with fixed time. Let $\rho( \psi ,\vec x_1, \ell) = \min\{\rho( \varphi ,\vec x_1, \ell_1), \rho (\varphi ,\vec x_1, \ell_2)\} $ with $\varphi$ satisfying correctness by induction. If  $|| \vec x_1 - \vec x_2||_\infty <\rho( \psi ,\vec x_1, \ell) $ and $\rho(\psi ,\vec x_1,\ell) > 0 $ (soundness property), i.e.  $(\vec x_1, \ell_i) \models \phi$ for $i= 1,2$, by induction on $\varphi$ and the minimum function we have that $(\vec x_2, \ell_i) \models \phi$ for $i= 1,2$, then $(\vec x_2, \ell) \models \psi$; similarly it can be proved for the maximum. 
\end{rem} 


\begin{rem}
The quantitative semantics is a measure of how robust is the truth value of a formula $\phi$ for a given trace $\vec{x}$. More specifically, it tells us the maximum size of a perturbation of the secondary signal to preserve the satisfaction of $\phi$ \cite{Donze2010} (a variant of our definition, \cite{fainekos_robustness_2009} returns a robustness measure for the primary signal, i.e. the trace $\vec{x}$). This perturbation can be applied to the quantitative value of the secondary signal at any point in time and space. As such, the value of the robustness defines a tube in the space of secondary signals, uniformly with respect to space-time, containing equisatisfiable trajectories (with respect to the Boolean semantics): any trace in this tube will have the same truth value of $\phi$ as  $\vec{x}$.
We stress that perturbations are applied pointwise in space and time: this version of the quantitative semantics tells us nothing about the effect of time and space warps. This would require a different notion of quantitative semantics, as done e.g. in \cite{Donze2010, avstl} for time only, coupled with a meaningful notion of perturbation of a discrete space (e.g. on edge weights).  
\end{rem}



\begin{exa}[{\bf Robustness of Spot Formation}]
As a first example, consider again the property $\phi_{ \mathrm{spotFormation}}$. 
Figure \ref{result}(b) shows the Boolean satisfaction of the property $\phi_{\mathrm{spotFormation}}$ in each cell $(i,j) \in L$, 
for $d_{1}=1$ and $d_{2}=6$ and Figure \ref{result}(c) shows the result of the {\em quantitative} semantics for the same formula, showing for each cell $(i,j)$ the value of the robustness with which the formula is satisfied. Positive values indicate cells where the formula holds, negative values where it does not. Where the value is positive its value amounts to 0.3. This low value is due to the specific choice of the threshold value $h=0.5$.
\end{exa}

\begin{exa}[{\bf Robustness of Pattern Formation}]
The second example considers the global formula $\phi_{\mathrm{pattern}}$. This formula resulted in {\em false} applying the Boolean semantics for both patterns in Figure~\ref{fig:NoPattern}. The quantitative semantics gives -0.05 as a result in both cases, indicating a weak robustness. 
\end{exa}

\begin{exa}[{\bf Perturbation}]
A strength of spatio-temporal logics is the possibility to nest the temporal and spatial operators. We illustrate this in the following scenario as a variant of the running example. We set as  initial conditions for the dynamical system (\ref{model}) its stable state, i.e. the concentrations of substances $A$ and $B$ at time $50$ (see Figure \ref{result}(a)). We introduce a small perturbation by changing a single value in a specific location in the centre of a spot. The idea is to study the effect of this perturbation over time by checking whether it will disrupt the system or not. Specifically, we perturb the cell $(6,6)$ in Figure~\ref{result}(a), by setting  $x_{6,6}^{A}(0)=10$ while its original value was 0. Dynamically, the perturbation is quickly absorbed and the system returns to the previous steady state. Formally, we consider the following property:
{
\begin{equation}
 \phi_{\mathrm{pert}} :=(x^{A}\geq h_{\mathrm{pert}}) \wedge(\phi_{\mathrm{absorb}} \surround{[d_{m},d_{M}]} \phi_{\mathrm{no\_effect}}); 
 \label{phi_pert}
\end{equation}
}
The meaning of $\phi_{\mathrm{pert}}$ is that the induced perturbation remains confined inside the original spot when the property is satisfied. More in detail, a trace $\vec{x}$ satisfies $\phi_{\mathrm{pert}}$ in the location $(i,j)$, i.e. $ (\vec x, (i,j)) \models \phi_{\mathrm{pert}}$,  if and only if $x_{i,j}^{A}(0) \geq h_{\mathrm{pert}} $, with $h_{\mathrm{pert}}=10$, (the location is perturbed) and if there is a subset $L' \subseteq L$ that contains location $(i , j)$ such that all its elements have a distance less than $d_{M}$ from $(i,j)$ and satisfy $\phi_{\mathrm{absorb}}=  \mathcal{F}_{[0, T_{p}]} \mathcal{G}_{[0, T_{d}]} (  x^{A} < h')$; $\phi_{\mathrm{absorb}}$ states that the perturbation of $x^A$ is absorbed within $T_p$ units of time, stabilising back to a value $ x^{A} < h'$ for an additional period of $T_{d}$ time units. Here $h'$ is a suitable threshold capturing the fact that concentrations have returned close to their  value before the perturbation, and it is set to $h'=3$ in our case.
Furthermore, within distance $[d_{m}, d_{M}]$ from the original perturbation, where $d_M$ is chosen such that we are within the spot at cell $(6,6)$ of the non-perturbed system, $\phi_{\mathrm{no\_effect}} := \mathcal{G}_{[0, T]} (  x^{A} < h' )$ is satisfied; i.e., no relevant effect is observed, the value of $x^A$ stably remains below $h'$.  
In Figure \ref{perturb}, we report the evaluation of the quantitative semantics for $\phi_{\mathrm{pert}}$, zooming in on the $15\times 15$ lower left corner of  the original grid. As shown in the figure, the perturbed location $(6,6)$ satisfies the property.
\end{exa}

It is interesting to observe that this property requires a genuine nesting of space and time modalities, and as such it cannot be expressed in logics like SpaTeL, in which the spatial properties are limited to the level of atomic propositions.

%

%
%
%
\begin{figure}[!t]
\begin{center}
\subfigure[]{
\label{perturbA}
\includegraphics[height=10em]{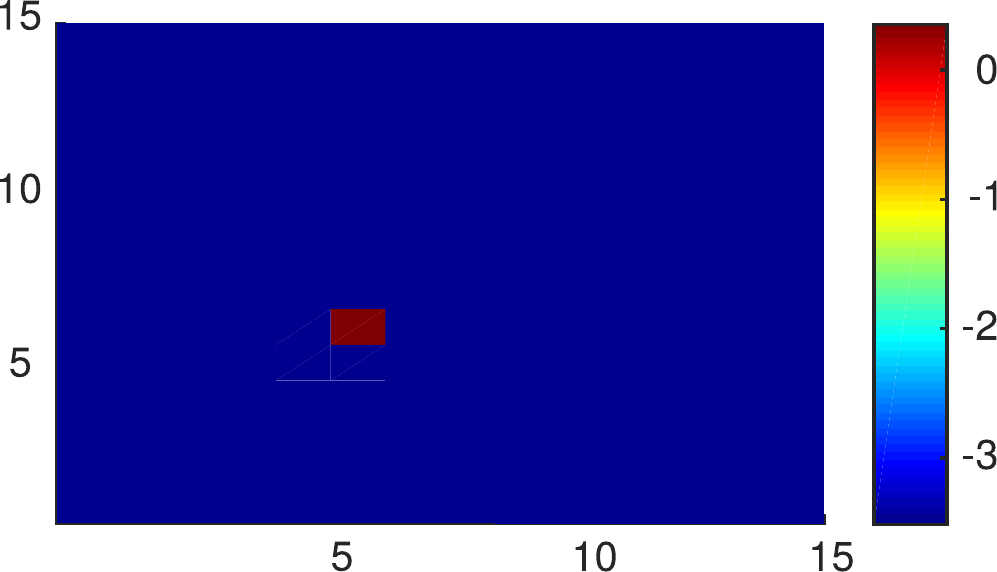} 
}
\hskip 20pt
\subfigure[]{
\label{perturbB}
\includegraphics[height=10em]{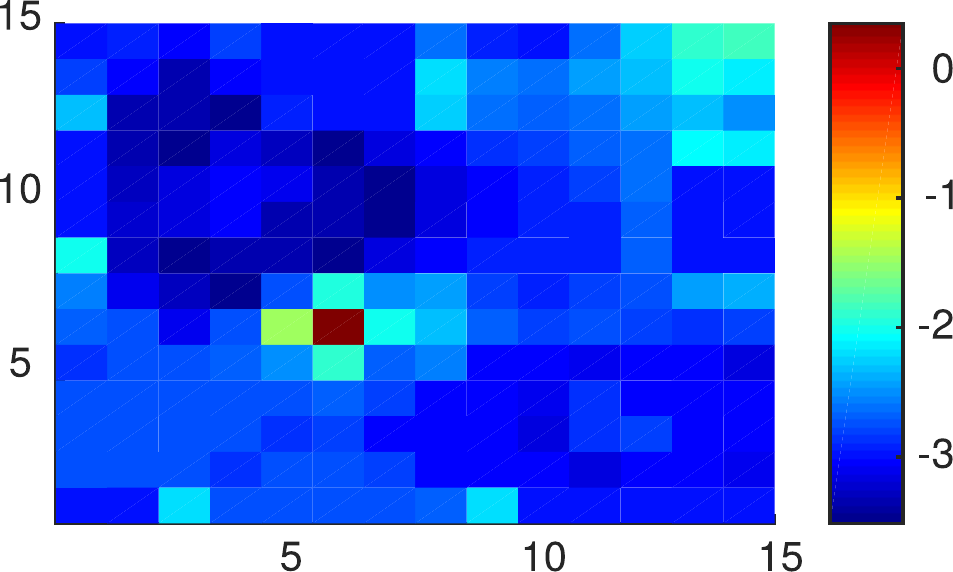} 
}
\end{center} 
\caption {Boolean (a) and quantitative (b) semantics for formula~\ref{phi_pert}  with parameters 
$h_{\mathrm{pert}} = 10$, $d_m =1$, $d_M=2$, $T_p=1$, $T_d = 10$, $h'=3$, and $T=20$.} 
%
\label{perturb}
\vspace{-3ex}
\end{figure}



%% file: algo.tex

\section{Monitoring Algorithms }
\label{sec:alg}

In this section, we present the offline  monitoring algorithms to check the satisfaction of a formula $\varphi$ on a trace $\vec{x}(t, \ell)$. 
The monitoring procedures spatially extend the property monitors introduced in \cite{Maler2004} for the Boolean and in  \cite{Donze2013} for the quantitative semantics of STL.

As for STL formulae, our algorithms work with a bottom-up approach on the syntax tree of  $\varphi$, iteratively computing the temporal signals of each subformula. Each node of the tree represents a subformula, the leafs are the atomic propositions and the root represents the complete formula.
Given a spatio-temporal trace $\vec x (t,\ell) $, the algorithm starts computing the spatio-temporal Boolean/quantitative signals of all the atomic propositions, then it moves upwards in the tree computing the spatio-temporal Boolean/quantitative signals of a node using the signals of its child and a specific algorithm for each operator of the logic. Finally, the spatial Boolean/quantitative satisfaction function corresponds to the value of the signal at time zero $\rho_{\varphi}( 0, \ell)$. An example of the procedure is shown in Figure~\ref{treeSSTL}.

In the case of the Boolean semantics, for each subformula $\psi$, the algorithm constructs a spatio-temporal signal $s_{\psi}$ s.t. $s_{\psi}(\ell,t)=1$ iff the subformula is true in position $\ell$ at time $t$. In the case of the quantitative semantics, for each subformula $\psi$, the signal $s_{\psi}$ corresponds to the value of the quantitative satisfaction function $\rho$, for any time $t$ and location $\ell$. Here, we discuss in detail the algorithms to check the new spatial operators: the somewhere and surrounded operators; the procedures for the other Boolean and temporal operators are similar to STL and will be just briefly recalled. 

The processing of the somewhere operator is a simple extension of the disjunction operator. 
The treatment of the bounded surrounded modality $\psi= \varphi_{1} \surround{[w_{1},w_{2}]}   \varphi_{2}$, instead, deviates substantially from all these procedures and will be discussed in more detail. In particular, in the following, we will present two recursive algorithms to compute the Boolean and the quantitative satisfaction, assuming the Boolean/quantitative signals of $\varphi_1$ and $\varphi_2$ being known.

We remark that this surround operator requires very different algorithms from those developed for timed modalities, as space is bi-directional, thus it makes sense to consider both \emph{reaching} and \emph{being reached}; classical path-based model checking does not coincide with spatial model checking also because loops in space are not relevant in the definition of \emph{surrounded} operators.


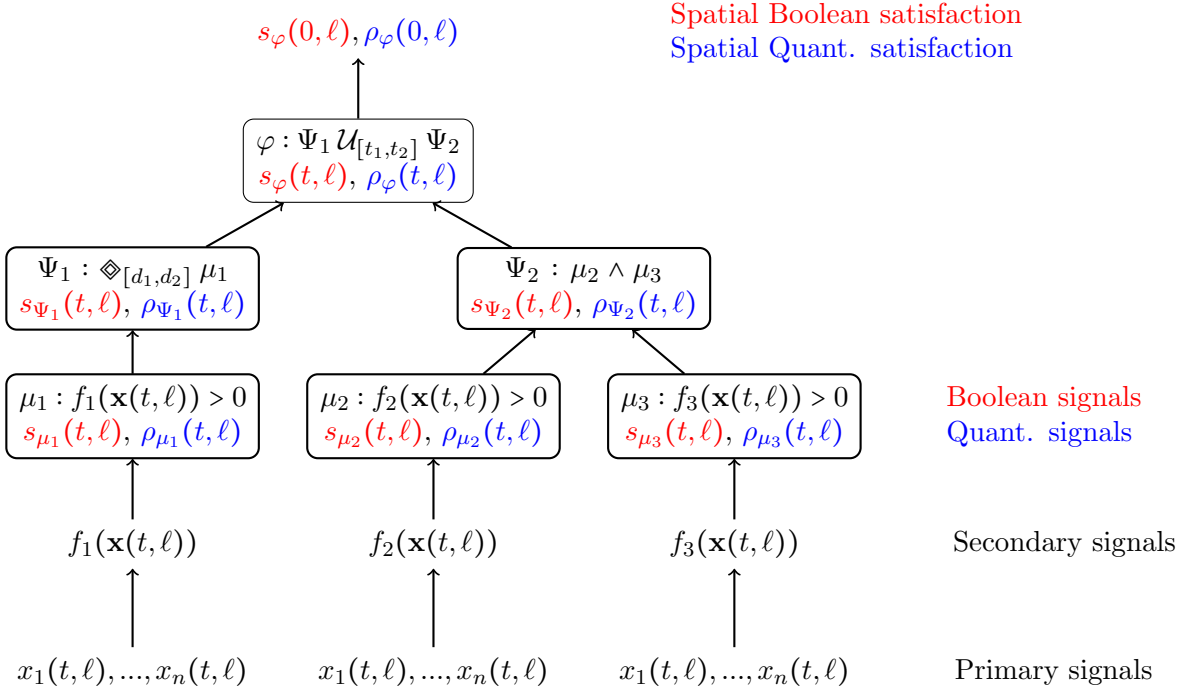
\begin{figure}[ht]
\center
\begin{tikzpicture}
  [level 1/.style={sibling distance=60mm}, level 2/.style={sibling distance=60mm}, level 3/.style={sibling distance=40mm}, level distance=17mm,
   every node/.style={fill=white!20,rounded corners},
   edge from parent/.style={black,<-,thick,draw}]
  \node[fill = white!50!white] (RR)  
{ ${\color{red}{s_{\varphi} (0,\ell)}}, {\color{blue}{\rho_{\varphi} (0,\ell)}}$ }  
 child {node[rectangle,draw,text width=2.8cm,text centered] (FM) {$\phi: \Psi_{1} \: \until{[t_{1},t_{2}]} \: \Psi_{2}$
 ${\color{red}{s_{\varphi} (t,\ell)}}$, ${\color{blue}{\rho_{\varphi} (t,\ell)}}$}
    child {node[rectangle,draw,text width=3.1cm,text centered] (Fl) {
$\Psi_1:  \somewhere{[d_1,d_2]}  \: \mu_{1}$
${\color{red}{s_{\Psi_{1}} (t,\ell)}}$, ${\color{blue}{\rho_{\Psi_{1}} (t,\ell)}}$}
    child {node[rectangle,draw,text width=3.1cm,text centered] (Ml) {
$\mu_{1} :  f_{1}(\vec x (t,\ell))  > 0$
${\color{red}{s_{\mu_{1}} (t,\ell)}}$, ${\color{blue}{\rho_{\mu_{1}} (t,\ell)}}$}
		child {node (Sl) { $f_{1}(\vec x (t,\ell)) $}
		child {node [](Pl) { $x_{1}(t,\ell),...,x_{n}(t,\ell)$}}
		}}
  }
    child {node[rectangle,draw,text width=3.1cm,text centered] (Fr) {
$\Psi_2:  \mu_{2} \wedge \mu_{3}$ 
${\color{red}{s_{\Psi_{2}} (t,\ell)}}$, ${\color{blue}{\rho_{\Psi_{2}} (t,\ell)}}$}
child {node[rectangle,draw,text width=3.1cm,text centered]  (Mr) {
		$\mu_{2} :  f_{2}(\vec x (t,\ell)) > 0$
		${\color{red}{s_{\mu_{2}} (t, \ell)}}$, ${\color{blue}{\rho_{\mu_{2}} (t,\ell)}}$}
	        child {node (Sr) { $f_{2}(\vec x (t,\ell)) $}
	        child {node [](Pr) { $x_{1}(t,\ell),...,x_{n}(t,\ell)$}}
	        }}
child {node[rectangle,draw,text width=3.1cm,text centered]  (Mrr) {
		$\mu_{3} :  f_{3}(\vec x (t,\ell)) > 0$
		${\color{red}{s_{\mu_{3}} (t, \ell)}}$, ${\color{blue}{\rho_{\mu_{3}} (t,\ell)}}$}
	        child {node (Sr) { $f_{3}(\vec x (t,\ell)) $}
	        child {node [](Pr) { $x_{1}(t,\ell),...,x_{n}(t,\ell)$}}
	        }}}};        
\node[fill = white!50!white,text width=4.1cm] (S) [right=of Mrr]  {\color{red}{Boolean signals} \\ \color{blue}{Quant. signals}};
    \node[fill = white!50!white] (S) [right=of Sr]  {\qquad Secondary signal};
  
     \node[fill = white!50!white,text width=7cm] (S) [right=of RR]  {\qquad \qquad \color{red}{Spatial Boolean satisfaction} \\ 
     \qquad \qquad \color{blue}{Spatial Quant. satisfaction}};
    \node[fill = white!50!white] (S) [right=of Sr]  {\qquad Secondary signals};
        \node[fill = white!50!white] (S) [right=of Pr]  {$ \mbox{ }$Primary signals};
\end{tikzpicture}
\caption{The monitoring procedure of the SSTL formula $\phi : (\somewhere{[d_1,d_2]}  \: \mu_{1}) \until{[t_{1},t_{2}]} ( \mu_{2} \wedge \mu_{3}) $.}
\label{treeSSTL}
\end{figure}


\subsection{Boolean Monitoring}

The algorithm proceeds inductively bottom-up on the parse tree of the formula.
Given a formula $\varphi$, to determine if $(\vec{x}, \ell) \models \varphi$, we construct for all  $\psi$ that are subformulas of $\varphi$, a Boolean signal $s_{\psi} : [0,T] \times L \rightarrow \bb{B}$  s.t. $s_{\psi}(t, \ell)=1 \mbox { iff } (\vec{x} ,t, \ell) \models \psi$ and 0 otherwise.  At the termination of the algorithm, we obtain the signal $s_{\phi}(t, \ell)$ whose value at $t=0$ determines whether the trace $\vec x$ satisfies  $\varphi$ in location $\ell$ (at time 0). 
The properties can be verified pointwise for each location and each time independently, indeed $(\vec{x}, \ell, t) \models \varphi$ means ``the trace $\vec{x}$ in location $\ell$ at time $t$ satisfies the property $\phi$''.
 
To optimise the monitoring procedure, we split the time domain according to the
{\it minimal interval covering $\mathcal{I}_{s_1,\ldots, s_n}$ consistent with a set of temporal Boolean signals $s_1, \ldots s_n$} (as in~\cite{Maler2004}) that we describe below.
A temporal Boolean signal is a function $s: [0,T] \rightarrow \bb{B}$. Note that,  we can represent the signal $s_{\psi}: [0, T] \times L \rightarrow \bb{B}$ as a finite collection of temporal signals $\{s_{\psi,\ell} \}_{\ell \in L}$ where $s_{\psi,\ell}(t) := s_{\psi}( \ell, t)$.

\begin{defi}
Given a time interval $I$, and a set of temporal signals $s_1, \ldots s_n$ with $s_i:I \rightarrow \mathbb{B}$, the \emph{minimal interval covering $\mathcal{I}_{s_1,\ldots, s_n}$ of $I$ consistent with the set of signals $s_1,\ldots, s_n$} is the shortest finite sequence of left-closed right-open intervals $I_{1}, ..., I_{h}$ such that $\bigcup_{j}I_{j}=I$, $I_{i} \bigcap I_{j}=\emptyset$, $ \forall i \neq j$, and 
for $k\in \{1,\ldots,n\}$, $s_k(t) = s_k(t')$ for all $t$, $t'$ belonging to the same interval \footnote{The fact that we can always obtain a finite interval covering is a consequence of the restriction to closed intervals $[t_1,t_2]$, $t_1<t_2$, in STL. Further details about signals and interval covering are provided in~\cite{Maler2004}.}. 
Restricting to a single signal $s$,  the \emph{positive minimal interval covering} of $s$ is $\mathcal{I}^{+}_{s}=\{ I \in \mathcal{I}_{s} \big{|} \forall t \in I : s(t)=1\}$. The \emph{negative minimal interval covering} of $s$ is $\mathcal{I}^{-}_{s}=\{ I \in \mathcal{I}_{s} \big{|} \forall t \in I : s(t)=0\}$, and it holds  $\mathcal{I}_{s}= \mathcal{I}^{+}_{s} \bigcup \mathcal{I}^{-}_{s}$ and $\mathcal{I}^{+}_{s} \bigcap \mathcal{I}^{-}_{s}= \emptyset$.

\label{def:intcov}
\end{defi}


The positive interval covering $\mathcal{I}^{+}_{s_{\psi, \ell}}$ corresponds to the {\it satisfaction set} of the formula over the signal $s_{\psi, \ell}$. Futhermore, any signal can be written as $s=s_{1} \vee s_{2} \vee \cdots \vee s_{k} $ where each $s_{i}$ is a {\it unitary signal}, meaning that it has a singleton positive interval, i.e.,  $\mathcal{I}^{+}_{s_{i}}=\{[t_1,t_2)\} $ for some $t_1<t_2 \in \mathbb{R}_{\geqslant 0} $. 
Then , we have that $\mathcal{I}_{s}=\mathcal{I}_{s_1,\ldots, s_k}$ and $\mathcal{I}^{+}_{s} = \bigcup_{i} \mathcal{I}^{+}_{s_{i}}$. The idea is that the satisfaction set of a formula over a signal can be seen as the union of disjoint intervals. 

Using these definitions of signals, interval coverings, and satisfaction set, the procedure for the classic operators of STL is similar to the one described in~\cite{Maler2004}.  We briefly recall these procedures in the following and then we describe the algorithms for the new spatial operators.

\subsection*{Atomic Predicates} $\psi= \mu$. The computation of the Boolean signal associated with an atomic predicate is a direct application of Definition \ref{def:boolean_sem}: $s_{\mu,\ell}(t) = \mu(\vec{x}(t,\ell))$.
\subsection*{Negation} $\psi=\neg \phi$, then  $\mathcal{I}^{+}_{s_{\neg \phi}, \ell}= \mathcal{I}^{-}_{s_{\phi}, \ell}.$
\subsection*{Disjunction}  $\psi=\phi_{1} \vee \phi_{2}$,  then, given $s_{ \phi_{1}, \ell}$, $s_{ \phi_{2}, \ell}$, let $\mathcal{I}$ be the minimal interval covering consistent with \emph{both} signals. For each $I_{i} \in \mathcal{I}$, we construct the signal $s_{ \psi, \ell}(I_{i})=s_{ \phi_{1}, \ell}(I_{i}) \vee s_{ \phi_{2},\ell}(I_{i})$ and  we merge adjacent positive intervals to obtain $\mathcal{I_{s_{\psi,\ell}}^{+}}$.
\subsection*{Until} $\psi= \phi_{1} \mathcal{U}_{[a,b]} \phi_{2}$. As we are working with future temporal modalities, we need to shift intervals \emph{backwards}. This  has to be done independently for each unitary signal,  then taking the union of the so obtained satisfaction sets. Given two unitary signals $p$ and $q$, the signal $\psi= p \mathcal{U}_{[a,b]} q$ is the unitary signal such that $\mathcal{I}^{+}_{\psi}=\{((I_{p} \cap I_{q}) \ominus [a,b]) \bigcap I_{p} \}$, where $[m,n) \ominus [a,b]=[m-b, n-a) \bigcap[0,T]$ is the Minkowski sum.  In the general case, let $s_{\phi_{1}, \ell}= p_{1} \vee \cdots \vee p_{n}$ and $s_{\phi_{2}, \ell}= q_{1} \vee \cdots \vee q_{m}$ be signals written as union of unitary signals, then $\psi= s_{\phi_{1}, \ell} \mathcal{U}_{[a,b]} s_{\phi_{2},\ell}= \bigvee^{n}_{i=1} \bigvee^{m}_{j=1} p_{i}\mathcal{U}_{[a,b]}q_{j}$. The proof of this result can be found in \cite{Maler2004}.

\subsection*{Somewhere} $\psi= \diamonddiamond_{[d_{1},d_{2}]}   \varphi$. As remarked at the beginning of this section, and relying on the fact that we have a finite number of locations, we can process each location in the signal independently. Given the signal $s_{\psi,\ell}$,  for a \emph{fixed location $\ell$}, we can rewrite the somewhere operator as a disjunction between all signals in locations $\ell'$ s.t.  $d_{1} \leqslant d(\ell',\ell) \leqslant d_{2}$. This allows us to use the monitoring procedure for disjunction,  constructing  the minimal interval covering $\mathcal{I}$ consistent with all $s_{{\phi},\ell'}$ signals s.t. $d_{1} \leqslant d(\ell',\ell) \leqslant d_{2}$, and defining, for each $I_{i} \in \mathcal{I}:$
$$ s_{\psi, \ell} (I_{i})=\bigvee_{d_{1} \leqslant d(\ell',\ell) \leqslant d_{2}}  s_{\varphi, \ell'}(I_{i}).$$
The satisfaction set $\mathcal{I}^{+}_{s_{\psi,\ell}}$ is then the union of the positive $I_{i}$ (i.e., $I_{i}$ s.t. $ s_{\psi, \ell} (I_{i})=1$), merging adjacent positive intervals.

We stress here that the rewriting of the somewhere operator as a finite disjunction is possible only at monitoring time, when the space structure is known. In fact,  one cannot express the somewhere operator in terms of the or-operator in general, as this requires the knowledge of the spatial model to be verified, both in terms of locations and of distances among them.  Therefore, the spatial somewhere operator is not merely syntactic sugar. Additionally, it can be applied so to countably infinite discrete spaces, and it can be generalised to continuous spaces~\cite{Ci+16}. Furthermore, even if we assume that the space is finite and if we encode the space structure of the model in the formula, expanding the operator as a disjunction would produce a  blowup of the size of the formula which is \emph{exponential} in the nesting level of spatial operators, and hence it would result in an exponential increase in the  complexity of the monitoring procedure. 

%

\subsection*{Surrounded} $\psi= \varphi_{1} \surround{[d_{1},d_{2}]}   \varphi_{2}$.
Algorithm \ref{algo:Boolean} presents the procedure to monitor the Boolean semantics of  $\psi$ at location $\ell$, returning the temporal Boolean signal $s_{\psi, \ell}$ of $\psi$ at location $\ell$.
The algorithm first computes the set of locations $L^{\ell}_{[0,d_{2}]}$ that are at distance $d_2$ or less from $\ell$, and then, recursively, the temporal Boolean signals $s_{\varphi_1,\ell'}$ and $s_{\varphi_2,\ell'}$, for $\ell' \in L^{\ell}_{[0,d_{2}]}$. 
These signals provide the satisfaction of the sub-formulas $\varphi_1$ and $\varphi_{2}$ at each point in time, and for each location of interest. 
Then, a minimal interval covering consistent with all the signals $s_{\varphi_{1},\ell'}$ and $s_{\varphi_{2},\ell'}$ is computed, and to each such interval, a core procedure similar to that of \cite{ciancia2014} is applied. 
More specifically, we first compute the set of locations $W$ in which both $\varphi_1$ and $\varphi_2$ are false, and that belong to the external boundary of the set of locations that satisfy $\varphi_1$ ($V$) or $\varphi_2$ ($Q$).
The locations in $W$ are ``bad'' locations, that cannot be part of the external boundary of the set $A$ of $\varphi_1$-locations which has to be surrounded exclusively by $\varphi_2$-locations. Hence, the main loop of the algorithm removes iteratively from $V$ all those locations that have a neighbour in $W$ (set $N$, line 13), constructing a new set $T$ containing only those locations in $N$ that do not satisfy $\varphi_2$, until a fixed point is reached. 
As each location can be added to $W$ and is processed only once, the complexity of the algorithm is linear in the number of locations and linear in the size of the interval covering. A correctness theorem, stated below, can be proven in a similar way as in~\cite{ciancia2014}. The proof is reported in Appendix~\ref{app:Booleanproof}.


\begin{restatable}{thm}{restBooleansurr}
Given a graph $G=(L,w,E)$, two properties $\phi_{1}$ and $\phi_{2}$, a trace $\vec x$ and a location $\ell$, let $s_{\psi, \ell}$=BoolSurround$(G, \vec x, \phi_{1}, \phi_{2}, \ell)$ and ${\mathcal{I}}_{s_{\psi,\ell}}$ be the minimal interval covering consistent with $\{s_{\phi_{1},\ell'}$, $s_{\phi_{2},\ell'}\}_{\ell' \in L^{\ell}_{[0,w_{2}]}}$, then, for all $I_{i}\in {\mathcal{I}}_{s_{\psi,\ell}}$
$$s_{\psi, \ell}(I_{i}) = 1 \iff (\vec x, t, \ell) \models \phi_{1}\surround{[d_{1},d_{2}]}\phi_{2} \quad \forall t \in I_{i}$$
\label{thm:Booleansurr}
\end{restatable}

%
%

\begin{algorithm}[tbp]
\caption {Boolean monitoring for the surrounded operator}
\label {bspuntil}
\begin{algorithmic}[1]
\STATE {\bf input} $\ell, \psi= \varphi_{1} \mathcal{S}_{[d_{1},d_{2}]} \varphi_{2}$
\STATE{$\forall \ell' \in L^{\ell}_{[0,d_{2}]}$ compute $s_{\varphi_{1},\ell'},s_{\varphi_{2},\ell'}$}

\STATE compute $\mathcal{I}_{s_{\psi, \ell}}$
\COMMENT {{\small the minimal interval covering consistent with $s_{\varphi_{1},\ell'},s_{\varphi_{2},\ell'},$\quad$\ell' \in L^{\ell}_{[0,d_{2}]}$}}
	\FORALL{$ I_{i} \in \mathcal{I}_{s_{\psi, \ell}}$}
	\STATE{ $V= \{ \ell' \in  L^{\ell}_{[0,d_{2}]} | s_{\varphi_{1},\ell'}(I_{i})=1 \}$} 
	\STATE{ $Q= \{ \ell' \in  L^{\ell}_{[d_{1},d_{2}]} | s_{\varphi_{2},\ell'}(I_{i})=1 \}$} 
	\STATE{ $W= B^{+} (Q \bigcup V) $} 
		\WHILE{$W \not= \emptyset$} 
			\STATE{$W'= \emptyset$}
			\FORALL{$\ell \in W$} 
				\STATE {$N= pre(\ell) \bigcap V =\{  \ell' \in V | \ell E \ell' \} $} 
				\STATE{$V=V \backslash N$}	
				\STATE{$W'=W' \bigcup (N \backslash Q)$}	
			\ENDFOR
			\STATE{$W= W'$}  
	 	\ENDWHILE 
	\STATE{$s_{\psi,\ell} (I_{i})= \begin{cases}
      1 & \text{ if } \ell \in V, \\
      0 & \text{otherwise}.
\end{cases}$}	   
\ENDFOR
\STATE merge adjacent positive intervals $I_{i}$, i.e., $I_{i}$ s.t. $s_{\psi,\ell} (I_{i})=1$
\RETURN{$s_{\psi, \ell}$}
\end{algorithmic}
\label{algo:Boolean}
\end{algorithm}

\subsection{Quantitative Monitoring}


We now turn to the monitoring algorithm for the quantitative semantics. As the input signals are functions of continuous time, we need to make some assumptions to represent them in a finite way. Contrary to other approaches, like \cite{Donze2010}, which assume signals to be piecewise linear, we make a simpler assumption, namely that they are piecewise constant functions. More specifically, we discretise time with a step $h$, and consider a piecewise constant representation of a signal $\vec{x}$, assuming its value between time $t_k = kh$ and $t_{k+1} = (k+1)h$ being equal to $\vec{x}(kh)$. We will show how to compute the quantitative semantics for such an approximation, and discuss the error we introduce in case the input signal is Lipschitz continuous. 

Monitoring Boolean operators is straightforward, we just need to apply the definition of the quantitative semantics pointwise in the discretisation, i.e. at each time step $kh$. The time bounded until operator can also be easily computed by replacing the min and max over dense real intervals in its definition by the corresponding min and max over the corresponding finite grid of time points. In this case, however, an error is introduced, due to the discrete approximation of the Lipschitz continuous signal in intermediate points. The error accumulates at a rate proportional to $M h$, where $M$ is the Lipschitz constant of the signal $\vec{x}$, formally defined in Proposition \ref{prop:error}.

Monitoring the somewhere operator $\somewhere{[d_1,d_2]}\varphi$ is also immediate: once the location $ \ell$ of interest is fixed, similarly to the Boolean semantics, we can just turn it into a disjunction of the signals $s_{\varphi,\ell'}$ for each location $\ell' \in L^{\ell}_{[d_{1},d_{2}]}$.
 
The only non-trivial monitoring algorithm  is the one for the spatial surrounded operator, which we discuss below. However, as the satisfaction score is  computed at each time point of the discretisation and depends on the values of the signals at that time point only, this algorithm introduces no further error w.r.t. the time discretisation. 
Hence, we can globally bound the error introduced by the time discretisation (see the Appendix for the proof):
\begin{restatable}{prop}{restError}
Let the primary signal $\vec x$ be Lipschitz continuous, as well as the functions defining the atomic predicates. Let $M$ be a Lipschitz constant for all secondary signals, and $h$ be the discretisation step.  Given a SSTL formula $\varphi$, let $u(\varphi)$ count the number of temporal until operators in $\varphi$, and denote by $\rho(\varphi,\vec x)$ its satisfaction score over the  trace $\vec x$ and by $\rho(\varphi,\vec{\hat x})$ the satisfaction score over the discretised version $\vec{\hat{x}}$ of $\vec x$ with time step $h$. Then 
$ \| \rho(\varphi,\vec x) - \rho(\varphi,\vec{\hat x}) \| \leq u(\varphi) M h $.
\label{prop:error}
\end{restatable}

The quantitative monitoring procedure for the bounded surrounded operator is shown in Algorithm \ref{algo:quantitative}.
Similarly to the Boolean case, the algorithm for the surrounded formula $\psi=\varphi_{1} \surround{[d_{1},d_{2}]} \varphi_{2}$ takes as input a 
location $\ell$ and returns the quantitative signal $s_{\psi,\ell}$, or better its piecewise constant approximation with time-step $h$ (an additional input, together with the signal duration $T=mh$). 
As a first step, it computes recursively the quantitative satisfaction signals of subformula $\varphi_1$,  for all locations $\ell' \in L^{\ell}_{[0,d_{2}]}$, and of subformula $\varphi_2$,  for all locations $\ell' \in L^{\ell}_{[d_{1}, d_{2}]}$. Furthermore, it sets all the quantitative signals for $\varphi_1$ and $\varphi_2$ for the other locations to the constant signal equal to minus infinity.
The algorithm runs a fixpoint computation for each time instant in the discrete time set $\{0,h,2h,\ldots,mh\}$. The procedure is based on computing a function $\mathcal{X}$, with values in the extended reals $\mathbb{R}^*$, which is executed on the whole set of locations $L$, but for the modified signals equal to $-\infty$ for locations not satisfying the metric bounds for $\ell$. The function $\mathcal{X}$ is defined below.

\begin{defi}
\label{defX}
Given a finite set of locations $L$ and two functions $s_{1}: L \rightarrow \mathbb{R}^{*}, s_{2}: L \rightarrow \mathbb{R}^{*}$. The function $\mathcal{X}: \mathbb{N} \times L \rightarrow \mathbb{R}$ is inductively defined as: 
\newpage
\begin{enumerate}
\item $\mathcal{X}(0,\ell)=s_{1}(\ell)$
\item $\mathcal{X} (i+1, \ell)=\min(\mathcal{X}(i, \ell), \min_{\ell'|\ell E \ell'}(\max(\mathcal{X}(i,\ell'),s_{2}(\ell'))))$
\end{enumerate}
\end{defi}
The algorithm then computes the function $\mathcal{X}$ iteratively, until a fixed-point is reached. 
\begin{restatable}{thm}{restFixedpointQuant}
\label{thm:fixedpoint_quant}
Let $s_{1}$ and $s_{2}$ be as in Definition \ref{defX}, and 

$$s(\ell)=\max_{A \subseteq L, \ell \in A}{( \min (\min_{\ell' \in A}s_{1}( \ell'),\min_{ \ell' \in B^{+}(A)}s_{2}( \ell')))},$$
then $$\lim_{i \rightarrow \infty} \mathcal{X}(i, \ell) = s(\ell), \qquad \forall \ell \in L.$$
 Moreover, $\exists K>0$ s. t. $\mathcal{X}(j, \ell) = s(\ell), \forall j\geq K$.
\end{restatable}

\begin{algorithm}[bp] 
\caption{Quantitative monitoring for the surrounded operator}
\label {sur_quant_alg}
\vspace{1mm}
\begin{algorithmic}[1]
\STATE inputs: $\ell, \psi=\varphi_{1} \mathcal{S}_{[d_{1},d_{2}]} \varphi_{2}$ , $h$, $T$ 
\FORALL {$\ell' \in L$}
	\STATE \qquad \textbf{if} $ 0\leq d(\ell,\ell')\leq d_{2}$ \textbf{then}
	\STATE \qquad \qquad compute $s_{\varphi_{1},\ell'}$
	\STATE \qquad \qquad  \textbf{if} $ d(\ell,\ell') \geq d_{1}$ \textbf{then} 
	\STATE \qquad \qquad \qquad  compute $s_{\varphi_{2},\ell'}$ 
	\STATE \qquad \qquad \textbf{else} $s_{\varphi_{2},\ell'} = -\infty$
	\STATE \qquad  \textbf{else} $s_{\varphi_{1},\ell'} = -\infty, s_{\varphi_{2},\ell'} = -\infty$
\ENDFOR
	\FORALL { $t \in \{ 0, h, 2h, \ldots,T\}$}
	 	\FORALL {$\ell' \in L$}
			\STATE \qquad $\mathcal{X}_{prec}(\ell')= +\infty$
			\STATE \qquad $\mathcal{X}(\ell')= s_{\varphi_{1},\ell} (t)$ 
		\ENDFOR
		\WHILE {$\exists \ell' \in L, \mbox{ s.t. } \mathcal{X}_{prec}(\ell') \not= \mathcal{X}(\ell')$} 
		     \STATE $\mathcal{X}_{prec} =\mathcal{X}$
			\FORALL {$\ell' \in L$}
				\STATE $\mathcal{X}(\ell')=\min(\mathcal{X}_{prec}(\ell'), \min_{\ell''|\ell' E \ell''}(\max(s_{\varphi_{2},\ell''}(t),\mathcal{X}_{prec}(\ell''))))$
			\ENDFOR
							
		\ENDWHILE
		\STATE  $s_{\psi,\ell}(t)=\mathcal{X}(\ell)$  
 	\ENDFOR
\RETURN $s_{\psi,\ell}$

\end{algorithmic}
\label{algo:quantitative}
\end{algorithm} 

The following corollary provides the correctness of the method. It shows that, when $\mathcal{X}$ is computed for the modified signals constructed by the algorithm, it returns effectively the quantitative satisfaction score of the spatial surrounded operator. 

\begin{restatable}{cor}{restRhomin}
Given an $\hat{\ell} \in L$, let $\psi=\varphi_{1} \mathcal{S}_{[d_{1},d _{2}]} \varphi_{2}$ and\\
 \begin{align*}
 & s_1(\ell)= \begin{cases}
       \rho(\phi_{1}, \vec x, t, \ell) & \text{ if }  \mbox{ } 0\leq d(\hat{\ell},\ell)\leq d_{2} \\
        -\infty &   \text{ otherwise}.
       \end{cases}  \\
 & s_2(\ell)= \begin{cases}
       \rho(\phi_{2}, \vec x, t, \ell) & \text{ if }  \mbox{ } d_{1} \leq  d(\hat{\ell},\ell) \leq d_{2} \\
       -\infty & \text{ otherwise}.
      \end{cases}
\end{align*}
 Then $\rho(\psi,  \vec x, t  ,\hat{\ell}) = s(\hat{\ell})=\max_{A \subseteq L, \hat{\ell} \in A}{( \min (\min_{\ell \in A}s_{1}( \ell),\min_{ \ell \in B^{+}(A)}s_{2}( \ell))).}$
 \label{rhomin}
\end{restatable}

In order to discuss the complexity of the monitoring procedure, we need an upper bound on the number of iterations of the algorithm. This is given by the following.
\begin{restatable}{prop}{restDiameterconv}
\label{diameterconv}
Let   $d_G$ be the diameter of the graph $G$ and $\mathcal{X} (\ell)$ the fixed point of $\mathcal{X}(i,\ell)$, then $\mathcal{X} (\ell)=\mathcal{X} (d_G+1, \ell)$ for all $\ell \in L$.
\label{prop:complexity}
\end{restatable}

It follows that the computational cost for each {\em single} location is $O(d_G |L| m)$, where $m$ is the number of sampled time-points.   The cost for {\em all} locations is therefore $O(d_G |L|^2 m)$.

The proofs of Theorem~\ref{thm:fixedpoint_quant}, Corollary~\ref{rhomin}  and Proposition~\ref{diameterconv} are reported in Appendix~\ref{app:quantproof}.


%% file: implementation.tex

\subsection{SSTL Monitoring Implementation}
\label{sec:inmplementation}

To support qualitative and quantitative monitoring of SSTL properties, a prototype tool has been developed. This tool, developed in \textsf{Java}, consists of a \textsf{Java} library (\textsf{jSSTL} API) and a front-end, integrated in ECLIPSE. Both the library and the ECLIPSE plugin can be downloaded from \url{https://github.com/Quanticol/jsstl}. All the scenarios considered in this paper are available at \url{https://github.com/Quanticol/jsstl-examples} to permit the replication of experiments.

The library can be used  to integrate \textsf{jSSTL} within other applications and tools, whereas the ECLIPSE plugin provides a user friendly interface to the tool. Furthermore, the modular approach of the implementation allows us to develop different front-ends for \textsf{jSSTL}. The tool has been described in more detail in \cite{jsstl16}.
   
The ECLIPSE plugin provides a simple user interface to specify and verify SSTL properties of spatio-temporal trajectories generated from the simulation of a system or from real observations.  We can specify the properties, describe a model of the space (i.e., its  graph structure), import the trajectories and then verify whether such trajectories satisfy the specified properties.

In Figure~\ref{fig:pluginEclipse}, a snapshot is shown of the ECLIPSE plugin.
It provides an {\it editor} for \textsf{jSSTL}, containing the script with the SSTL properties that we want to analyse in our scenario (on the left)
and a {\it view} to visualise the space model, the data and the results of the analyses (on the right). 

\begin{figure} [htbp] 
\hspace{-2mm}
\includegraphics[width=0.95\textwidth]{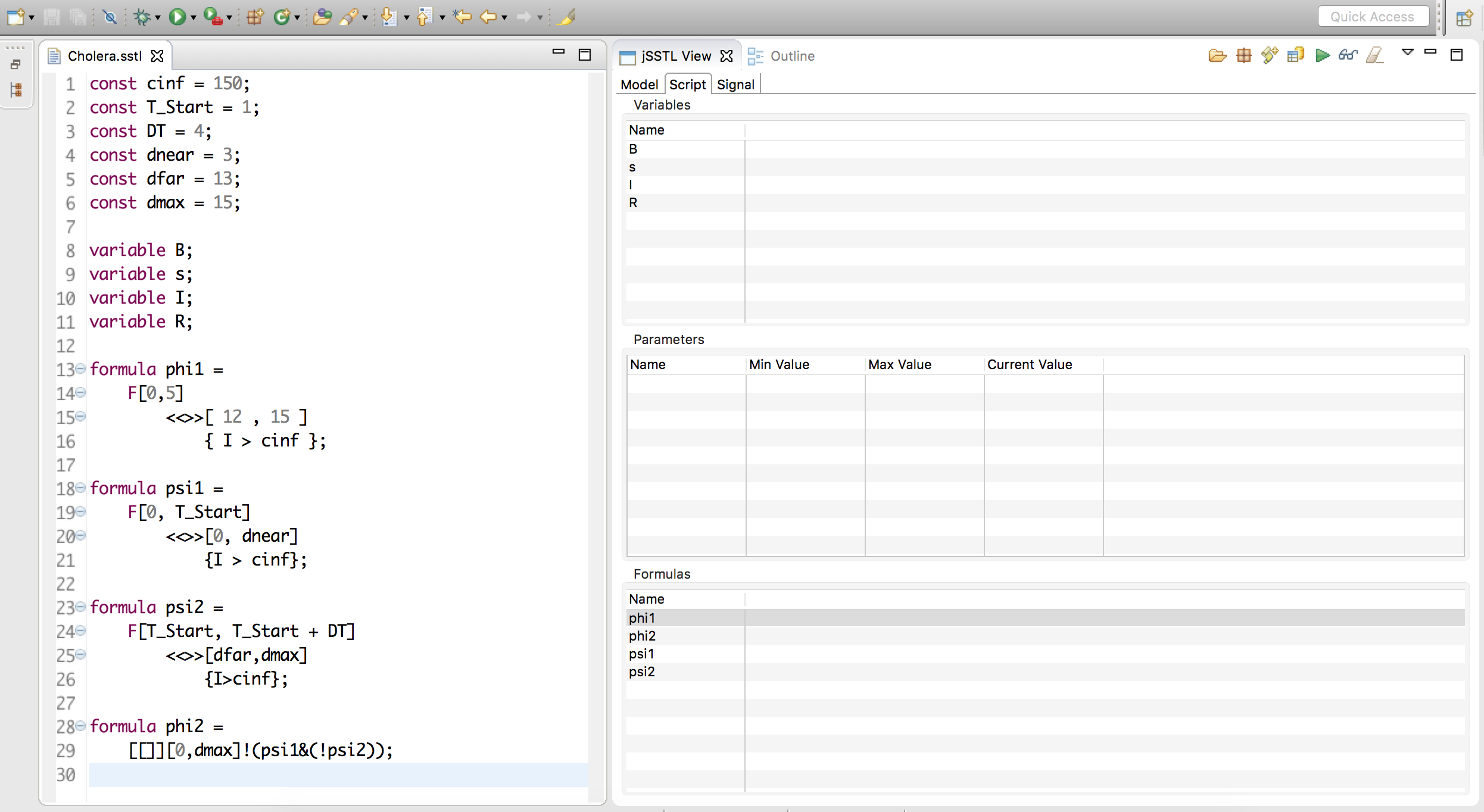}
\caption{The \textsf{jSSTL} ECLIPSE plugin}
\label{fig:pluginEclipse}
\end{figure}

\vspace{1mm}

\subsection*{Implementation performance.}
Model (\ref{model}) has been coded in Matlab/Octave, and the monitoring has been performed by our Java implementation. 
As time performance, 
the verification of property $\phi_{\mathrm{spot}}$ took $1.04s$ (Boolean) and $69.39s$ (quantitative) for all locations and 100 time points, while property $\phi_{\mathrm{pattern}}$ took $1.81s$ and $70.06s$, and property $\phi_{\mathrm{pert}}$ took $28,19s$ and $55,31s$, respectively. The computation of the distance matrix can be done just once because it remains always the same for a given system, this takes about $23s$.
 All the experiments were run on an 
 Intel Core i5  2.6 GHz CPU, with 8GB 1600 MHz RAM.


%% file: stoch.tex
\section{SSTL Analysis of Stochastic Spatio-temporal Systems}
\label{sec:stochAnalysis}
 

SSTL can also be applied to describe properties of stochastic spatio-temporal systems. Stochastic models describe the evolution of systems in space and time that show noisy behaviour, due to internal random mechanisms (like in epidemic spreading models), or environmental effects (typically captured with Stochastic Differential Equations, SDEs). Independently from the mathematical device to describe them, stochastic models induce a probability measure on the space of all possible traces (i.e. on the so-called Skorokhod space, the space of c\`adl\`ag functions, which are piecewise continuous functions from time taking real values). Each SSTL formula $\phi$ describes a subset of these traces, those satisfying it, which is measurable (with respect to the topology of the Skorokhod space), as can be proved by a simple adaptation of the proof in \cite{TCS2015}, owing to the discrete nature of space considered here. Measurability implies that we can calculate the probability $p(\phi)$ of this set, which is known as the satisfaction probability of $\phi$ for the stochastic model considered.  The goal for analysing these systems is thus to compute such satisfaction probability.  Due to computational unfeasibility of exact methods \cite{barbot2011}, the  mainstream approach in verification is to rely on Statistical Model Checking \cite{Younes2004}, which combines simulation of the stochastic model (i.e. an algorithm that samples traces according to the probability distribution of the model in the Skorokhod space) with a monitoring routine for the property $\phi$. More specifically, every time a trace is generated by the simulator, it is passed to the Boolean monitor, which returns either 0 (false) or 1 (true). Probabilistically, this can be seen as a sample of a Bernoulli random variable, having probability $p(\phi)$ of observing 1. From a finite sample of such values, we can rely on standard statistical tools to estimate $p(\phi)$ and to compute the confidence level of such an estimate. A scheme of Statistical Model Checking (SMC) is shown in Figure~\ref{fig:SMC}. 
Examples of spatio-temporal model checking to compute the approximated probabilistic satisfaction can be found in~\cite{bartocci2015, CLMPV16}

%


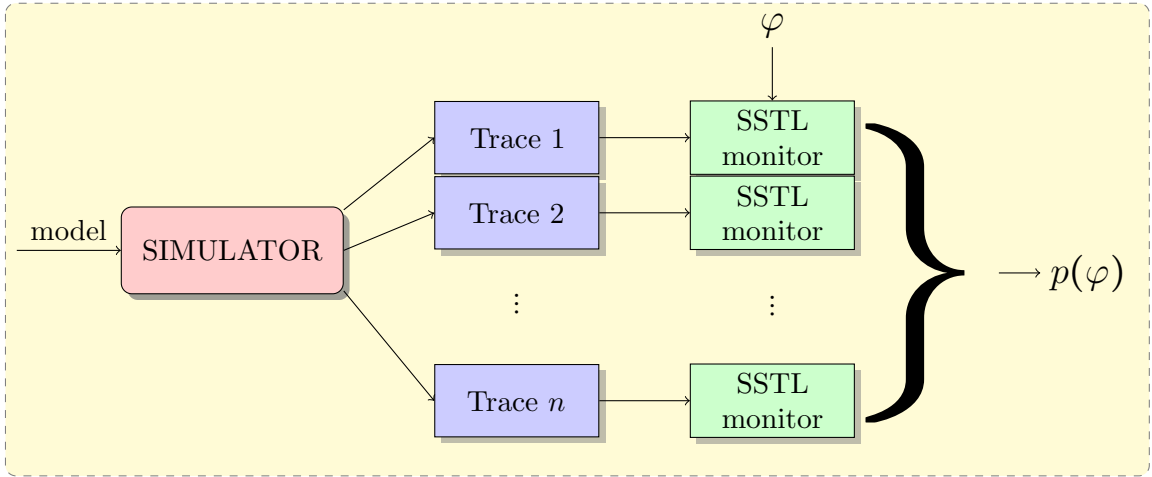
\begin{figure}[h]
\centering
\pgfdeclarelayer{background}
\pgfdeclarelayer{foreground}
\pgfsetlayers{background,main,foreground}
\tikzstyle{trace}=[draw, fill=blue!20, text width=5em, 
    text centered, minimum height=2.5em,drop shadow]
    \tikzstyle{monitor}=[draw, fill=green!20, text width=5em, 
    text centered, minimum height=2.5em,drop shadow]
\tikzstyle{ann} = [above, text width=5em, text centered]
\tikzstyle{wa} = [trace, text width=7em, fill=red!20, 
    minimum height=3em, rounded corners, drop shadow]
\tikzstyle{sc} = [trace, text width=13em, fill=red!20, 
    minimum height=10em, rounded corners, drop shadow]
    \tikzstyle{sc} = [monitor, text width=13em, fill=red!20, 
    minimum height=10em, rounded corners, drop shadow]
\def\blockdist{2.3}
\def\edgedist{2.5}
\begin{tikzpicture}
    \node (wa) [wa]  {SIMULATOR};
    \path (wa.east)+(\blockdist,1.5) node (asr1) [trace] {Trace $1$};
    \path (asr1.east)+(\blockdist,0) node (mon1) [monitor] {SSTL monitor};
       \path (mon1.north)+(0,1) node (phi) {{\fontsize{0.5cm}{0.5cm}\selectfont $\phi$}};
    \path (wa.east)+(\blockdist,0.5) node (asr2)[trace] {Trace $2$};
    \path (asr2.east)+(\blockdist,0) node (mon2) [monitor] {SSTL monitor};
    \path (wa.east)+(\blockdist,-1.0) node (dots)[ann] {$\vdots$}; 
      \path (dots.east)+(\blockdist,-0.3) node (dots2)[ann] {$\vdots$}; 
    \path (wa.east)+(\blockdist,-2.0) node (asr3)[trace] {Trace $n$};    
    \path (asr3.east)+(\blockdist,0) node (mon3) [monitor] {SSTL monitor};
       \path (mon2.east)+(0.8,-0.8) node (vote)  { {\fontsize{4cm}{4cm}\selectfont $\}$}};
    \path (vote)+(\blockdist,0) node (result)  {{\fontsize{0.5cm}{0.5cm}\selectfont $p(\phi)$}};

    \path [draw, ->] (asr1)-- node [above] {}   (mon1)  ;
        \path [draw, ->] (asr2)-- node [above] {}   (mon2)  ;
            \path [draw, ->] (asr3)-- node [above] {}   (mon3)  ;
            
    \path [draw, ->] (wa.0) -- node [above] {}  (asr2.west);
    \path [draw, ->] (wa.-20) -- node [above] {}  (asr3.west) ;
    \path [draw, ->]   (vote)-- node [above] {} (result.west) ;   
      \path [draw, ->] (phi)-- node [above] {}   (mon1)  ;
          \path [draw, ->] (wa.20)-- node [above] {}   (asr1.west)  ;
      
    \begin{pgfonlayer}{background}
        \path (asr1.west |- asr1.north)+(-5.7,1.3) node (a) {};
        \path (vote.east |- wa.east)+(2,-3) node (c) {};
          
        \path[fill=yellow!20,rounded corners, draw=black!50, dashed]
            (a) rectangle (c);           
        \path (asr1.north west)+(-0.2,0.2) node (a) {};
            
              \path (-3,0) node (xx) {};
            \path [draw, ->] (xx) -- node [above] {model} (wa) ;
    \end{pgfonlayer}

\end{tikzpicture}
\caption{Scheme of Statistical Model Checking (SMC).}
\label{fig:SMC}
\end{figure}

%

%
 
The integration of SSTL monitoring with statistical model checking opens the possibility of checking spatio-temporal properties of stochastic models, providing a powerful tool to explore complex stochastic behaviours in space and time. Moreover, a similar mechanism can be put in place for the quantitative semantics. In this case, SMC enables one to estimate the  average robustness of a formula (and other moments, like the variance). The combination with SMC and the quantitative semantics has been explored earlier for STL  in~\cite{TCS2015} and applied to tasks like system design and parameter synthesis \cite{TCS2015,silvetti16,silvetti18}. The spatio-temporal logics SpaTeL and STLCS have also been used in conjunction with statistical spatio-temporal model checking for the Boolean semantics; see~\cite{bartocci2015} and in~\cite{CLMPV16}, respectively, the latter focusing on usability issues in bike-sharing systems.


In the following, we give a taste of the use of SSTL to analyse stochastic models, considering a variant of the  running example, in which the pattern formation mechanism is subject to external random perturbations, which is a more realistic scenario that the one described by a deterministic ODE model.    An additional more complex example can be found in the next section.

\begin{exa}[{\bf Persistence of External Perturbations}]
We consider the effects of external perturbations of the Turing pattern formation system, adding a white Gaussian noise to the set of equations (\ref{model}).  In particular, we add a random fluctuation $\eta( \vec z, t)$, with zero-mean and covariance matrix
$<\eta(\vec z,t),\eta(\vec{z}',t)> = \epsilon^2 \delta z^2 \delta t$, where $\epsilon$ is the noise intensity and $\vec z = (i,j) \in [0, 32]^2$ corresponds to the position.
The methodology follows \cite{lesmes_noise-controlled_2003,leppanen_effect_2003}, and results in a set of Stochastic Differential Equations. In the discretisation, we set $\Delta z = 1$ and $\Delta t =0.01$.
We study how the noise affects the dynamics for fixed deterministic 
parameters, analysing in detail the capability of the system to maintain a specific pattern with respect to the external perturbation. 
To this aim, we fix the parameters of the system:  $K=32, R_{1}=1, R_{2}=-12, R_{3}=-1, R_{4}= 16, D_{1} = 5.6$ and $D_{2}= 25.5$ and we 
evaluate formula \ref{phi:pattern} (the pattern formula) with parameters
 $h=0.5, T_{\mathrm{pattern}}= 19, \delta = 1, T_{\mathrm{end}} =30,  w_{1}=1, w_{2}=6$ for different values of the noise intensity $\epsilon$.

In Figure~\ref{fig:prob},
we show how the satisfaction probability decreases as a function of the noise intensity $\epsilon$; $\epsilon$ was varied between $0$ and $0.9$ in steps of 0.1 units.  We estimate the probability statistically from $10,000$ runs for each parameter value. This result shows that an increasing intensity of noise prevents the system from establishing a regular pattern of spots of the form we have seen in a non-perturbed variant of the system.
%
In Figure~\ref{fig:probvsRob}, 
 we plot the satisfaction probability versus the average 
robustness degree, estimating them statistically from $10,000$ runs for each parameter value. The satisfaction probability varies from  1 to 0 while the average robustness score varies from $-80.01  \pm 10.7$ to $0.008  \pm 0$.
As we can see, these two quantities seem to be correlated. In other words, the higher is the number of runs in which we find that the pattern formation property holds, the higher is the average robustness with which this happens.
By varying the threshold, the Pearson's correlation coefficient 
between satisfaction probability and robustness degree is $0.784$.
%
  \begin{figure}[!t]
\begin{center}
\subfigure[]{
\label{fig:prob}
\includegraphics[width=.465 \textwidth]{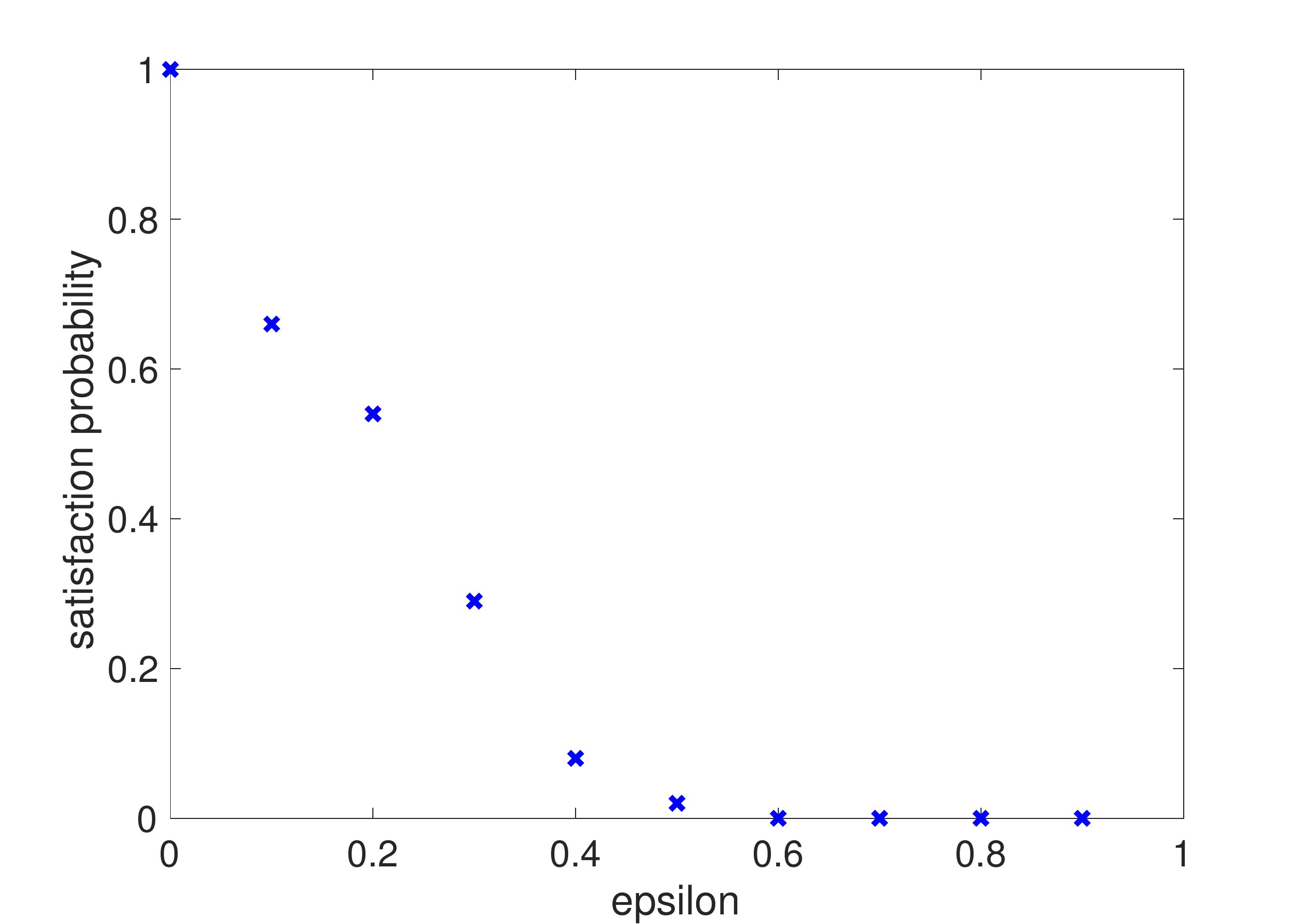}
}
\subfigure[]{
\label{fig:probvsRob}
\includegraphics[width=.44 \textwidth]{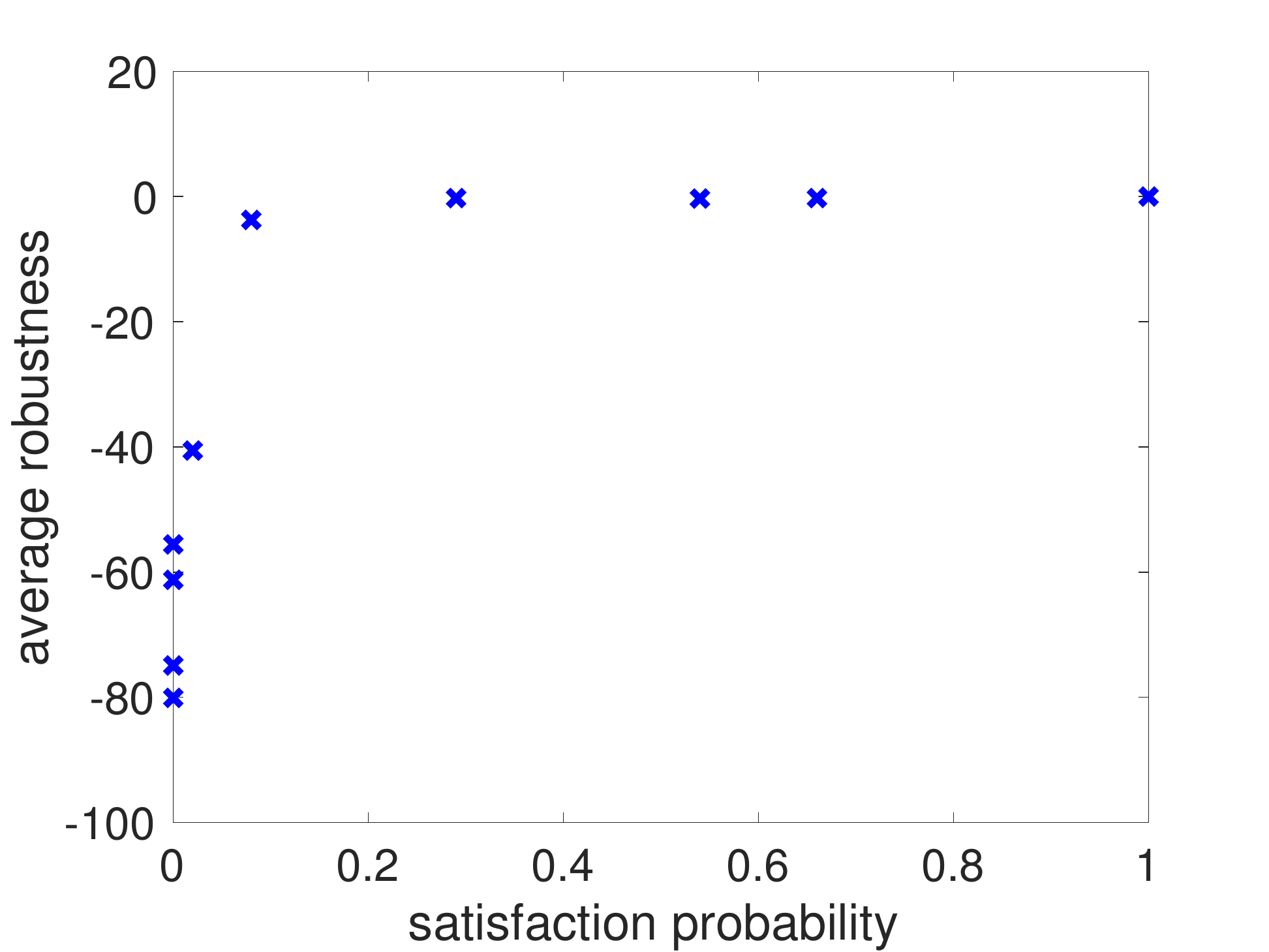}
}
\end{center}
\caption{ (a) Satisfaction probability vs. noise intensity; (b) Satisfaction probability vs. average robustness degree.}
\label{fig:}
\end{figure}
 %
\end{exa}


%% file: caseStudy2.tex

\section{Case study: bike sharing system}
\label{sec:bikeSharing}

In this section, we present an analysis of the London Santander Cycles Hire scheme, a bike sharing system, modelled  as a {\it Population Continuous Time Markov Chain} (PCTMC)\cite{tutorial} with time-dependent rates. We use SSTL to study a number of spatio-temporal properties of the system and to explore their robustness considering a set of parameter values for the formulas. 

\subsection{Model} \label{sec:model}
The Bike-Sharing System (BSS) is composed of a number of bike stations, distributed over a geographic area. Each station has a fixed number of bike slots. The users can pick up a bike, use it for a while, and then return it to another station in the area. Following~\cite{Feng:2016}, we model the BSS as a Population Continuous Time Markov Chain (PCTMC) with time-dependent rates, leading to a hybrid system. The model, given the bike availability in a station at time t, predicts the probability distribution of the number of available bikes in that station at time $t + h$ with $h \in [0, 40]$ minutes. The parameters of the model have been set using the historic journey data and bike availability data from January 2015 to March 2015 from the London Santander Cycles Hire scheme. In detail, we can describe the model through the following transitions:
\begin{align*}
\label{}
&B_i \rightarrow  \mbox{ } S_i &\mbox{ at rate } out_i(t) \cdot p^i_{out}, &\qquad \forall i \in (1,N)\\
&S_i \rightarrow  B_i  & \mbox{ at rate } in_i(t) \cdot \lambda^i_{in}, &\qquad \forall j \in (1,N) \\
&B_i \rightarrow  \mbox{ } S_i + T^i_j &\mbox{ at rate } out_i(t)\cdot p^i_j(t), &\qquad \forall i, j \in (1,N)\\
&S_j + T^i_j   \rightarrow B_j  & \mbox{ at rate } in^i_j( \# T^i_j ), &\qquad \forall i, j \in (1,N).
\end{align*}
Here,  $B_i$ (respectively $S_i$) represents the bike agent (respectively the slot agent) in the $i^{th}$ station, $T^i_j$ is the bike agent travelling from pick-up station $i$ to return station $j$. To these different agents we associate counts, e.g. $\# T^i_j$ denotes the population size of an agent type $T^i_j$, while $N$ is the total number of stations.
Each transition in the list describes a possible event changing the state of the system. In general,  a transition rule like $B_i \rightarrow  S_i + T^i_j$ models that an agent $B_i$ is removed from the system (a bike leaves station $i$), while new agents  $S_i$ and $T^i_j$ are added to the system (a free slot is added to station $i$, while the bike is set to travel towards station $j$). This is reflected also in the population counts: upon the firing of such a transition, $\# B_i$ will be decreased by 1, while $\# S_i$ and $\# T^i_j$  are increased by 1.

In the list of transitions above, the first transition represents the pick-up of a bike from a station $i$ and its ``removal'' from the system. A bike is considered removed from the system in two cases: if its destination is a station outside the model (in case we model a subset of stations) or if the bike is taken for a time which is much longer than the usual travel time between two stations, which is about 15 minutes. This can happen, for example, if the bike is taken for a day trip. Similarly, the second transition models the return to station $i$ of a previously removed bike. The last two transitions model the pick-up and return of a bike by users of the system that are travelling between the bike stations modelled in the system.
Each transition has a rate, encoding the average frequency with which it happens, based on historic journey data. These rates, together with the updates in population counts, define a stochastic model in terms of a Continuous Time Markov Chain, which can be simulated with standard stochastic simulation algorithms.  We refer the interested reader e.g. to \cite{tutorial} for further details.  
More specifically, $out_i(t)$  (respectively $in_i(t)$) is the bike pick-up (respectively return) rate in station $i$ at time $t$,  $\lambda^i_{out}(t)$ (respectively $\lambda^i_{in}(t$) is the probability that a bike leaves (respectively enters from) outside the system, $p^i_j(t)$ is the probability that a journey will end at station $j$ given that it started from station $i$ at time $t$, $in^i_j$, instead, is the return rate of a bike pick-up in station $i$ that will be returned in station $j$.
The model considered is very large comprising 733 bike stations (each with 20-40 slots) and a total population of 57,713 agents (users) picking up and returning bikes. 
In our model the first two transitions have very small rates. They are more significant in the analysis of submodels where just a subsystem of BSS stations are considered, i.e. where there are considered stations outside the system.

 \subsection{Spatio-temporal analysis of the Bike Sharing Systems using SSTL} 
%
We simulate the model using Simhya~\cite{Bort_Qapl2012}, a Java tool for the simulation of stochastic and hybrid systems. In particular we exploit its Gibson-Bruck (GB) algorithm. Following~\cite{Feng:2016}, the time-dependent rates change 2 times, all at the same time unit: the first interval is 14 minutes, the second interval is 20 minutes, the third interval is 6 minutes; we simulate the model for 40 minutes. Then, we consider the trajectories only of the bike (B) and slot (S) agents, in each station. Our spatio-temporal trace is then $X_B (t,\ell) = (X_B(t,\ell), X_S(t,\ell))$, i.e.  the number of bikes and free slots at each time, in each station.
 The space is represented by a weighted graph, where the nodes are the stations and the edges describe the connection between each station. Two nodes are connected is they are at a distance less or equal to $1$ Kilometer.
The weighted function $w: E \rightarrow \mathbb{R}$ returns the distance in kilometres between any two stations, where $E=L \times L$ is the set of edges and $L$ is the set of stations.
 
One of the main problems of these systems consists in the availability of bikes or free slots in each station. Two important questions related to this issue from a user's point of view are: 
 \begin{itemize}
 \item if I don't find a bike (free slot, resp.) is there another station at the distance less than a certain value where I can find a bike (free slot, resp.) slot? 
 \item if I don't find a bike (free slot, resp.), how long should I wait before another user returns a bike or picks one up, respectively?
 \end{itemize}
 These concerns can be expressed by the SSTL properties described below. The first one is given by:
 \begin{equation}
 \phi_{1}= \glob{[0,T_{end}]}\{\somewhere{[0,d]} (B>0) \wedge \somewhere{[0,d]} (S > 0)\}
\label{formuladistance}
\end{equation}
A station $\ell$ satisfies $\phi_1$ if and only if it is always true that, between 0 and $T_{end}$ minutes, there exists a station at a distance less than or equal to $d$, that has at least one bike and a station at a distance less or equal to $d$ that has at least one free slot. 

In the analysis, we explore the value of parameter $d \in [0, 1]$ kilometres to see how the satisfaction of the property changes in each location.
Figure~\ref{phi1all} shows the approximate satisfaction probability $p_{\phi_1}$ for 1000 runs for all the stations, for  (a) $d = 0$, (b) $d = 0.2$ , (c) $d = 0.3$ and (d) $d = 0.5$ km. We can see that for $d=0$ many stations have a high probability to be full or empty (indicated by red points), with standard deviation of all the locations in the range [0,   0.0158] and mean standard deviation  0.0053.
 However, increasing $d$ to $d=0.2$ km, i.e. allowing a search area of up to 200 metres from the station that currently has no bikes, or no slots resp., we greatly increase the satisfaction probability of $\phi_1$, with a standard deviation that remains in the same range and mean standard deviation of 0.0039. For $d=0.5$, the probability of $p_{\phi_1}$ is greater than 0.5 for all the stations; standard deviation is in the range [0, 0.0142] and mean stdv is 0.0002. Figure~\ref{singleloc} (a) shows the satisfaction probability of some BBS stations vs distance d=[0,1.0].
\begin{figure}[!t]
\begin{center}
\subfigure[]{
\label{phi1alld0}
\includegraphics[width=.465 \textwidth]{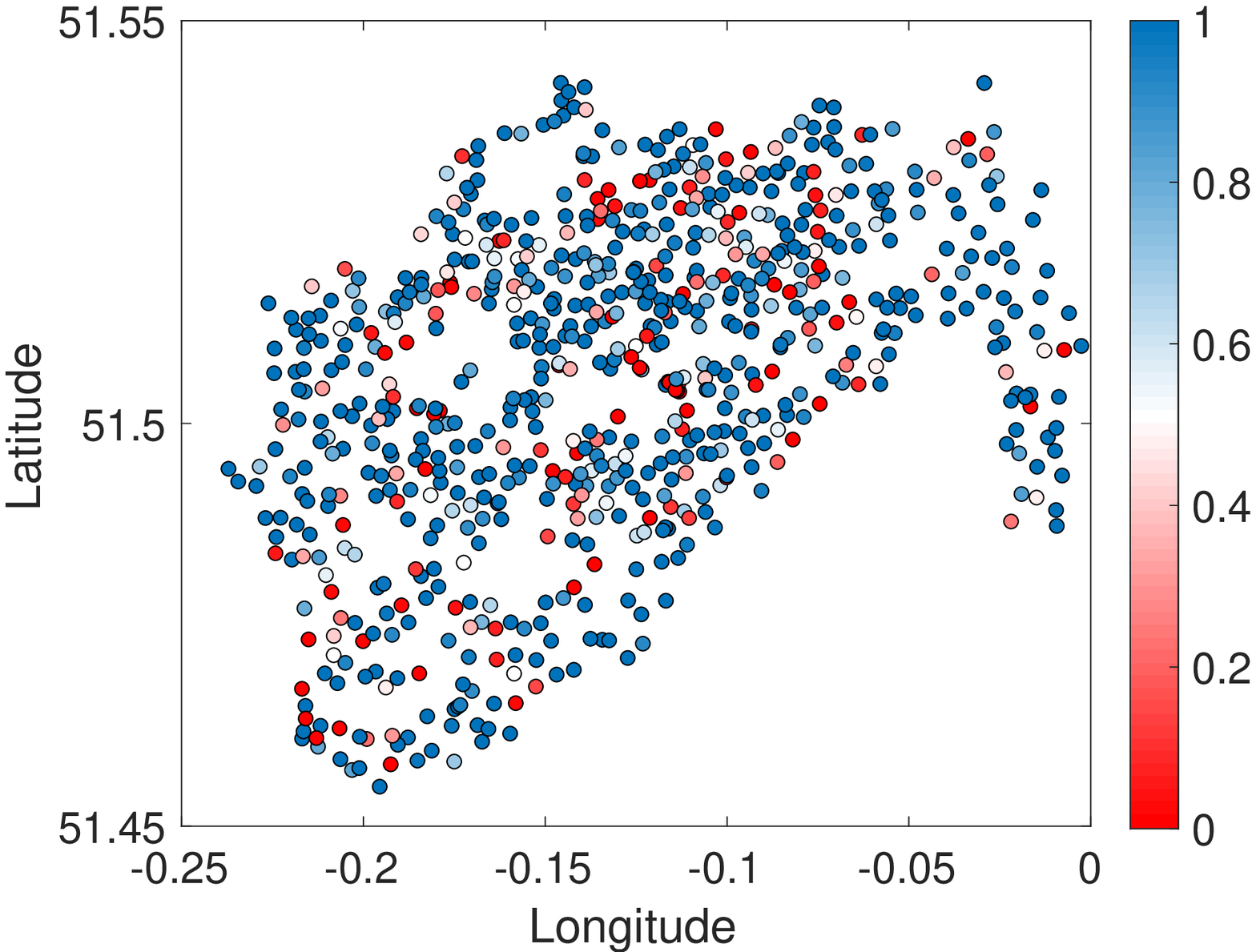}
}
\subfigure[]{
\label{phi1alld200}
\includegraphics[width=.465 \textwidth]{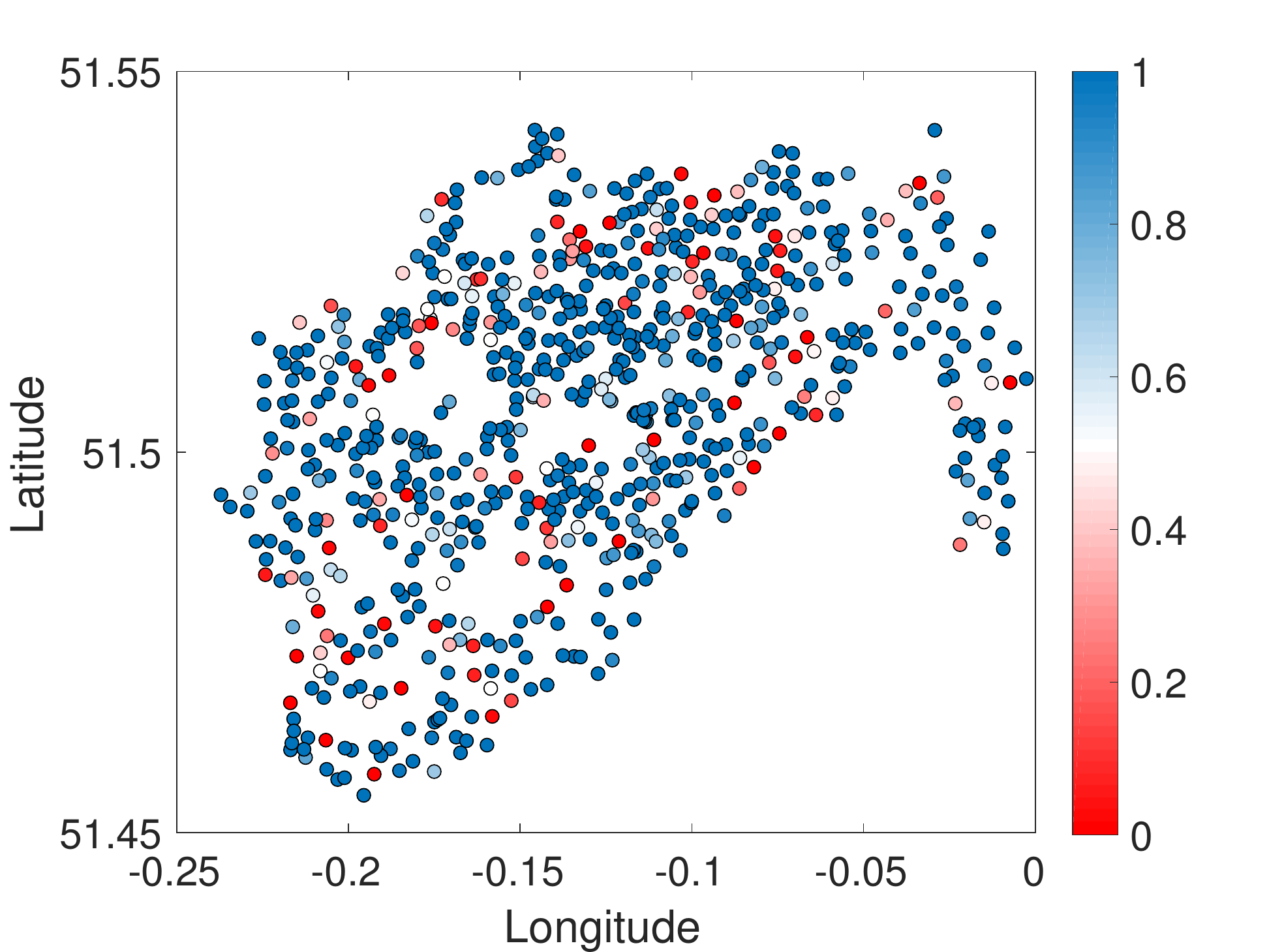}
}
\subfigure[]{
\label{phi1alld300}
\includegraphics[width=.465 \textwidth]{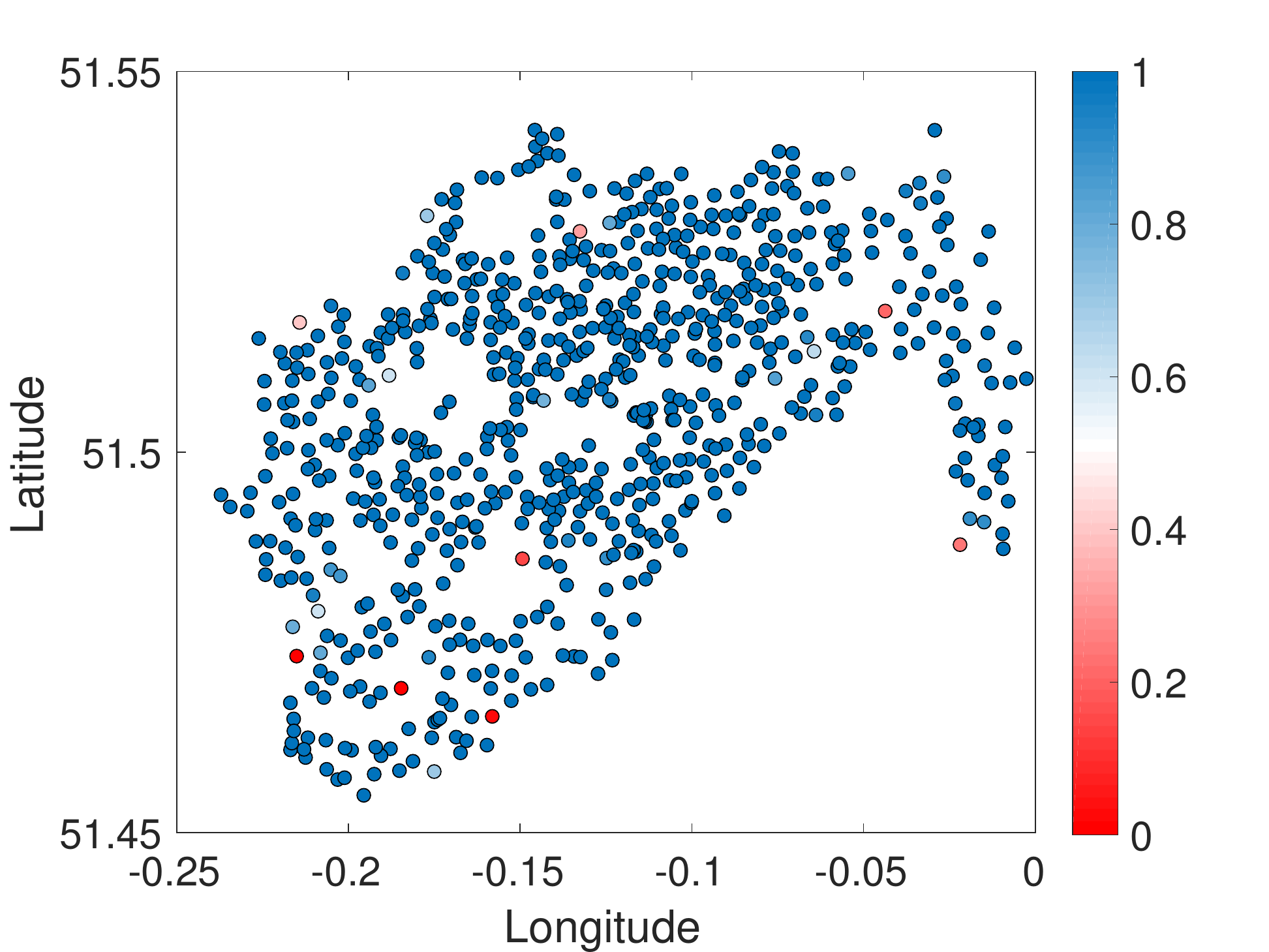}
}
\subfigure[]{
\label{phi1alld500}
\includegraphics[width=.465 \textwidth]{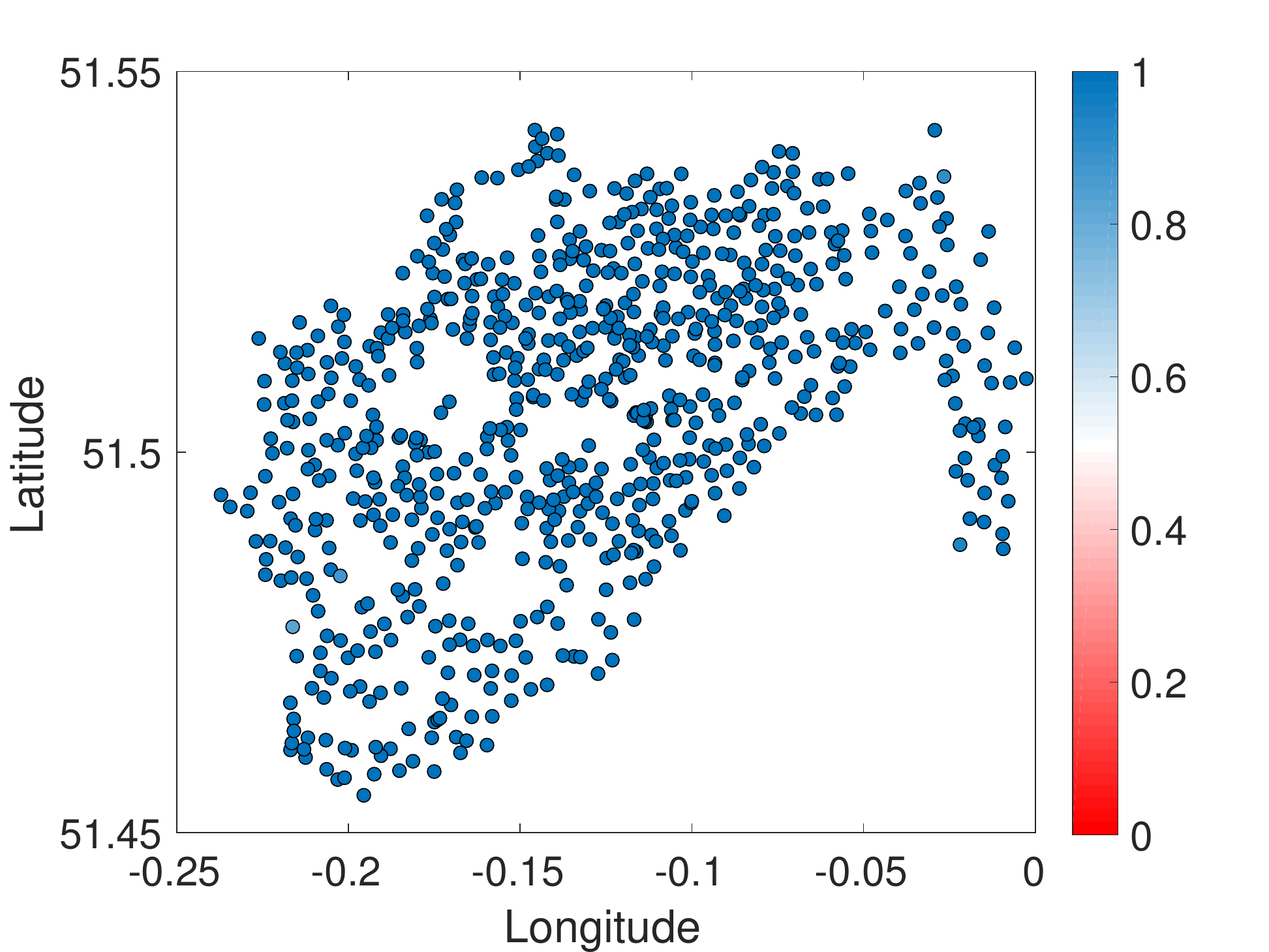}
}
\end{center}
\caption{ Approximate satisfaction probability of formula $\phi_1$ for $1000$ runs for each BSS station for (a) $d=0$, (b) $d=0.2$, (c) $d=0.3$ and (d) $d=0.5$. The value of the probability is given by the color legend.}
\label{phi1all}
\end{figure}

%

The second property ``if I wait $t$ minutes, I will always find a bike and a free slot'' can instead be formalised by the following property: 
\begin{equation}
 \phi_{2}= \glob{[0,T_{end}]}\{\ev{[0,t]} ((B>0)  \wedge(S > 0)) \}
\label{formulatime}
\end{equation}
A station $\ell$ satisfies $\phi_2$ if and only if it is always true that, for each time $h \in [0, T_{end}]$, eventually, in a time between $h$ and $h+t$, there will be a bike and a free slot available in $\ell$. 

In the analysis, we explore parameter $t \in [0, 10]$ minutes to see how the satisfaction of the property changes in each location with respect to the waiting time $t$ for a bike or a free slot. Figure~\ref{phi2all} shows the approximate satisfaction probability $p_{\phi_2}$ for 1000 runs for all the stations, for (a) $t = 0$ and (b) $t = 10$.  In this case, we can see that waiting a certain amount of time does not greatly increase  the probability to always find a bike or a free slot in the stations; even after waiting for 10 minutes Fig.~\ref{phi2all} (b) shows that there are still many stations that satisfy $\phi_2$ with a nearly zero probability. This means that there are stations that remain full or empty for at least 10 minutes.
Considering that the average human walking speed is about 5.0 kilometres per hour (km/h), we can conclude that if users do not find a bike or a free slot in a station, they usually have more probability to find the bike/free slot in another station rather than when waiting in the same station. The standard deviation remains in the interval [0,0.0159]. Figure~\ref{singleloc} (b) shows the satisfaction probability of some BBS stations vs the time t=[0,10].

\begin{figure}[!t]
\begin{center}
\subfigure[]{
\label{phi2t0}
\includegraphics[width=.465 \textwidth]{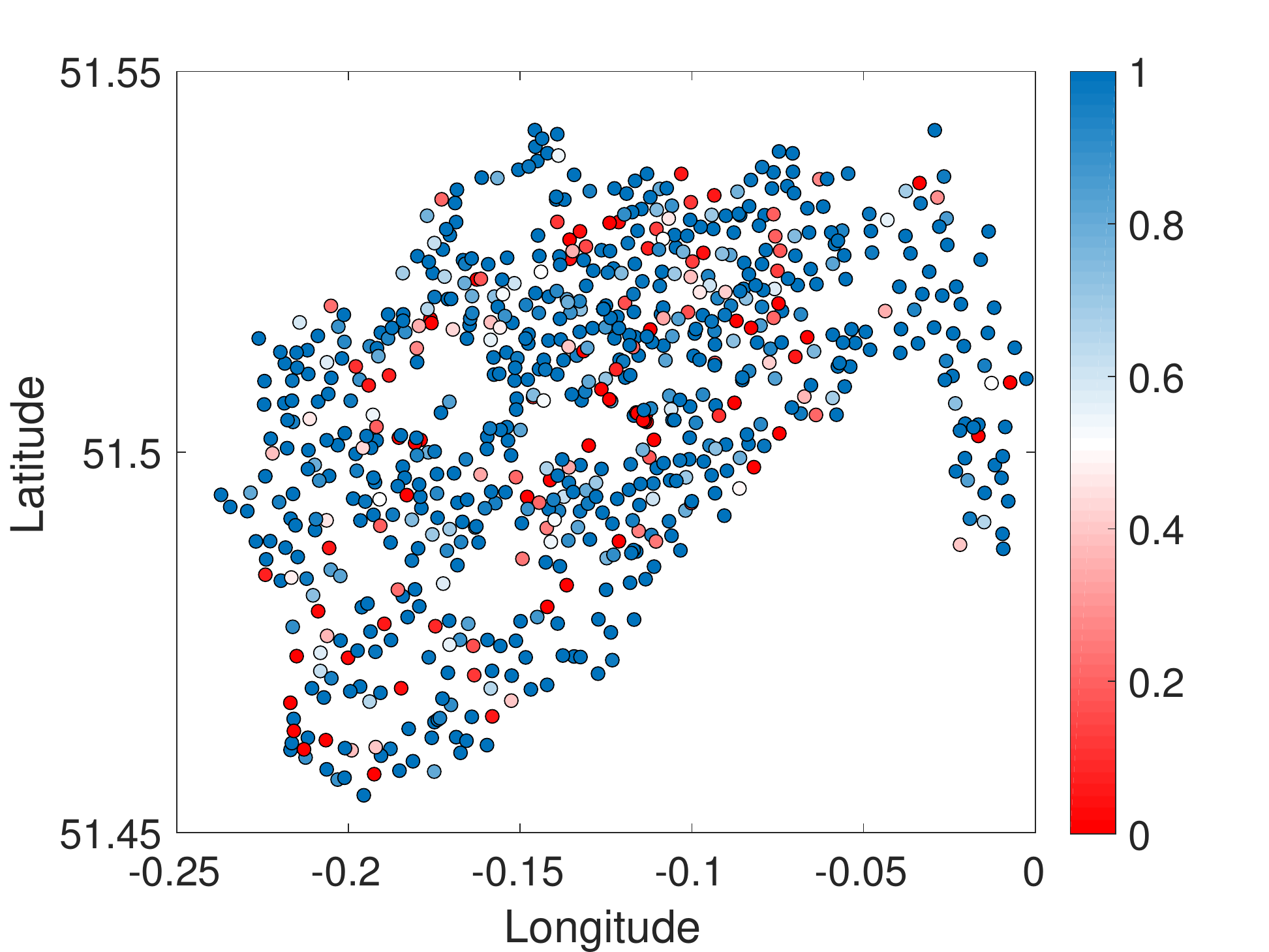}
}
\subfigure[]{
\label{phi2tend}
\includegraphics[width=.465 \textwidth]{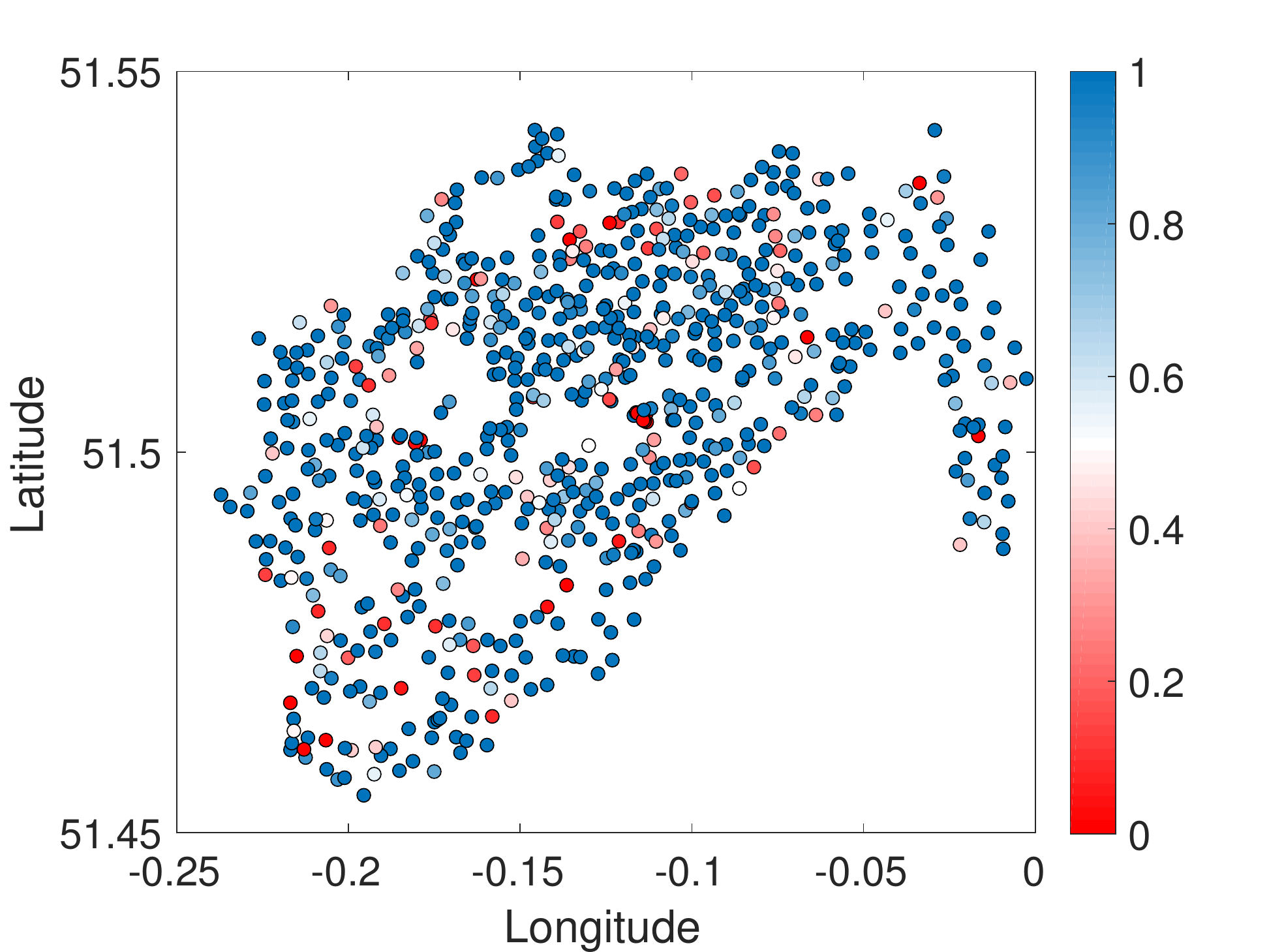}
}
\end{center}
\caption{Approximate satisfaction probability degree of formula $\phi_2$ for $1000$ runs for each BSS station for (a) $t=0$ and (b) $t=10$. The value of the degree is given by the color legend.}
\label{phi2all}
\end{figure}

\begin{figure}[!t]
\begin{center}
\subfigure[]{
\label{phi1single}
\includegraphics[width=.465 \textwidth]{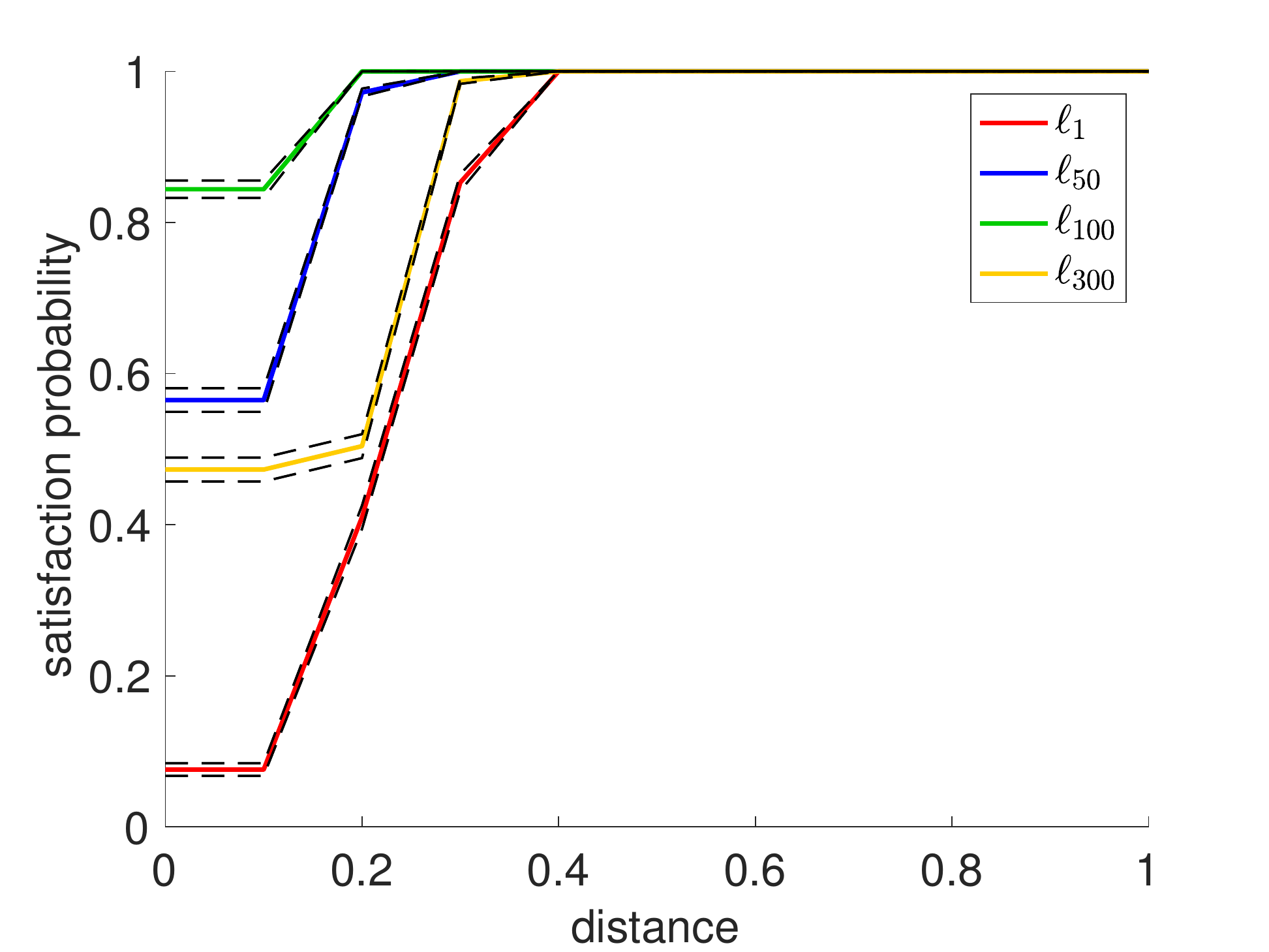}
}
\subfigure[]{
\label{phi2single}
\includegraphics[width=.465 \textwidth]{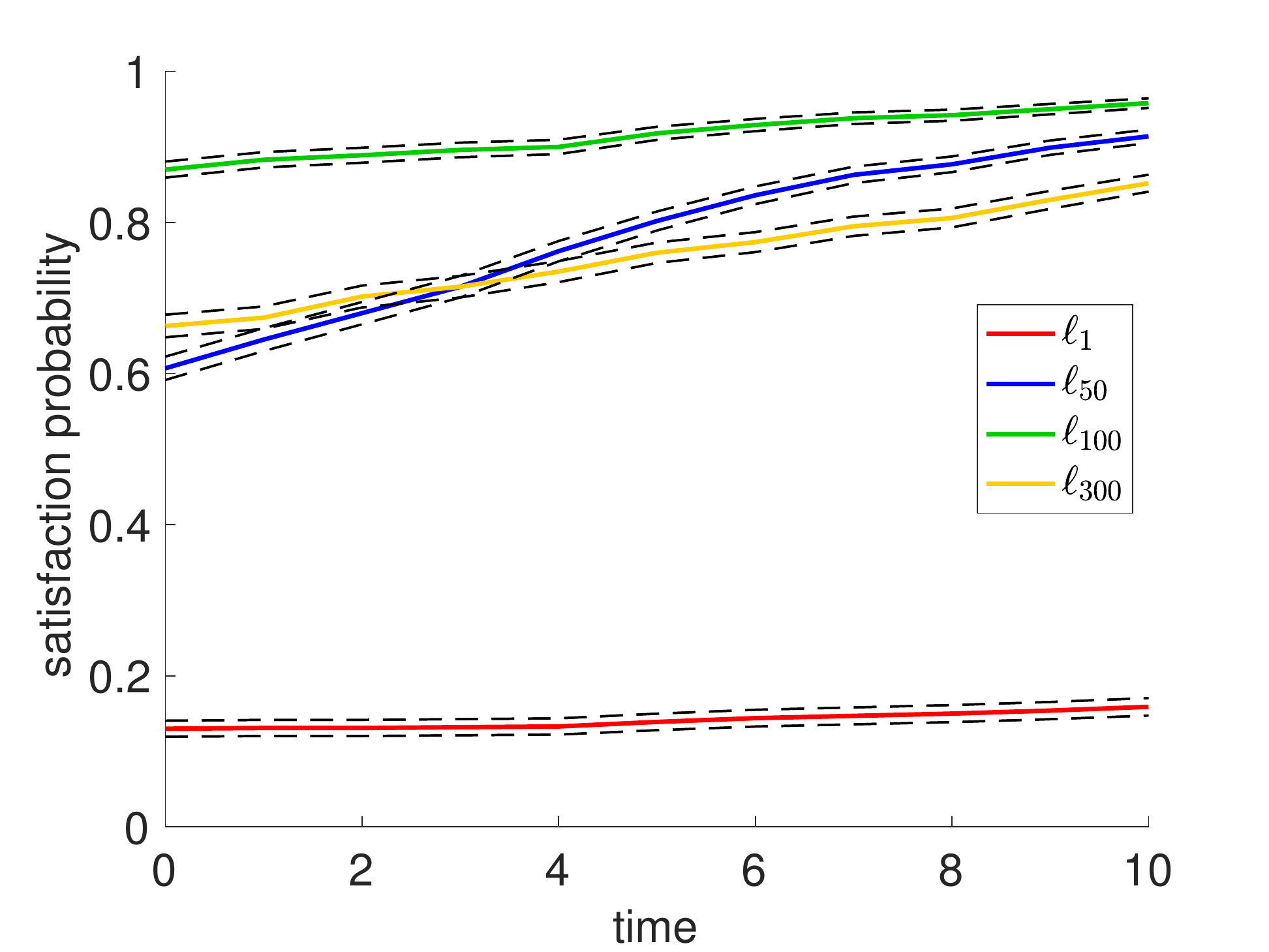}
}
\subfigure[]{
\label{phi3single}
\includegraphics[width=.465 \textwidth]{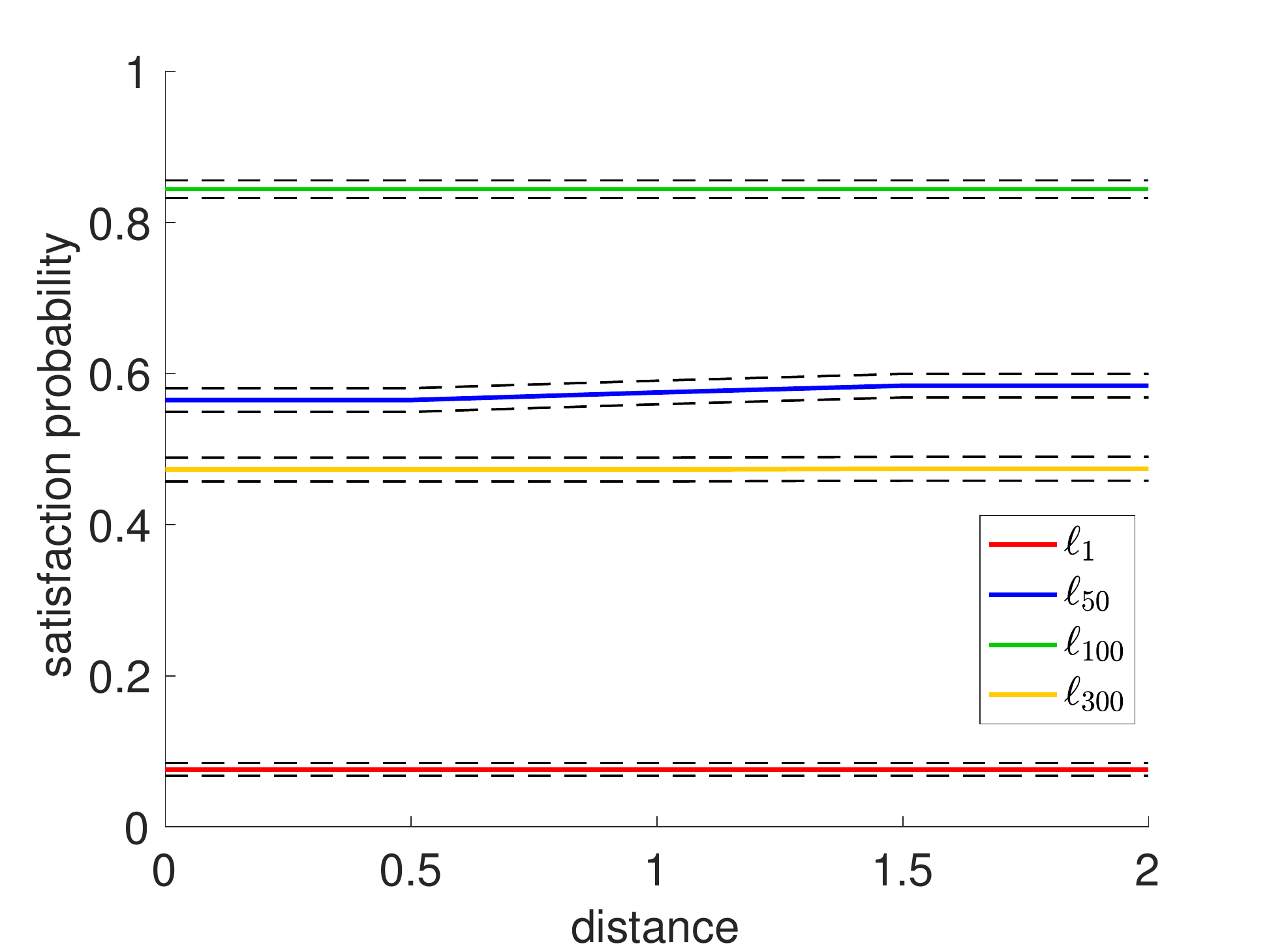}
}
\subfigure[]{
\label{psi3single}
\includegraphics[width=.465 \textwidth]{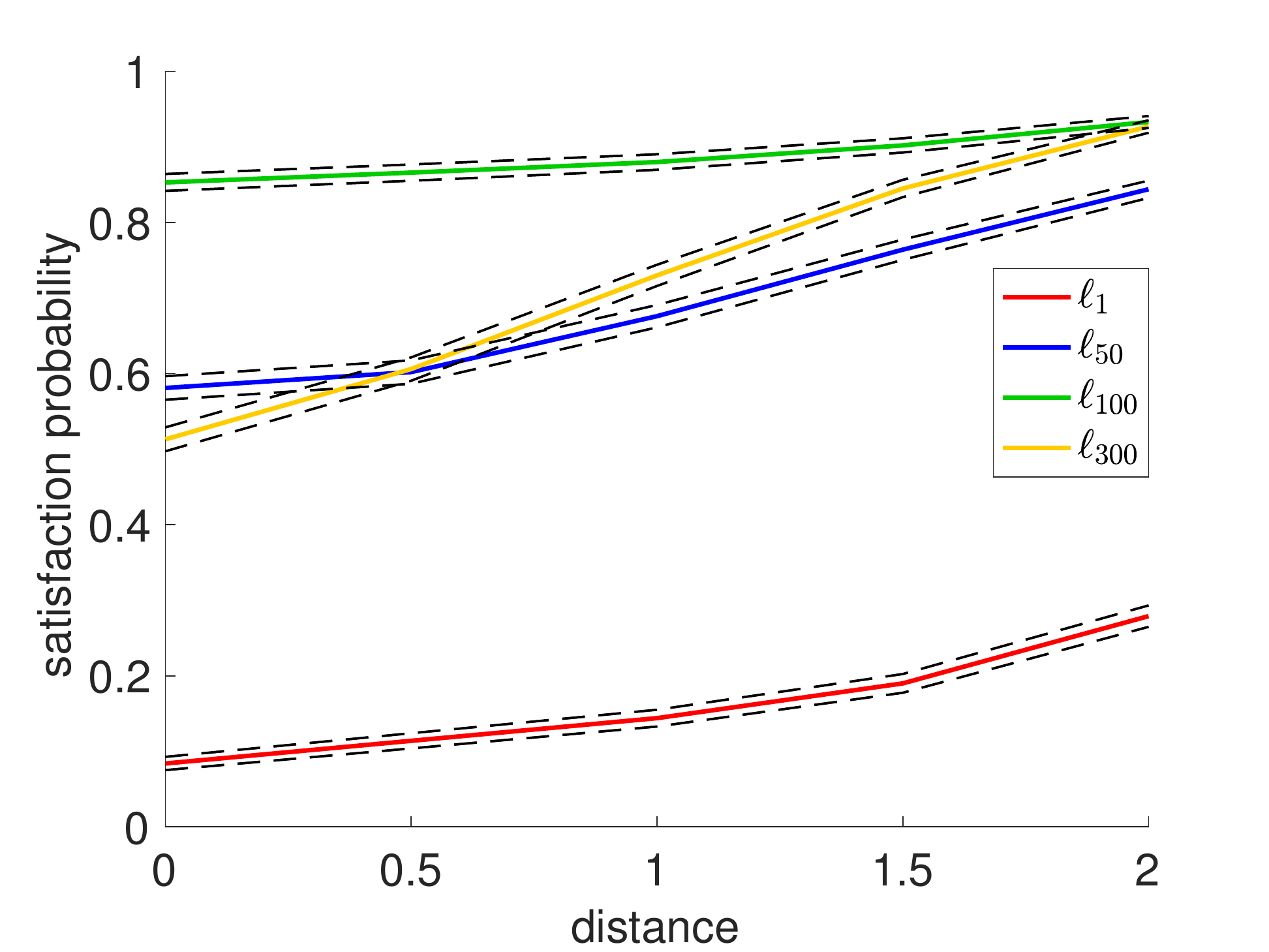}
}
\end{center}
\caption{ Approximate satisfaction probability for 1000 runs of BSS stations $\ell_{1}$,$\ell_{50}$, $\ell_{100}$, and $\ell_{300}$ for property (a) $\phi_1$  and d=[0,1.0],  (b) $\phi_2$ and t=[0,10], (c) $\phi_3$  and d=[0,2.0]  (d), $\psi_3$  and d=[0,2.0].}
\label{singleloc}
\end{figure}

Using the somewhere operator, one can only check if there is a bike in a station within a certain distance. It may be the case that  the station with a bike is located in the opposite direction the user is headed to, costing her a large deviation and more time to pick up the bike. A more precise check is given by the property below:
\begin{equation}
 \phi_{3}= \glob{[0,T_{end}]}\{B=0 \Rightarrow (B=0)\surround{[0,d]} (B > 0) \}
\label{surBike}
\end{equation}
A station $\ell$ satisfies $\phi_3$ if and only if it is always true, for each time $h \in [0, T_{end}]$, that a station with no bikes is surrounded by stations with at least one bike at a distance less than or equal to $d$. A similar property can be defined for free slots. 
 In other words, this would guarantee that a user who finds a station without bikes will find a station with at least one bike within distance $d$, no matter in which direction the user looks for another station. 
 In the analysis, we explore parameter $d \in [0, 0.2]$ kilometers to see how the satisfaction of the property changes in each location.
The results show that increasing the distance does not significantly increase the probability of $\phi_3$; that is, the probability to be surrounded by stations with at least one bike does not increase. Figure~\ref{singleloc} (c) shows the satisfaction probability of some BBS stations vs the distance d=[0,2]. The standard deviation remains in the interval [0,0.0158].  In other words, there are always some stations with  almost a constant number of bikes/free slots in some directions.
This analysis suggests that considering spatial operators that include a notion of direction might be interesting for future work.

Finally, in all the previous properties we did not consider that a user will need some time to reach a nearby station. Both properties $\varphi_1$ and $\varphi_2$ can be refined to take this aspect into consideration by considering a nested spatio-temporal property: 
\begin{equation}
 \psi_{1}= \glob{[0,T_{end}]}\{\somewhere{[0,d]} (\ev{[t_w, t_w]} B>0) \wedge \somewhere{[0,d]} (\ev{[t_w, t_w]}  S > 0)\}
\end{equation}
A station $\ell$ satisfies $\psi_1$ if and only if it is always true between 0 and $T_{end}$ minutes that there exists a station at a distance less than or equal to $d$, that, eventually in a time equal to $t_{w}$ (the walking time), has at least one bike and a station at a distance less than or equal to $d$, that, eventually in a time  equal to $t_w$ has at least one free slot. 

Similarly, 
\begin{equation}
 \psi_{3}= \glob{[0,T_{end}]}\{B=0 \Rightarrow (B=0)\surround{[0,d]} (\ev{[t_w,t_w]} B > 0) \}
\end{equation}
which expresses that a station $\ell$ satisfies $\psi_3$ if and only if it is always true between 0 and $T_{end}$ minutes that if it is empty, it is surrounded by stations within a distance of at most $d$ that have at least one bike, eventually in a walking time equal to $t_{w}$.

We consider an average walking speed of $6.0$~km/h, this means for example that if we evaluate a distance $d = 0.5$ kilometers, we consider a walking time $t_w = 6$ minutes. The results of $\psi_1$ are very similar to the results of $\phi_1$. This means that there is not much difference between looking at $t=0$ or after the walking time.
This is in accordance with the result of property $\phi_2$ where increasing the time does not significantly increase the satisfaction probability of $\phi_2$. This result becomes even more evident evaluating property $\psi_3$. Figure~\ref{singleloc} (d) shows the satisfaction probability of  $\psi_3$ of some BBS stations vs distance d=[0,2.0]. In this case, the probability increases slightly with respect to property $\phi_3$, obtaining a similar result for property $\phi_2$ and  meaning that in the selected stations ($l_0$, $l_{50}$, $l_{100}$ and 
$l_{300}$) the probability that a user finds a bike while walking to another station is increasing with the distance.
Figure~\ref{diff} shows the difference between the satisfaction probability of properties $\psi_1$, $\phi_1$  and properties $\psi_3$, $\phi_3$ for the same locations as in Figure~\ref{singleloc}. We can see that, in the second case, Figure~\ref{diff} (b), there is a clear improvement of the satisfaction probability vs the distance parameter. 

\begin{figure}[!t]
\begin{center}
\subfigure[]{
\label{diff1single}
\includegraphics[width=.465 \textwidth]{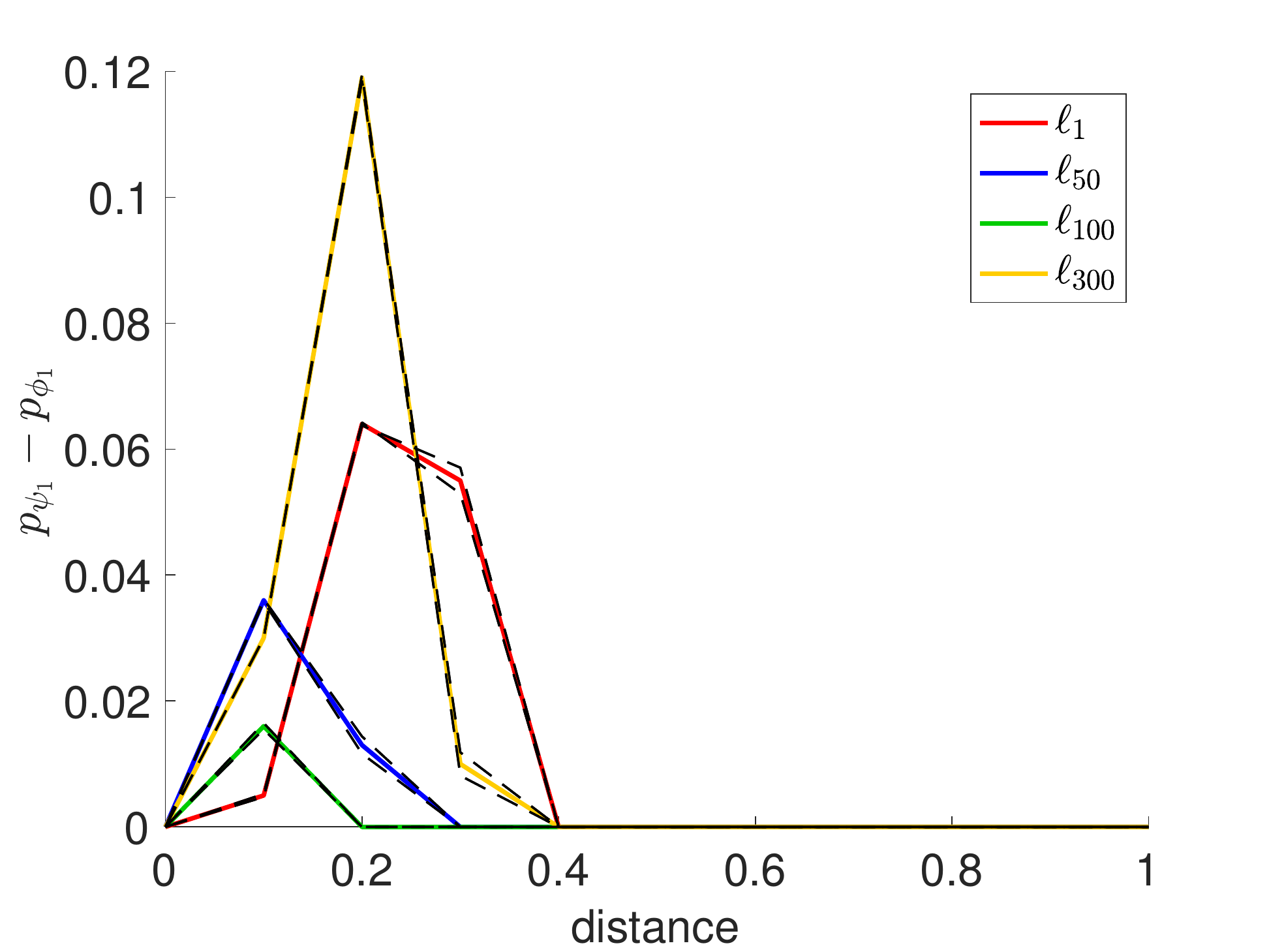}
}
\subfigure[]{
\label{diff3single}
\includegraphics[width=.465 \textwidth]{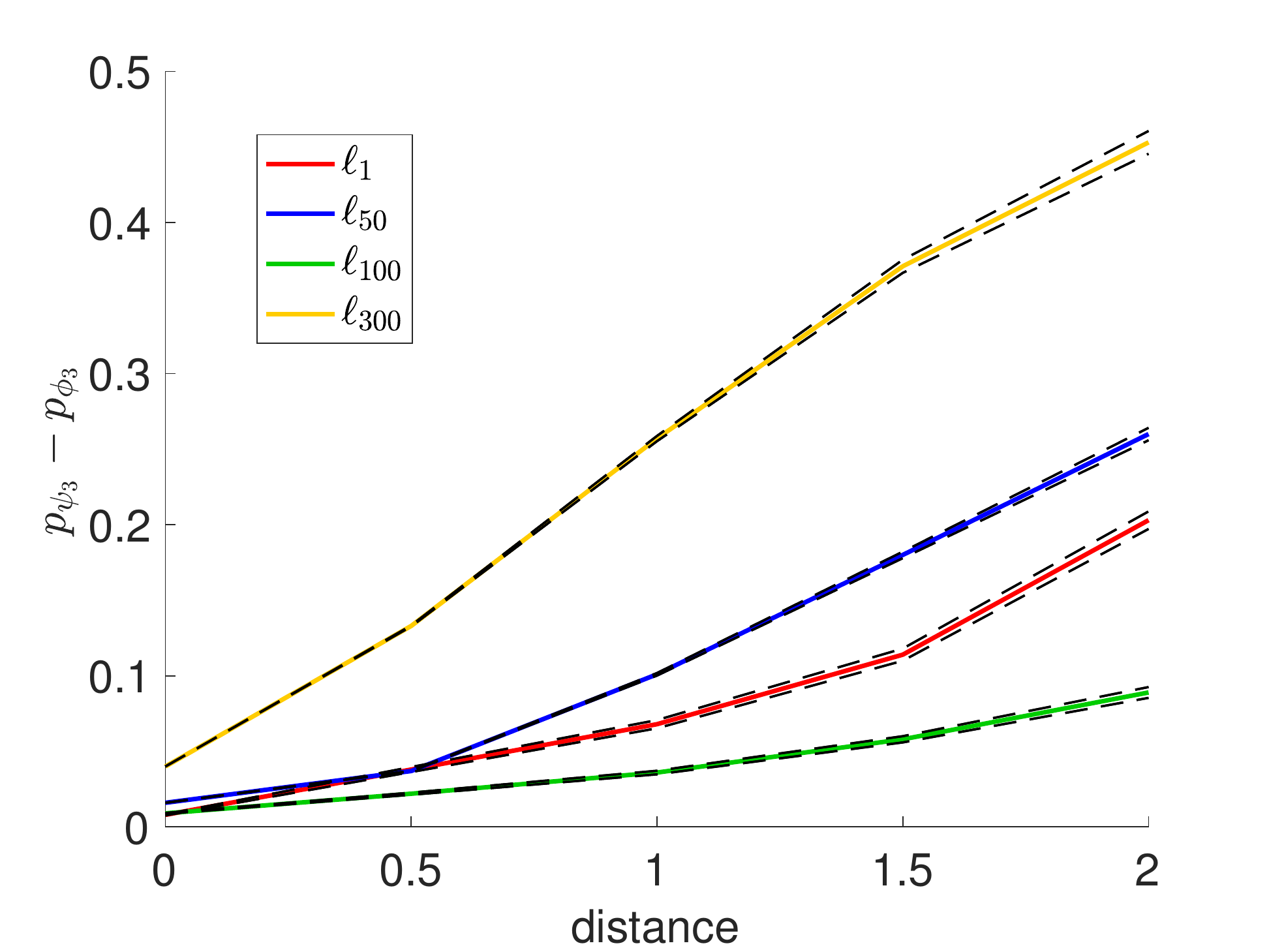}
}
\end{center}
\caption{  (a) $p_{\psi_1} -p_{\phi_1}$ vs the distance d = [0,1.0]  (b)$p_{\psi_3} -p_{\phi_3}$ vs the distance d = [0,2.0]  .}
\label{diff}
\end{figure}

The analysis is performed using the Java implementation. 
As time performance, the monitoring of a single traces takes $0.207s$ for property $\phi_1$, $0.011s$ for property $\phi_2$, $0.971s$ for property $\phi_3$, $0.178s$ for property $\psi_1$, and $0.667s$ for property $\psi_3$.  The computation of the distance matrix takes about $5.574s$.  All the experiments were run on an 
 Intel Core i7  3.5 GHz CPU, with 16GB 2133 MHz RAM.

%% file: conc.tex

\section{Discussion}
\label{sec:disc}

\vspace{-1ex}

We defined the Signal Spatio-Temporal Logic, a spatio-temporal extension of STL \cite{Donze2010},
with two spatial operators: the somewhere operator and 
the spatial (bounded) surrounded operator. In SSTL spatial and temporal operators can be arbitrarily nested. We provided the logic with a Boolean and a quantitative semantics in the style of STL~\cite{Donze2010}, and defined novel monitoring algorithms to evaluate such semantics on spatio-temporal trajectories. The monitoring procedures, implemented in Java, have been applied on two case studies. In the first one, we study a Turing reaction-diffusion system, modelling a process of morphogenesis~\cite{turing_chemical_1952} in which spots are formed over time. The second one is a bike sharing system modelled as a CTMC where we analyse the likelihood for users of the system to find bikes and free parking slots when they need them.

The work in this paper can be extended in several directions. First, we plan to perform a more thorough investigation of the expressivity of the logic, and to apply it on further case studies. 
Secondly, we plan to extend our logic to more general quasi-discrete metric spatial structures, exploiting the topological notion of closure spaces~\cite{ciancia2014,Ci+16} and extending it to the metric case. Note that the current monitoring algorithms work already for more general spatial structures, like finite directed weighted graphs, but we plan to provide a more precise characterisation of the class of discrete spatial structures on which they can be applied. We also plan to further optimise the implementation to improve performance, and additionally investigate if and how directionality can be expressed in SSTL. Finally, we plan to exploit the quantitative semantics for the robust design of spatio-temporal systems, along the lines of~\cite{TCS2015}.



%% file: appendix.tex
%


\section{Proofs}
\label{sec:appendix}

In this appendix, we present the proofs of Proposition \ref{prop:error} and 
the correctness of the Boolean (Algorithm~\ref{bspuntil}) and the quantitative (Algorithm \ref{algo:quantitative}) monitoring algorithms for the surrounded operator. 
%

\noindent\textbf{Proposition \ref{prop:error}.}\emph{
Let the primary signal $\vec x$ be Lipschitz continuous, as well as the functions defining the atomic predicates. Let $M$ be a Lipschitz constant for all secondary signals, and $h$ be the discretisation step.  Given a SSTL formula $\varphi$, let $u(\varphi)$ count the number of temporal until operators in $\varphi$, and denote by $\rho(\varphi,\vec x)$ its satisfaction score over the  trace $\vec x$ and by $\rho(\varphi,\vec{\hat x})$ the satisfaction score over the discretised version $\vec{\hat{x}}$ of $\vec x$ with time step $h$. Then 
$ \| \rho(\varphi,\vec x) - \rho(\varphi,\vec{\hat x}) \| \leq u(\varphi) M h $.
}

\begin{proof} We first observe that the monitoring algorithm for Boolean and spatial operators preserve the error of the input quantitative signals. This means that if $\|s_{\varphi_j,\ell} - \hat{s}_{\varphi_j,\ell}\| \leq \varepsilon$, then $\|s_{\psi,\ell} - \hat{s}_{\psi,\ell}\| \leq \varepsilon$, for $\psi$  one of $\neg\varphi_1$, $\varphi_1 \wedge \varphi_2$, $\varphi_1 \surround{[w_1,w_2]} \varphi_2$, $\somewhere{[w_1,w_2]} \varphi_1$. Hence, temporal discretisation introduces errors only for temporal operators. 

Now, let $I=[t_1,t_2]$ be such that $t_j = k_jh$, and denote the Minkowski sum by $\oplus$, so that $t\oplus I = [t+t_1,t+t_2]$. Denote by $\hat{I}$ the discretised version of $I$, with step $h$, $\hat{I} = \{ k_1 h, (k_1+1)h,\ldots, k_2 h\}$. 
We observe two facts for the maximum, with identical statements holding for the minimum. 
\begin{itemize}
\item Let $f(t)$ be Lipschitz with constant $K$. Let $g(t) = \max_{t'\in t\oplus I} f(t)$ and $\hat{g}(t) = \max_{t'\in t\oplus \hat{I}} f(t)$. Then $\| g(t) - \hat{g}(t) \| \leq Kh/2$. This holds by applying the Lipschitz property between a generic point in $t\oplus I$ and the closest point in $t\oplus \hat{I}$, and noting that the maximum distance between such points is $h/2$. 
\item If $\tilde f$ is such that $\|\tilde{f}(t) - f(t)\| \leq \varepsilon$ uniformly in $t$, and we let $g$, $\hat{g}$ as above, and $\tilde{g}(t) = \max_{t'\in t\oplus \hat{I}} \tilde{f}(t)$, then 
\[ \|g(t) - \tilde g(t)\| \leq \|g(t) - \hat g(t)\| + \|\hat{g}(t) - \tilde g (t)\| \leq Kh/2+\varepsilon.  \] 
\end{itemize}
Hence, the second property implies that the additional error we introduce by evaluating a time bounded until is an additive term no larger than  $Kh$, as in the definition of the quantitative semantics of the until, there are a nested minimum and a maximum over dense time intervals. Hence the total error will be bounded by $Kh$ times the number of temporal operators. 
\end{proof}

\subsection{Correctness of the Boolean Monitoring Algorithm for the Surrounded Operator}
\label{app:Booleanproof}
In this section we prove the correctness of Algorithm~\ref{bspuntil};  let us call the algorithm BoolSurround. 


\restBooleansurr*

\begin{proof}
First we note that $s_{\psi, \ell}(I_{i}) = 1 \iff \ell \in V_{I_{i}} $, where $V_{I_{i}}$ is the set $V$ at the end of the iteration of the $I_{i}$ interval. Then, it is enough to prove that, for all $I_{i}\in {\mathcal{I}}_{s_{\psi,\ell}}$
$$ \ell \in V_{I_{i}} \iff (\vec x, t, \ell) \models \phi_{1}\surround{[d_{1},d_{2}]}\phi_{2} \quad \forall t \in I_{i}.$$
Furthermore, for the definition of the minimal interval covering (Definition~\ref{def:intcov}), $s_{\varphi_{1},\ell'},s_{\varphi_{2},\ell'}, \ell' \in L^{\ell}_{[0,w_{2}]}$ have the same value in each $I_{i} \in {\mathcal{I}}_{s_{\psi,\ell}}$. This implies, for the Boolean semantics of the surrounded operator, that, $(\vec x, \hat t, \ell) \models \phi_{1}\surround{[d_{1},d_{2}]}\phi_{2}$ for a specific $\hat t \in I_{i}$ if and only if it satisfies the property for all $t \in I_{i}$. 

Let's consider now the distance constraints of the formula. 
We redefine the property $\phi_{1}$ and $\phi_{2}$ in this way: 
$$ (\vec x, t, \ell) \models \hat{\phi_{1}} \iff (\vec x, t, \ell')\models \phi_{1} \wedge d(\ell, \ell') \leq d_{2},$$ and  $$ (\vec x, t, \ell) \models \hat{\phi_{2}} \iff (\vec x, t, \ell')\models \phi_{1} \wedge d(\ell, \ell') \in [d_{1}, d_{2}].$$
Hence, we have that
 $V= \{ \ell' | s_{\hat \phi_{1},\ell'}(I_{i})=1 \}$ and  $Q= \{ \ell' | s_{\hat \phi_{2},\ell'}(I_{i})=1 \}$.

Furthermore, a location $\ell$ is a bad location if it can reach a point satisfying $\neg \hat \phi_{1}$ passing for a node $\neg \hat \phi_{2}$.
Let's consider the set 
{\small
$${\mathcal{C}}_{\ell} = \{ i \in \bb{N} | \exists p: \ell \leadsto  \infty.G, (\vec x, t, p(i)) \models \neg \hat \phi_{1}, \mbox{ and } \forall j \in \{1, \cdots, i \} (\vec x, t, p(j)) \models \neg \hat \phi_{2} \},$$ 
}
where $p: \ell \leadsto  \infty.G$ is a path of the graph $G$, starting from $\ell$, then 
$$(\vec x, t, \ell) \models \phi_{1} \surround{[d_{1},d_{2}]}\phi_{2} \iff (\vec x, t, \ell) \models \hat \phi_{1} \wedge {\mathcal{C}}_{\ell} = \emptyset. $$
Hence, what we have to prove at the end is that
$$\ell \in V_{I_{i}} \iff (\vec x, t, \ell) \models \hat \phi_{1} \wedge {\mathcal{C}}_{\ell} = \emptyset, \mbox{ for a }t \in I_{i}.$$

We will prove it by induction. From this point on, we fix the trace $\vec x$ and the time $t$ and we will write $\ell \models \phi$ to indicate $(\vec x, t, \ell) \models \phi$ and $V$ for $V_{I_{i}}$.
\begin{description}
\item[($\Rightarrow$)] 
We have to prove that if $(\vec x, t, \ell) \models \hat \phi_{1} \wedge \min{\mathcal{C}}_{\ell}=k$ then $\ell$ is removed at iteration $k$ from $V$.
\begin{description}
\item[(basis step)] $ \ell \models \hat \phi_{1} \wedge  \min{\mathcal{C}}_{\ell}=1$ then $p(1) \models \neg \hat \phi_{1} \wedge p(1) \models \neg \hat \phi_{2}$. This implies that $\exists \ell'  \in T= B^{+} (Q \bigcup V),$ and $(\ell, \ell') \in E$, then $\ell$ is removed from V at the first iteration.  
\item[(inductive step)] Let's suppose that if $\ell \models \hat \phi_{1} \wedge \min{\mathcal{C}}_{\ell}=k$ then $\ell$ is removed at iteration $k$ from $V$. We have to prove that this is true also for $k+1$.  Let's suppose that $\ell \models \hat \phi_{1} \wedge \min{\mathcal{C}}_{\ell}=k + 1$. This implies that $p(k+1) \models \neg \hat \phi_{1}$ and $\forall j \in \{ 1, \cdots, k \}$, $p(j) \models \neg \hat \phi_{2}$. But if $k+1 = \min{\mathcal{C}}_{\ell}$ then $\ell'=p(1) \models \hat \phi_{1}$ and $\min{\mathcal{C}}_{\ell'}=k$, i.e., $\ell'$ is removed at iteration $k$ from $V_{I_{i}}$, then $\ell$ is removed at iteration $k+1$ because $(\ell, \ell') \in E$.
\end{description}
\item[($\Leftarrow$)]
We have to prove that if $\ell$ is removed at iteration $k$ from $V_{I_{i}}$ then $\ell \models \hat \phi_{1} \wedge \min{\mathcal{C}}_{\ell}=k$.
\begin{description}
\item[(basis step)] If $\ell$ is removed from V at the first iteration then $\exists \ell' \in T$ s.t. $(\ell,\ell') \in E$ and $(\vec x, t, \ell') \models \neg \hat \phi_{1} \wedge \neg \hat \phi_{2}$, this implies $\min{\mathcal{C}}_{\ell}=1$.
\item[(inductive step)] 
Let's suppose that if $\ell$ is removed at iteration $k$ from $V$ then
$\ell \models \hat \phi_{1} \wedge \min{\mathcal{C}}_{\ell}=k$. We have to prove that this is true also for $k+1$.  
Let's suppose that $\ell$ is removed at iteration $k + 1$ from $V$.
This implies that $\exists \ell' \in L$ s.t. $(\ell, \ell') \in E$ and $\ell' \in T$ but this means that $\ell'$ was removed from $V$ at the previous iteration $k$ and from the inductive step we have $\min{\mathcal{C}}_{\ell'} = k$.
If  $\min{\mathcal{C}}_{\ell'} = k$ then $\exists p: \ell' \leadsto  \infty.G$ s.t. $p(k) \models \neg \hat \phi_{1}$ and, $\forall i \in \{1, \cdots, k \}, p(i) \models \neg \hat \phi_{2}$. But $(\ell,\ell') \in E$ and $(\vec x, t, \ell') \models \neg \hat \phi_{2}$ (because $\ell' \in T$) then $\exists p': \ell \leadsto  \infty.G$ s.t. $p(k+1) \models \neg \hat \phi_{1}$ and, $\forall i \in \{1, \cdots, k+1 \}, p(i) \models \neg \hat \phi_{2}$. This implies that $\min{\mathcal{C}}_{\ell} \leq k+1$, but it can be less than $k+1$ because in that case it has to be removed before. Hence, we can conclude that $\min{\mathcal{C}}_{\ell} = k+1.$\qedhere
\end{description}
\end{description}
\end{proof}

\subsection {Correctness of the Quantitative Monitoring Algorithm for the Surrounded Operator}
\label{app:quantproof}

In this section, we present the proofs of Theorem~\ref{thm:fixedpoint_quant}, Corollary~\ref{rhomin}  and Proposition~\ref{diameterconv}. For simplicity, we report again the statements.

\restFixedpointQuant*
%

Note that $s$ is equivalent to the quantitative semantics of the surrounded operator $\varphi_{1}\mathcal{S}\varphi_{2}$, with $s_{i}$ denoting the robustness of $\varphi_{i}$, without the distance constraints. We first present two lemmas, followed by the proof of Theorem~\ref{thm:fixedpoint_quant}.
\begin{lem}
\label{stepequal}
If $\mathcal{X} (k+1, \ell)=\mathcal{X} (k, \ell)$ for all $\ell \in L$ then, $\forall i > k,$ $\mathcal{X} (i, \ell)=\mathcal{X} (k, \ell).$

\begin{proof}
By induction.
\begin{description}
\item[(basis step)] i=k +1 is true by hypothesis,
\item[(inductive step)] suppose the assert holds for $i>k$, i.e., $\mathcal{X} (i, \ell)=\mathcal{X} (k, \ell)$ (I.H.),  then we have to prove that it holds for $i+1$. 
\begin{align*}
\mathcal{X} (i+1, \ell) &= \min(\mathcal{X}(i, \ell), \min_{\ell'|\ell E \ell'}(\max(\mathcal{X}(i,\ell'),s_{2}(\ell')))) &\{\mbox{by Def. of\;} \mathcal{X}\}
\\  &=  \min(\mathcal{X} (k, \ell), \min_{\ell'|\ell E \ell'}(\max(\mathcal{X}(k,\ell'),s_{2}(\ell'))))   &\{\mbox{by I.H.}\}
\\  &=  \mathcal{X} (k+1, \ell)=\mathcal{X} (k, \ell). &\{\mbox{by Def. of\;} \mathcal{X}\tag*{\qedhere}\}
\end{align*}
\end{description}
\end{proof}
\end{lem}

\begin{lem}
\label{inA}
Let $A_{\ell}$ be the subregion that maximizes $s(\ell)$, then, $\forall \ell' \in A_{\ell},$ $s(\ell') \geq s(\ell)$.

\begin{proof}
If $A_{\ell}$ is the subregion that maximizes $s(\ell)$ then
$$s(\ell)=\min (\min_{\ell' \in A_{\ell}}s_{1}( \ell'),\min_{ \ell' \in B^{+}(A_{\ell})}s_{2}( \ell')))$$
Suppose by contradiction that $\exists \hat{\ell} \in A_{\ell}$ s.t. $s(\hat{\ell})<s(\ell)$. Let $Q=\{A \subseteq L, \hat{\ell} \in A \}$.
Then
$$s(\hat{\ell}) = \max_{A \in Q} (\min(\min_{\ell' \in A}s_{1}(\ell'),\min_{\ell' \in B^{+}(A)}s_{2}(\ell')))$$
and $s(\hat{\ell})<s(\ell)$ implies
$$
\max_{A \in Q} (\min(\min_{\ell' \in A}s_{1}(\ell'),\min_{\ell' \in B^{+}(A)}s_{2}(\ell')))  <  \min (\min_{\ell' \in A_{\ell}}s_{1}( \ell'),\min_{ \ell' \in B^{+}(A_{\ell})}s_{2}( \ell')))
$$

But $A_{\ell}$ is a subset of L and $\hat{\ell} \in A_{\ell}$ therefore $A_{\ell} \in Q$, thus the inequality can not hold. 
\end{proof}
\end{lem}

\begin{proof}[Proof of Theorem~\ref{thm:fixedpoint_quant}]
We have to prove that (1) $\mathcal{X}(i,\ell)$ converges  in a finite number of steps, in each location $\ell$, to $\mathcal{X}(\ell) \in \mathbb{R}^{*}$ and that (2) $\forall \ell \in L$, $\mathcal{X}(\ell)= s(\ell)$.
\begin{enumerate}
\item Convergence of $\mathcal{X}.$ \\
First note that $\mathcal{X}(i, \ell) \geq \min(\mathcal{X}(i, \ell), \min_{\ell'|\ell E \ell'}(\max(\mathcal{X}(i,\ell'),s_{2}(\ell'))))=\mathcal{X} (i+1, \ell)$, thus $\mathcal{X}_{|\ell}$ is a monotonic decreasing function. Second, note that $\mathcal{X} ( i, \ell) \in \{s_{j}(\ell) \:\big{|}\:  j \in \{ 1,2 \}, \ell \in L \}$ is a finite set of sortable values. So, in every step,  $\mathcal{X}$ takes a value of a sortable finite set. Finally, if it happens that for a step, for all $\ell \in L$, $\mathcal{X}(i,\ell)$ does not change then, from Lemma \ref{stepequal}, it will remain the same for all the next steps. The convergence of $\mathcal{X}$ to the maximum fixed point follows then from Tarsky's theorem.
\item We have to prove that $ \forall \ell$, $\mathcal{X}(\ell)=s(\ell)$. 

Let $A_{\ell}$ be the subregion that maximizes $s(\ell)$ then 
$$s(\ell)= \min (\min_{\ell' \in A_{\ell}}s_{1}(\ell'),\min_{ \ell' \in B^{+}(A_{\ell})}s_{2}(\ell'))).$$

First we prove that   (2a) $\forall \ell$, $\mathcal{X}(\ell) \geq s(\ell)$ and then that (2b) they are equal.
\begin{enumerate}[label=(2\alph*)]
\item To prove that $\mathcal{X}(\ell) \geq s(\ell)$ it suffices to prove that, for a generic $\ell$, $\forall i \in \mathbb{N}$, $\mathcal{X}(i, \ell) \geq s(\ell)$, and for the convergence of $\mathcal{X}$ that $\exists j \in \mathbb{N}$ s.t. $\mathcal{X}(\ell)=\mathcal{X}(j,\ell), \forall \ell, \forall j \geq i$. The proof is by induction. 
\begin{itemize}
\item (basis step)
%
\begin{align*}
\mathcal{X}(0, \ell) &= s_{1} (\ell) \quad &\{\mbox{by Def. of\;} \mathcal{X}\} \\
	& \geq \min_{\ell' \in A_{\ell}}s_{1}(\ell') \quad &\{\mbox{Because\;} \ell \in A_{\ell} \} \\
	& \geq  \min (\min_{\ell' \in A_{\ell}}s_{1}( \ell'),\min_{ \ell' \in B^{+}(A_{\ell})}s_{2}(\ell'))) \quad &\{ \min \mbox{ property} \} \\
	& =s(\ell) &\{\mbox{by Def. of\;} s(\ell)\} 
\end{align*}
                                    
\item (inductive step)  Assume $\mathcal{X}(i,\ell) \geq s(\ell)$, to prove that $\mathcal{X}(i+1,\ell) \geq s(\ell)$;

{\small \begin{align*}
  \mathcal{X}(i+1,\ell)  &=  \min(\mathcal{X}(i,\ell), \min_{\ell' |\ell E \ell'}(\max(\mathcal{X}(i,\ell'),s_{2}(\ell')))) &\{\mbox{by Def. of } \mathcal{X}\}
\end{align*}
}
We know by I.H. that $\mathcal{X}(i,\ell) \geq s(\ell)$, so it remains to be shown that also:
\begin{align}
\label{xdef}
\min_{\ell' |\ell E \ell'}(\max(\mathcal{X}(i,\ell'),s_{2}(\ell'))) \geq s(\ell)
\end{align}
Note that it is assumed that $\ell \in A_{\ell}$ and that $\ell'$ are direct neighbours of $\ell$. Therefore we can distinguish the following two cases:
\begin{itemize}
\item Suppose $\ell' \in A_{\ell}$. By I.H. we know that $\mathcal{X}(i,\ell') \geq s(\ell')$ and by Lemma~\ref{inA} we also know that $s(\ell') \geq s(\ell)$. For what concerns $s_2(\ell')$, if $s_2(\ell') \leq \mathcal{X}(i,\ell')$ then the max leads to $\mathcal{X}(i,\ell')\geq s(\ell)$. If instead $s_2(\ell') \geq \mathcal{X}(i,\ell') \geq s(\ell)$, then obviously also $s_2(\ell') \geq s(\ell)$. So inequation~(\ref{xdef}) holds in this case.
\item Suppose $\ell' \in B^{+}(A_{\ell})$. Then, by definition of $s(\ell)$ we know that $s_2(\ell') \geq s(\ell)$. So, if $s_2(\ell') \geq \mathcal{X}(i,\ell')$ then the inequation holds. If $\mathcal{X}(i,\ell') \geq s_2(\ell')$ then since $s_2(\ell') \geq s(\ell)$, inequation~(\ref{xdef}) also holds.
\end{itemize}

\end{itemize}

\item Suppose by contradiction that $\exists \hat{\ell} \in L$ s.t. $\mathcal{X}(\hat{\ell}) > s(\hat{\ell})$. 
At the fixed point we have that
$$ \mathcal{X}(\hat{\ell})  =  \min(\mathcal{X}(\hat{\ell}), \min_{\ell |\hat{\ell} E \ell}(\max(\mathcal{X}(\ell),s_{2}(\ell))))$$

This means that the inequality

\begin{equation}
\label{ineq}
 \min_{\ell |\hat{\ell} E \ell}(\max(\mathcal{X}(\ell),s_{2}(\ell)))> s(\hat{\ell})  
\end{equation}

has to be true.

Let $A \subseteq L$, we define:
\begin{itemize}
\item $C(A):=\{\ell \in L \big{|} \exists \ell' \in A \mbox{ s.t. } \ell' E \ell \wedge \mathcal{X}(\ell) \geq s_{2}(\ell) \}$
\item  $ C^{i}(A)=C(C^{i-1}(A))$
\end{itemize}

We can then compute the closure of C, as $C^{*}(A)=A \bigcup_{i=0}^{\infty}C^{i}( A  )$.

Because of the definition of C and the inequality (\ref{ineq}), we have that  $s_{1}(\ell) \geq \mathcal{X}(\ell) >  s(\hat{\ell})$,  $\forall \ell \in C^{*}(\{ \hat{\ell}\} )$,   and that $s_{2}(\ell) >  s(\hat{\ell})  $, $\forall \ell \in B^{+}(C^{*}(\{ \hat{\ell} \}))$; hence

$$ \min (\min_{\ell \in C^{*}(\{ \hat{\ell} \})}s_{1}(\ell),\min_{ \ell \in B^{+}(C^{*}(\{ \hat{\ell} \}))}s_{2}(\ell))) >  s(\hat{\ell})  $$
i.e.
{\small
$$ \min (\min_{\ell \in C^{*}(\{ \hat{\ell} \})}s_{1}(\ell),\min_{ \ell \in B^{+}(C^{*}(\{ \hat{\ell} \}))}s_{2}(\ell))) > \min (\min_{\ell \in A_{ \hat{\ell} }}s_{1}(\ell),\min_{ \ell' \in B^{+}(A_{ \hat{\ell} })}s_{2}(\ell')))$$
}
but this contradicts the assumption of maximality of $A_{\hat{\ell}}$.\qedhere
\end{enumerate}
\end{enumerate}
\end{proof}

\noindent\textbf{Corollary \ref{rhomin}.}\emph{
Given an $\hat{\ell} \in L$, let $\psi=\varphi_{1} \mathcal{S}_{[d_{1},d _{2}]} \varphi_{2}$ and\\
 \begin{align*}
 & s_1(\ell)= \begin{cases}
       \rho(\phi_{1}, \vec x, t, \ell) & \text{ if }  \mbox{ } 0\leq d(\hat{\ell},\ell)\leq d_{2} \\
        -\infty &   \text{ otherwise}.
       \end{cases}  \\
 & s_2(\ell)= \begin{cases}
       \rho(\phi_{2}, \vec x, t, \ell) & \text{ if }  \mbox{ } d_{1} \leq  d(\hat{\ell},\ell) \leq d_{2} \\
       -\infty & \text{ otherwise}.
      \end{cases}
\end{align*}
 Then $\rho(\psi,  \vec x, t  ,\hat{\ell}) = s(\hat{\ell})=\max_{A \subseteq L, \hat{\ell} \in A}{( \min (\min_{\ell \in A}s_{1}( \ell),\min_{ \ell \in B^{+}(A)}s_{2}( \ell))).}$
}
  

\begin{proof}
We recall that
{\small
$$\rho( \psi , \vec x,t, \hat{\ell}) =\max_{A \subseteq L^{\hat{\ell}}_{[0, d_{2}]}, \ell \in A, B^{+}(A)\subseteq L^{\hat{\ell}}_{[d_{1}, d_{2}]}}{( \min (\min_{\ell \in A}\rho(\phi_{1}, \vec x, t, \ell),\min_{ \ell \in B^{+}(A)}\rho(\phi_{2}, \vec x, t, \ell)))},$$
}
where 
$L^{\hat{\ell}}_{[ d_{1},d_{2} ]}:= \{ \ell \in A | d_{1} \leqslant   d(\ell,\hat{\ell}) \leq d_{2} \}$. This means that $\ell \in A$ iff $d(\ell,\hat{\ell}) \leq d_{2} $ and, for all $\ell' E \ell$, $d_{1} \leqslant   d(\ell',\hat{\ell}) \leq d_{2} $. 

So, we consider a restricted number of subsets of $L$ for $\rho$ and all the possible subsets of $L$ for $s$. Furthermore, the value of the locations considered by both are always the same, i.e., the value of $s_{1}$ and $s_{2}$ differ only in the locations considered by $s$ and not by $\rho$. For this reason $s(\ell) \geq \rho(\ell)$.

Let $A_{\rho}$ be the subset that maximizes $\rho$ of $\hat{\ell}$ and $A_{s}$ the subset that maximizes $s$ of $\hat{\ell}$. 
And suppose by contradiction that
$$  \min (\min_{\ell \in A_{s}}s_{1}( \ell),\min_{ \ell' \in B^{+}(A_{s})}s_{2}( \ell))) > \min (\min_{\ell \in A_{\rho}}\rho(\phi_{1}, \vec x, t, \ell),\min_{ \ell \in B^{+}(A_{\rho})}\rho(\phi_{2}, \vec x, t, \ell))),$$
but the values considered by $s$ and not by $\rho$ are all equal to $-\infty $ (see line 8 of Alg.~\ref{sur_quant_alg}), so if $A_{s}$ has a location that cannot be considered by $\rho$ it means that 
$$\min (\min_{\ell \in A_{s}}s_{1}( \ell),\min_{ \ell' \in B^{+}(A_{s})}s_{2}( \ell)))=-\infty$$ but minus infinity cannot be bigger than any number. 
\end{proof}

\noindent\textbf{Proposition \ref{prop:complexity}.}\emph{
Let   $d_G$ be the diameter of the graph $G$ and $\mathcal{X} (\ell)$ the fixed point of $\mathcal{X}(i,\ell)$, then $\mathcal{X} (\ell)=\mathcal{X} (d_G+1, \ell)$ for all $\ell \in L$.
}

\begin{proof}
The graph diameter of G is equal to $d_{g}=\max_{\ell,\ell' \in L}d(\ell,\ell')$. Recall that $\mathcal{X} ( d_{g}, \ell) \in \{s_{j}(\ell) \:\big{|}\:  j \in \{ 1,2 \}, \ell \in L \}$ is a finite set of sortable values. 
 At step zero the value of $\mathcal{X}$ is equal to $s_{1}$ in all the locations.
At each next step, the value of $\mathcal{X}(i,\ell)$ depends only on the value of $\mathcal{X}$ in the same location at the previous step and the value of $s_{2}$ and $\mathcal{X}$ in the previous step in the direct neighbours of $\ell$, $\ell' E \ell$.
This means that, after a number of steps equal to the diameter of the graph, i.e., the longest shortest path of the network, $\mathcal{X}$, for all nodes $\ell$, has taken into account the values $s_{1}$ and $s_{2}$ of all the nodes. 
%
%
%
%
%
\end{proof}
